\documentclass[english]{article}
\usepackage[T1]{fontenc}
\usepackage[utf8]{inputenc}
\usepackage{geometry}
\usepackage{xcolor}
\usepackage{setspace}

\usepackage{float}
\usepackage{tabularx}
\usepackage{mathrsfs}
\usepackage{mathtools}
\usepackage{bm}
\usepackage{amsmath}
\usepackage{amsthm}
\usepackage{amssymb}
\usepackage{graphicx}
\usepackage{esint}
\usepackage{tabularx}
\providecommand{\lemmaname}{Lemma}
\providecommand{\theoremname}{Theorem}
\providecommand{\propositionname}{Proposition}
\providecommand{\definitionname}{Definition}

\newtheorem{theorem}{Theorem}

\DeclareMathOperator{\sign}{sign} 

\newcommand{\rem}[1]{}
\newtheorem{thm}{\protect\theoremname}
\newtheorem{prop}[thm]{\protect\propositionname}
\newtheorem{lem}[thm]{\protect\lemmaname}

\newcommand\summaryname{Abstract}
    {\normal\begin{center}%
    \bfseries{\summaryname} \end{center}}

\usepackage[unicode=true,pdfusetitle,
 bookmarks=true,bookmarksnumbered=false,bookmarksopen=false,
 breaklinks=true,pdfborder={0 0 1},backref=page,colorlinks=true,hyperindex]
 {hyperref}
\hypersetup{linkcolor=blue,citecolor=magenta}



\usepackage{babel}

\begin{document}

\title{Quantum Integrable Systems arising from Separation of Variables on
$S^{3}$}
\author{Sean R.~Dawson,  Holger R.~Dullin\footnote{Emails: \url{sdaw6022@uni.sydney.edu.au}, \url{holger.dullin@sydney.edu.au.}}\\School of Mathematics and Statistics,\\The University of Sydney, Australia}

\maketitle
\begin{abstract}
    We study the family of quantum integrable systems that arise from separating the Schrödinger equation in all $6$ separable orthogonal coordinates on the $3-$sphere: ellipsoidal, prolate, oblate, Lam\'{e}, spherical and cylindrical. On the one hand each separating coordinate system gives rise to a quantum integrable system on $S^2 \times S^2$, on the other hand it also leads to families of harmonic polynomials in $\mathbb{R}^4$. We show that separation in ellipsoidal coordinates yields a generalised Lam\'{e} equation - a Fuchsian ODE with 5 regular singular points. We seek poly\-nomial solutions so that the eigenfunctions are analytic at all finite singularities. We classify eigenfunctions by their discrete symmetry and compute the joint spectrum for each symmetry class. The latter 5 separable coordinate systems are all degenerations of the ellipsoidal coordinates. We perform similar analyses on these systems and show how the ODEs degenerate in a fashion akin to their respective coordinates. For the prolate system we show that there exists a defect in the joint spectrum which prohibits a global assignment of quantum numbers: the system has quantum monodromy. This is a companion paper to \cite{Nguyen2023} where the respective classical systems were studied.
\end{abstract}
\section{Introduction}
The study of separation of variables yields a vast array of classical and quantum integrable systems that are ripe for exploration. 
Separation of variables originated with Jacobi and St\"ackel \cite{Staeckel93}.
On the quantum side, seminal work was done in the early $20$th century by Robertson \cite{robertson1927} who showed, for a given orthogonal coordinate system, that the Schr\"{o}dinger equation is separable if both the classical Hamilton-Jacobi equation is separable and the Robertson condition is satisfied. Work done soon after by Eisenhart \cite{eisenhart1934} showed that all systems that arise from separation of variables are St\"{a}ckel systems and so the corres\-ponding classical integrals obtained will be quadratic in the momenta. 
Work by Kalnins and Miller in the late $20$th century yielded a complete classification of all orthogonal separable coordinates on both $\mathbb{R}^{n}$ and $S^n$ \cite{Kalnins1986}. This work was further summarised in \cite{Kress06}, where connections to superintegrability were also drawn. 
Separation for systems with potential are also well studied. In \cite{Kalnins2002} invariants for a system that admit separation of variables were constructed. 
Sch\"{o}bel and Veselov  gave a topology to the space of separable coordinates on the sphere \cite{Schoebel2015, Schoebel2016} and identified it with Stasheff polytopes. Previously, our motivation was to study the corresponding integrable systems, which inherit the same topology.
In our paper \cite{Nguyen2023}, we showed that the space of orthogonally separable coordinates on $S^3$ induces a family of classical integrable systems on $S^2 \times S^2$ after reduction. Here we seek to broaden this analysis to the space of quantum integrable systems, their corresponding separated ODEs and special functions. In particular, we study 5 families of Fuchsian equations: the generalised Lam\'{e}, Heun, Jacobi, Gegenbauer and Legendre  differential equations. Similar to our previous work, we aim to understand how degenerations in the coordinate systems descend to the space of integrable systems, and correspondingly ODEs and special functions.

The general idea to start with a super-integrable system and exploit its multiseparability to define interesting Liouville integrable systems by reduction has been exploited in the case of the Kepler problem \cite{DW18} and the harmonic oscillator \cite{Chiscop_2019}. This paper follows the same approach for the geodesic flow on the three dimensional sphere $S^3$. In this case the space of separable coordinate systems is a Stasheff polytope which is the pentagon shown in Fig.~\ref{fig:StasheffQuantum}. This is the parameter space of the family of quantum integrable systems on $S^2 \times S^2$. While the previous examples (Kepler and Harmonic oscillator) also lead to reduced systems on compact symplectic manifolds and hence quantised systems on finite dimensional Hilbert spaces, starting with $S^3$ separation of variables in the Schr\"odinger equation leads to Fuchsian equations, while the presence of the potential in the earlier examples lead to confluent Fuchsian equations.
Degeneration of parameters in the general separable coordinate system on $S^3$ to the so-called prolate case (the bottom edge of the pentagon in Fig.~\ref{fig:StasheffQuantum}) then leads to a semi-toric family whose quantum mechanics is described by polynomial solutions of the Heun equation. The semi-toric systems on $S^2\times S^2$ studied in \cite{SadovskiiZhilinskii99,FlochPelayo18,ADH19} do not appear to include the prolate system studied here, because the non-trivial integral is a homogeneous polynomial. In fact, all our Liouville integrable systems have integrals that are in general homogeneous quadratics in the momenta, and may degenerate into squares of linear functions of the momenta. Because of this all our integrable systems are related to spherical harmonics on $S^3$. The ellipsoidal reduced system is related to the Manakov top and has been studied classically and quantum mechanically in \cite{Sinitsyn2007}.

\begin{figure}[tbh]
\begin{centering}
\includegraphics[width=11cm]{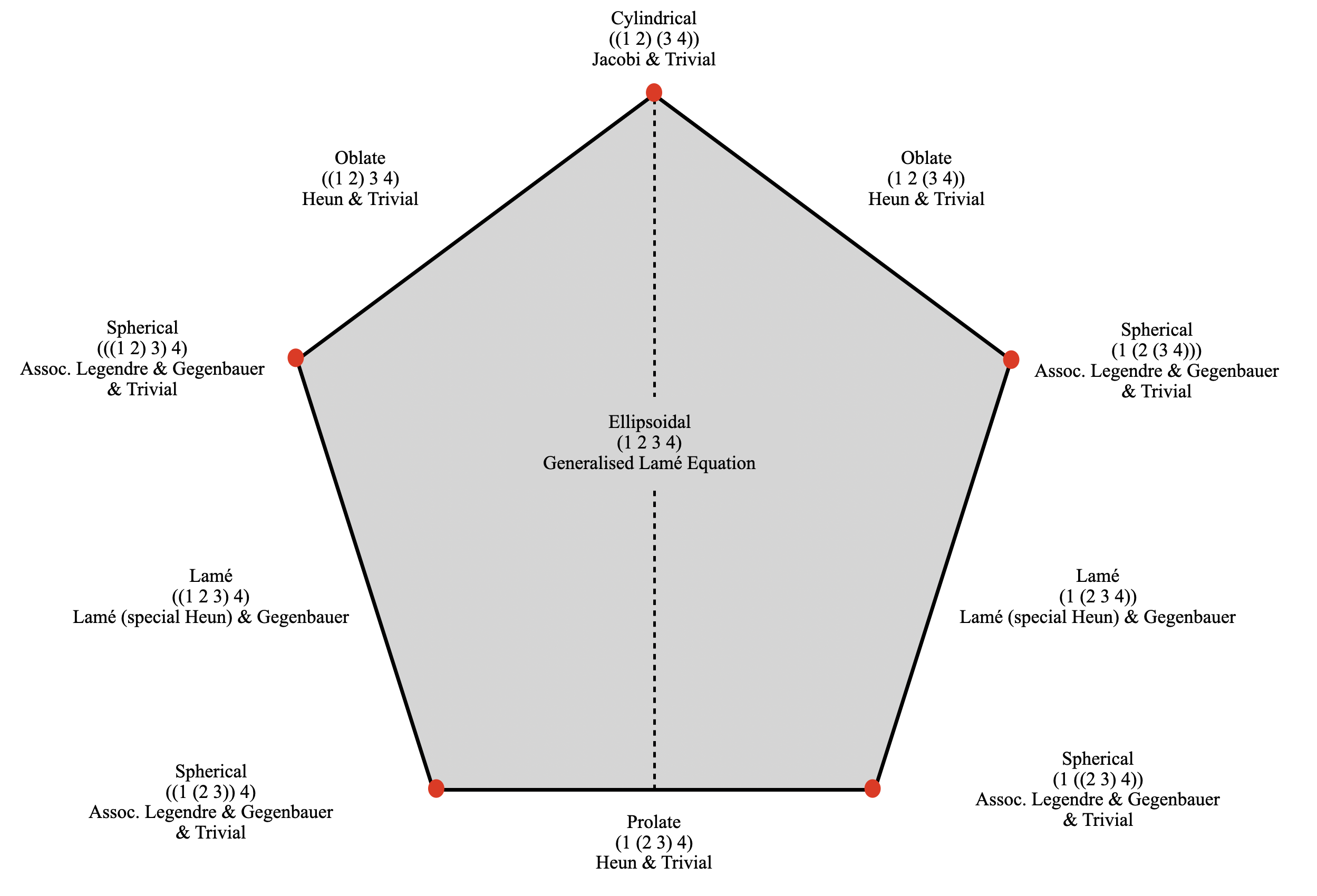}
\par\end{centering}
\caption{Stasheff polytope of orthogonal separable coordinates on $S^3$ with corresponding ODEs arising from separation. \label{fig:StasheffQuantum}}
\end{figure}


In our study of these quantum systems, we perform a variety of numerical computations to find the joint spectrum. For all degenerate systems, a power series ansatz yields a three term recurrence relation. Given the necessary quantisation conditions, this recurrence terminates and yields a polynomial as an eigenfunction. The joint spectrum is then computed from the associated tri-diagonal matrix. The ellipsoidal coordinates on $S^3$ gives the more difficult generalised Lam\'{e} equation which has $5$ regular singularities. Consequently, the aforementi\-oned methodology falls short and we use the method of computing Heine-Stieltjes polynomials as described in \cite{Alam1979, Al-Rashed1985, Volkmer1999}.
It is interesting to note that for all systems on $S^3$, the joint spectrum is recovered from finite matrices, with corresponding polynomial eigenfunction. This contrasts with our previous work \cite{Dawson2022, spheroidal} where the spectrum is determined by infinite matrices and the eigenfunctions are given by infinite series. 

Another motivation for this work is to suggest a natural family of integrable systems that contains the semi-toric case, but includes deformation to more typical quantum integrable systems. The classification of semi-toric systems is now fully understood classically \cite{Pelayo2009} and quantum mechanically \cite{FlochVuNgoc21}. This line of research was initiated by Duistermaat \cite{Duistermaat80} and the first example of quantum monodromy is due to Cushman and Duistermaat \cite{CushDuist88}. Monodromy, both classical and quantum, is caused by the existence of a non-degenerate focus-focus critical value in the classical system \cite{VuNgoc99}. 
Many important examples of integrable systems in physics do have monodromy, see, e.g.~\cite{CushDuist88,SadovskiiZhilinskii99,spheroidal,DW18,Dawson2022}.
In this paper, we obtain a semi-toric system through separation of the geodesic flow on $S^3$ in prolate coordinates. This system thus arises as a degeneration of a more general 2 parameter family corresponding to the inside of the pentagon in Fig. \ref{fig:StasheffQuantum}.

The outline of this paper is as follows. In Section \ref{sec:S3-1} we describe the ellipsoidal-spherical coordinates, the most general orthogonal separable coordinates on $S^3$. The ODEs that arise from separating the Schr\"{o}dinger equation in these coordinates are known as the generalised Lam\'{e} equation. The joint spectrum is computed, along with the corresponding quantised actions and their symmetry classes. The first degenerate system is analysed in Section \ref{sec:ProlateS3} with the prolate coordinates. Separation yields a semi-toric system whose monodromy matrix we compute. This contrasts with the oblate system in Section \ref{sec:OblateS3} which, while similar to prolate coordinates, does not possess monodromy and instead has hyperbolic singularities. Section \ref{sec:LameS3} studies the Lam\'{e} coordinates (an extension of ellipsoidal coordinates on $S^2$ onto $S^3$) while Sections \ref{sec:SphericalS3} and \ref{sec:CylindricalS3} focus on the spherical and cylindrical systems respectively. All these results are visualised by the Stasheff polytope in Fig.~\ref{fig:StasheffQuantum}. In the appendix we recall some results on $S^2$ for the only two separable coordinates: ellipsoidal and spherical.

\section{Ellipsoidal Coordinates and the Generalised Lam\'{e}
Equation}\label{sec:S3-1}

Let $\bm{x}\coloneqq(x_{1},x_{2},x_{3},x_{4})$
be Cartesian coordinates on $\mathbb{R}^{4}$ and consider the unit
sphere $S^{3}\subset\mathbb{R}^{4}$ where $S^{3}=\left\{ (x_{1},x_{2},x_{3},x_{4})\in\mathbb{R}^{4}\,|\,x_{1}^{2}+x_{2}^{2}+x_{3}^{2}+x_{4}^{2}=1\right\} $.
The geodesic flow on $S^{3}$ is a constrained system on $T^{*}\mathbb{R}^{4}$
with corresponding momentum $\bm{y}=(y_{1},y_{2},y_{3},y_{4})$ such
that $\bm{x}\cdot\bm{x}=1$ and $\bm{x}\cdot\bm{y}=0$. Set
\[
H=\frac{1}{2}\sum_{i>j}\ell_{ij}^{2}
\]
 to be the Hamiltonian of the geodesic flow where the angular momentum
in the $(i,j)$ plane is given by $\ell_{ij}\coloneqq x_{i}y_{j}-x_{j}y_{i}$ and the symplectic structure is $\sum_i dx_i \wedge dy_i$.

Let a homogeneous polynomial $\Psi_D(x_1, x_2, \dots, x_n)$ of degree $D$ be a solution to Laplace's equation in $\mathbb{R}^n$,
\[
   \Delta \Psi_D = \sum \partial_i^2 \Psi_D = 0\,.
\]   
By definition $\Psi_D$ is called a \emph{harmonic} polynomial.
Introduce the radius $r = \sqrt{\sum x_i^2}$ and some local coordinates $s$ on the sphere $S^{n-1}$. The Laplacian in such a coordinate system is
\[
   \Delta_{\mathbb{R}^n} = \frac{\partial^2}{\partial r^2} + \frac{n-1}{r} \frac{\partial}{\partial r} + \frac{1}{r^2} \Delta_{S^{n-1}}\,.
\]
Now rewrite the harmonic polynomial as $\Psi_D(\bm{x}) = r^D \psi_D(s)$ and thus define the spherical harmonic $\psi_D : S^{n-1} \to \mathbb{R}$. Inserting this into the above Laplacian gives
\[
      -\Delta_{S^{n-1}} \psi_D = D(D-2+n) \psi_D\,,
\]
so we see that spherical harmonics are eigenfunctions of the Laplace-Beltrami operator on the sphere. For $n=3$ on $S^2$ these are the usual spherical harmonics typically denoted by $Y_{lm}$,
for the case of general $n$ see, e.g.,~\cite{Axler01}.
The above Hamiltonian can be canonically quantised which leads to $\hat\ell_{ij} = -i \hbar (x_i \partial_j - x_j \partial_i)$. This again gives  the Laplace-Beltrami operator on $S^{n-1}$ and hence 
\[
    2 \hat H \Psi_D = D( D - 2 + n) \Psi_D\,.
\]
In the following we will speak of both the spherical harmonic $\psi_D$ in local coordinates $s$ on the sphere $S^{n-1}$, and the corresponding harmonic polynomial $\Psi_D$ in global Cartesian coordinates $\bm{x}$ on $\mathbb{R}^n$, as eigenfunctions of the Laplace-Beltrami operator. These considerations apply to any dimension $n$, but in the following we will consider $n=4$ only.

Ellipsoidal-spherical (also known as sphero-conal or spherical ellipsoidal) coordinates $\bm{s}\coloneqq(s_{1},s_{2},s_{3})$
on $S^{3}$ are defined as the roots of 
\begin{align*}
T(s) & =\sum_{i=1}^{4}\frac{x_{i}^{2}}{s-e_{i}}
\end{align*}
 where $0 \le e_{j}\le s_{j}\le e_{j+1}$ for all $j=1,2,3$ and the
$e_{i}$ are called the semi-major axes. 
This is by analogy with ellipsoidal coordinates on the tri-axial ellipsoid which are different and defined through $T(s) = 1$. We will only use ellipsoidal-spherical coordinates in the following and call them ellipsoidal for short since no confusion can arise. The transformation to ellipsoidal coordinates is found by solving $T(s_{j})=0$ where $j=1,2,3$ together with $r^2 = \sum_{i=1}^{4}x_{i}^{2}=1$ to give 
\begin{align}
x_{i}^{2} & =\frac{\prod_{j=1}^{3}(s_{j}-e_{i})}{\prod_{k\ne i}(e_{k}-e_{i})}.\label{eq:ellipsoidal def}
\end{align}

From \cite{Nguyen2023} we have the following result.
\begin{lem}
The Hamiltonian $H$ for the geodesic flow on $S^{3}$
separates in ellipsoidal coordinates (\ref{eq:ellipsoidal def}) to
give integrals 
\begin{align}
\eta_{1}= & \sum_{i<j}\left(\ell_{ij}^{2}\sum_{k\ne i,j}e_{k}\right), & \eta_{2}=\sum_{i<j}\left(\ell_{ij}^{2}\prod_{k\ne i,j}e_{k}\right).\label{eq:EllipsoidalIntegralDef}
\end{align}
The triple $(H,\eta_{2},\eta_{2})$ is a Liouville
integrable system on $T^{*}S^{3}$. 
\end{lem}
Consider the following linear coordinate transformation
\begin{align*}
\begin{aligned}X_{1} & =\frac{1}{\sqrt{2}}(\ell_{12}+\ell_{34}) &  &  & Y_{1} & =\frac{1}{\sqrt{2}}(\ell_{12}-\ell_{34})\\
X_{2} & =\frac{1}{\sqrt{2}}(\ell_{13}-\ell_{24}) &  &  & Y_{2} & =-\frac{1}{\sqrt{2}}(\ell_{13}+\ell_{24})\\
X_{3} & =\frac{1}{\sqrt{2}}(\ell_{14}+\ell_{23}) &  &  & Y_{3} & =\frac{1}{\sqrt{2}}(\ell_{14}-\ell_{23}).
\end{aligned}
\end{align*}
Since $\sqrt{H}$ generates
an $S^{1}$ flow with constant period, we are able to reduce by the
flow of $\sqrt{H}$ to $S^{2}\times S^{2}$ with local coordinates
$(\bm{X},\bm{Y})=(X_{1},X_{2},X_{3},Y_{1},Y_{2},Y_{3})$ and Casimirs
$\bm{X}\cdot\bm{X}=\bm{Y}\cdot\bm{Y}=1$. Again, from \cite{Nguyen2023}
we have the following result.
\begin{lem}
The original integrable system $(H,\eta_{1},\eta_{2})$
on $T^{*}S^{3}$ descends to a two degree of freedom integrable system
$(\eta_{1}(\bm{X},\bm{Y}),\eta_{2}(\bm{X},\bm{Y}))$ on $S^{2}\times S^{2}$
with Poisson structure 
\begin{align}
B_{\bm{X},\bm{Y}} & =\begin{pmatrix}\hat{\bm{X}} & \bm{0}\\
\bm{0} & \hat{\bm{Y}}
\end{pmatrix}\label{eq:BXY}
\end{align}
 where, for a vector $\bm{v}\in\mathbb{R}^{3}$, the corresponding
antisymmetric hat matrix $\hat{\bm{v}}$ is given by
\begin{align}
\hat{\bm{v}}\bm{u} & =\bm{v}\times\bm{u}, &  & \forall\bm{u}\in\mathbb{R}^{3}.
\end{align}
\end{lem}
When written in terms of $\bm{X}$ and $\bm{Y}$ on $S^2 \times S^2$ the system can be quantised through Berezin-Toeplitz quantisation in the same way it was done in \cite{FlochPelayo18,Floch23}. The quantisation is the same, but the operators are different. In their case the global $S^1$ integral is linear and the commuting second integral has linear and quadratic terms, while the integrals $\eta_1$, $\eta_2$ are both homogeneous quadratic. In later chapters on degenerate coordinate systems (e.g.~the prolate case) linear integrals will also appear, in particular the global $S^1$ actions are linear in $\bm{X},\bm{Y}$. We will not use Berezin-Toeplitz quantisation, but instead employ separation of variables, so that the connection to spherical harmonics and special functions remains.

Quantising through separation of variables reduces the quantisation problem to the problem of finding polynomial solutions of a single 2nd order Fuchsian ODE. It would be interesting to study in detail the relationship between the two approaches. The one thing that we do use from the Berezin-Toeplitz approach is the fact that $\hbar$ is the inverse of an integer. This also appears naturally through the separation of  variables, but would be harder to justify without reference to the Berezin-Toeplitz approach, because separation of variables does not clearly execute the reduction to a system with compact phase space.

To construct a quantised version of the Hamiltonian we derive the diagonal metric tensor $g$ in ellipsoidal coordinates on the $3-$sphere given by (\ref{eq:ellipsoidal def}). This is an orthogonal coordinate systems, so the off-diagonal entries of the metric vanish. The diagonal entries of the metric are
\[
g_{jj}=\frac{-(s_{j}-s_{k})(s_{j}-s_{m})}{4(s_{j}-e_{1})(s_{j}-e_{2})(s_{j}-e_{3})(s_{j}-e_{4})}
\]
where $k$ and $m$ are the two unique distinct indices different from $j$. Denote the determinant of the metric by $g = \prod g_{jj}$. The inverse metric with upper indices has entries $g^{jj} = 1/g_{jj}$.
Thus the Laplace-Beltrami operator on $S^3$ in these coordinates is given by 
\[
\Delta = - \frac{1}{\sqrt{g}} \sum \partial_i ( \sqrt{g}  \, g^{jj} \partial_i)\,.
\]
When quantising with the semi-classical parameter $\hbar$ the corresponding stationary Schr\"odinger equation is given by 
\begin{align}\label{eq:Schrodinger}
-\hbar^{2}\Delta\psi & =\tilde{E}\psi \,.
\end{align}
As already remarked, it is possible to quantise in the original Euclidean coordinates and replace $\ell_{ij}$ with $x_i \partial_j - x_j \partial i$, see \cite{gurarie95,toth94}. We need separation of variables, however, and hence obtain commuting operators via the St\"ackel matrix of the system.
The classical commuting integrals \eqref{eq:EllipsoidalIntegralDef} are obtained from separation of variables in the Hamilton-Jacobi equation. 
The same separation works for the Schr\"odinger equation because the sphere has constant curvature, and hence the Robertson condition is satisfied, see, e.g. \cite{KKM18}. 

The reduction to a two degree of freedom integrable system is achieved on the quantum level by fixing the eigenvalue of the Laplace-Beltrami operator, which for $S^3$ is $E = D(D+2)$ for non-negative integer $D$.

\begin{theorem}
There are three commuting 2nd order differential operators $(\hat H, \hat \eta_1, \hat \eta_2)$ on $S^3$ with rational coefficients, whose classical limit are the integrals $(H, \eta_1, \eta_2)$ given in \eqref{eq:EllipsoidalIntegralDef}.
\end{theorem} 
\begin{proof}
 A St\"{a}ckel matrix for ellipsoidal
coordinates is given by 
\begin{equation}
\sigma_{\text{el}}=\frac{1}{4}\begin{pmatrix}-\frac{s_{1}^{2}}{A(s_{1})} & -\frac{s_{1}}{A(s_{1})} & -\frac{1}{A(s_{1})}\\
-\frac{s_{2}^{2}}{A(s_{2})} & -\frac{s_{2}}{A(s_{2})} & -\frac{1}{A(s_{2})}\\
-\frac{s_{3}^{2}}{A(s_{3})} & -\frac{s_{3}}{A(s_{3})} & -\frac{1}{A(s_{3})}
\end{pmatrix}.\label{eq:Stackel El S3-1}
\end{equation}
where $A(z)=\prod_{k=1}^{4}(z-e_{k})$. Note that the metric and St\"{a}ckel
matrix (\ref{eq:Stackel El S3-1}) are related by
\[
\frac{1}{g_{jj}}=\frac{\det(\Omega_{j})}{\det(\sigma_{\text{el}})}
\]
where $\Omega_{j}$ is the minor formed by deleting the $j^{th}$
row and first column of $\sigma_{\text{el}}$.

Denote the rows of the inverse of the St\"ackel matrix by $r^1, r^2, r^3$. By definition of the St\"ackel matrix the first row contains the diagonal entries of the metric $g^{jj} = r^{1j}$ with upper indices. 
Then the commuting operators are
\[
     \hat \eta_{k-1} = -\frac{1}{\sqrt{g}} \sum_i \partial_i ( \sqrt{g} \, r^{ki} \partial_i ), \quad k= 1,2,3.
\]
where $\hat H = \hat \eta_0$ corresponding to the first row of the inverse of the St\"ackel matrix is the Laplace-Beltrami operator. The corresponding classical integrals are $\eta_{k-1} = \sum r^{ki} p_i^2$ as defined in \eqref{eq:EllipsoidalIntegralDef} in Cartesian coordinates. Here $p_i$ are the canonical curvilinear momenta conjugate to the coordinates $s_i$.
We denote the eigenvalues of $(\hat H,\hat \eta_1,\hat \eta_2)$ by $(E,\lambda_1,\lambda_2)$. 
\end{proof}

 The system of ellipsoidal coordinates on the sphere is used in the separation of the quantised C.~Neumann system, 
 see e.g. \cite{gurarie95, GurarieBook} and also \cite{Bellon:2005ae, BELLON2005360, BELLON2006110}. Our system before reduction is the Neumann system without potential, and hence  formulas in \cite{gurarie95} can be specialised to our case. The reduction step (or fixing the eigenvalue of the Schr\"odinger operator) does not make sense for the Neumann system, because unlike the geodesic flow on $S^3$ it is not a superintegrable system.
 
It is also possible to obtain the operators $\hat \eta_i$ by canonical quantisation of the angular momentum operators expressed in ellipsoidal coordinates using
\[
   \ell_{ij} = 2 x_i(s) x_j(s)( e_i - e_j) \sum_{k=1}^3 c_{ij}(s_k) p_k, \quad c_{ij}(s_k) = \frac{\prod_{n \not = i,j} (s_k - e_n)}{\prod_{m \not = k} (s_k - s_m) }
\]
where $x_i(s)$ as given in \eqref{eq:ellipsoidal def}.
This equation is linear in momenta and can be quantised by replacing $p_k \to -i \hbar \partial_{s_k}$. In this way the classical commuting integrals expressed in terms of $\ell_{ij}$ can be directly converted into the corresponding quantum operators. However, the calculation using the St\"ackel matrix is much simpler.

The inverse of the St\"ackel matrix is used to define the commuting operators, while the St\"ackel matrix itself determines how the spectral parameters enter the separated equation(s), which is the content of
\begin{lem}
\label{SeparationLemma}
In ellipsoidal coordinates
(\ref{eq:ellipsoidal def}) on $S^{3}$, the Schrödinger equation
(\ref{eq:Schrodinger}) separates into 
\begin{equation}
\psi_{j}^{''}+\frac{1}{2}\left(\frac{1}{s_{j}-e_{1}}+\frac{1}{s_{j}-e_{2}}+\frac{1}{s_{j}-e_{3}}+\frac{1}{s_{j}-e_{4}}\right)\psi_{j}^{'}+\frac{-Es_{j}^{2}+\lambda_{1}s_{j}-\lambda_{2}}{4(s_{j}-e_{1})(s_{j}-e_{2})(s_{j}-e_{3})(s_{j}-e_{4})}\psi_{j}=0\label{eq:Gen Lame}
\end{equation}
for all $j=1,2,3$ with $E\coloneqq\frac{\tilde{E}}{\hbar^{2}}$,
$\lambda_{1}$ and $\lambda_{2}$ are spectral parameters and $\psi=\psi_{1}\psi_{2}\psi_{3}$
where $\psi_{j}=\psi(s_{j})$. 
\end{lem}
\begin{proof}
Let 
\begin{align*}
f_{j}\coloneqq & \frac{\partial}{\partial s_{j}}\log(g^{jj}\sqrt{\det(g)}).
\end{align*}
Following  \cite{gurarie95} the equation (\ref{eq:Schrodinger}) separates into 
\begin{align*}
\psi_{j}^{''}+f_{j}\psi_{j}^{'}+\sum_{k=1}^{3}c_{k-1}(\sigma_{\text{el}})_{jk}\psi_{j}=0
\end{align*}
where $(c_{0},c_{1},c_{2})=(-E,\lambda_{1},-\lambda_{2})$
are spectral parameters. Computing the relevant quantities gives (\ref{eq:Gen Lame}).
\end{proof}

Equation (\ref{eq:Gen Lame}) is known as the generalised Lam\'{e} equation. This is a second order Fuchsian equation with $4$ finite regular singular points which can be writen as
\begin{align}
\frac{d^{2}w}{dz^{2}}+\left(\sum_{j=1}^{4}\frac{\gamma_{j}}{z-e_{j}}\right)\frac{dw}{dz}+\left(\sum_{j=1}^{4}\frac{q_{j}}{z-e_{j}}\right)w & =0\label{eq:ArbGenLame}
\end{align}
where $\sum_{j=1}^{4}q_{j}=0$ and the $q_{j}$ are known as the accessory
parameters. 
Often, we write the multiplicative term of \eqref{eq:ArbGenLame} over a common denominator, i.e. \begin{equation}
    \sum_{j=1}^4\frac{q_j}{z-e_j} = \frac{c_0z^2 + c_1z + c_2}{4\Pi_{j=1}^4(z-e_j)}.
\end{equation}
The exponents at the pole $e_{j}$ are
$(0,1-\gamma_{j})$ while those at infinity are $(\alpha,\beta)$ defined by
\begin{equation}
\begin{aligned}
    \alpha+\beta+1 & =\sum_{j=1}^{4}\gamma_{j}, & 
    \alpha\beta & =\sum_{j=1}^{4}e_{j}q_{j}.
\end{aligned}
\label{eq:Gen FuschianInfinity}
\end{equation}
From (\ref{eq:Gen FuschianInfinity}) it is clear
that there are two free (so-called) accessory parameters, while the more well
known Lam\'{e} equation only
has one, see \cite{DLMF} 31.14 for more detail. These two free parameters turn out to be the spectral parameters, i.e.~the eigenvalues of $\hat \eta_1$ and $\hat \eta_2$. For some other examples of occurrences of the generalised Lam\'{e} equation, see \cite{Pawellek2007,Mussardo2006, Chen2019} and the references therein.

In this paper, we seek polynomial solutions of \eqref{eq:Schrodinger} giving harmonic polynomials on $S^3$ in Cartesian coordinates. Ellipsoidal coordinates are singular along co-dimension one submani\-folds of $S^3$, and these coordinate singularities translate to the singularities of the Fuchsian ODE \eqref{eq:ArbGenLame}. Solutions to the separated equations that are analytic at all finite singularities correspond to polynomial solutions in the local coordinates $s_i$, possibly after factoring out single powers of $x_k(s)$ to achieve all possible discrete symmetries in the original coordinates. Thus polynomial solutions of the Fuchsian ODEs in local coordinates lead to polynomial solutions in the original cartesian coordinates. The quantisation conditions at the level of the Fuchsian equation \emph{are} the conditions that enforce a polynomial separated eigenfunction $\psi(s)$. In general these are conditions for each factor $\psi_j(s_j)$ of the separated eigenfunction $\psi(s) = \prod \psi_j(s_j)$, but it turns out that in the present highly symmetric case the three equations are the same equation, just evaluated on different intervals $e_j \le s_j \le e_{j+1}$, and hence a single polynomial gives three $\psi_j(s_j)$ and hence the complete eigenfunction $\psi(s)$. In cases of more degenerate coordinate systems discussed in the next section there will actually be different equations for the different coordinates.

It is known from 31.15 of \cite{DLMF} and \cite{Volkmer1999} that for every
vector $\bm{n}=(n_{1},n_{2},n_{3})$ of non negative integers, there
are uniquely determined real values of $(c_{0},c_{1},c_{2})$ such
that (\ref{eq:ArbGenLame}) has polynomial solutions with $n_{j}$ zeros in the interval $(e_{j},e_{j+1})$ for all $j=1,2,3$.
Such solutions $S_d(z)$ of degree $d$ where $d = n_1 + n_2 + n_3$ are known as Heine-Stieltjes polynomials. 

From \cite{Volkmer1999} we have the following Lemma
for the spectral parameter $c_{0}$.
\begin{lem}
\label{Lemma3}Let $S_d(z)$ be a Heine-Stieltjes polynomial solution to (\ref{eq:ArbGenLame}) of degree $d$ over the interval $[e_{1},e_{4}]$.
Then 
\begin{align}
\frac{1}{4}c_{0} & =-d(d-1+\sum_{i=1}^{4}\gamma_{i}) = \alpha\beta \label{eq:c2 expression}
\end{align}
and $(\alpha,\beta)=(-d,(-1+d+\sum\gamma_{i}))$. Further, let $K(d)$ be the space of
all such polynomials that satisfy (\ref{eq:ArbGenLame}).
Then, we have
\begin{align}
\text{dim}(K(d)) & \coloneqq\begin{pmatrix}d+2\\
2
\end{pmatrix}.\label{eq:dim d}
\end{align}
\end{lem}
Note that $\alpha$ is a negative integer in order to obtain polynomial solutions. We also have from \cite{Volkmer1999} the following
Lemma relating polynomial solutions of the Schrödinger equation in
Cartesian coordinates to products of Heine-Stieltjes polynomials.
\begin{prop}
The product $S_d(s_1)S_d(s_2)S_d(s_3)$, where  $S_d(s)$ is a Heine-Stieltjes polynomial solution to \eqref{eq:Gen Lame} of degree $d$, when expressed in the original Cartesian coordinates using (\ref{eq:ellipsoidal def}), is a homogeneous polynomial of degree $\tilde{D}\coloneqq2d$ in the variables $x_{i}$ given by 
\begin{align}
\Phi_{\tilde{D}}(\bm{x})\coloneqq\left(\prod_{i=1}^{4}\prod_{k=1}^{d}(z_{k}-e_{i})\right)\prod_{k=1}^{d}\sum_{i=1}^{4}\frac{x_{i}^{2}}{z_{k}-e_{i}}\,,\label{eq:PsidEllipsoidal}
\end{align}
 where $z_{k}$ are the roots of $S_d(z)$.
\end{prop}

In Cartesian coordinates there are 16 discrete symmetry classes of eigenfunctions 
corres\-pon\-ding to parity
about each of the $x_{i}$ axes. Let $\bm{\mu} = (\mu_1, \mu_2,\mu_3,\mu_4)$ with $\mu_i\in \{0,1\}$ where $\mu_i = 1$ denotes a wavefunction odd about the respective axis and even otherwise. 
By \eqref{eq:PsidEllipsoidal} the wave function $\Phi_{\tilde{D}}(\bm{x})$ is quadratic in all $x_i$, and hence cannot give any odd symmetry. In order to find these other symmetries, consider the change of dependent variable in \eqref{eq:Gen Lame} given by 
\begin{equation}
     \phi_j = \prod_{i}(s_j-e_i)^{\mu_i/2}\psi_{j}.
\end{equation}
This corresponds to multiplication by $x_1^{\mu_1}x_2^{\mu_2}x_3^{\mu_3}x_4^{\mu_4}$ in Cartesian coordinates. This yields the transformed generalised Lam\'{e} equation 
\begin{equation}
\phi_j^{''}+\left(\frac{\tilde{\gamma}_{1}}{z-e_{1}}+\frac{\tilde{\gamma}_{2}}{z-e_{2}}+\frac{\tilde{\gamma}_{3}}{z-e_{3}}+\frac{\tilde{\gamma}_{4}}{z-e_{4}}\right)\phi_j^{'}+\frac{(u_{0} - E)z^{2}+(u_{1} + \lambda_1)z+(u_{2}-\lambda_2)}{4(z-e_{1})(z-e_{2})(z-e_{3})(z-e_{4})}\phi_j=0\label{eq:translame}
\end{equation}
where
\begin{align}
\tilde{\gamma}_{m} & =\begin{cases}
\frac{3}{2} & \text{if \ensuremath{\mu_{m}=1}}\\
\frac{1}{2} & \text{if \ensuremath{\mu_{m}=0}}
\end{cases}\label{eq:gammarule}
\end{align}
and the $u_{j}$ are shown in Table \ref{tab: CoeffTable} for a given symmetry class $\bm{\mu}$. 
Note that in the table we adopted the notation where $(\mu_i)$ represents symmetry classes where only one $\mu_i$ is $1$ (i.e. $(1,0,0,0), (0,1,0,0)$ etc). Similarly, $(\mu_i,\mu_j)$ represents all classes which are odd about two axes and so forth for $(\mu_i,\mu_j,\mu_k)$. 

Denote by $S_d^{\bm{\mu}}(z)$ solutions to \eqref{eq:translame} for the symmetry class $\bm{\mu}$. Repeating the same reasoning used to obtain \eqref{eq:PsidEllipsoidal}, we let the product $S_d^{\bm{\mu}}(s_1)S_d^{\bm{\mu}}(s_2)S_d^{\bm{\mu}}(s_3)$, when converted to Cartesian coordinates, be given by $\Phi_{\tilde{D}}^{\bm{\mu}}(\bm{x})$.
Finally, we set 
\begin{equation}
    \Psi_{D}^{\bm{\mu}}(\bm{x}) \coloneqq x_1^{\mu_1}x_2^{\mu_2}x_3^{\mu_3}x_4^{\mu_4}\Phi_{\tilde{D}}^{\bm{\mu}}(\bm{x})\label{eq:phi solution}
\end{equation}
where $D = \tilde{D} + \sum_{i=1}^4 \mu_i$. Note that $\Phi_{\tilde{D}}(\bm{x})$ in \eqref{eq:PsidEllipsoidal} is the special case of \eqref{eq:phi solution} with $\bm{\mu} = (0,0,0,0)$.
\begin{table}[H]
\begin{centering}
\textcolor{black}{}%
\begin{tabular}{|c|c|c|c|}
\hline 
\textcolor{black}{Symmetry} & \textcolor{black}{$u_{2}$ } & 
\textcolor{black}{$u_{1}$} & \textcolor{black}{$u_{0}$}\tabularnewline
\hline 
\textcolor{black}{$(0,0,0,0)$} & \textcolor{black}{$0$} & \textcolor{black}{$0$} & \textcolor{black}{$0$}\tabularnewline
\hline 
\textcolor{black}{$(\mu_i)$} & \textcolor{black}{$e_{j}e_{k}+e_{j}e_{m}+e_{k}e_{m}$} & \textcolor{black}{$-2(e_{j}+e_{k}+e_{m})$} & \textcolor{black}{$3$}\tabularnewline
\hline
\textcolor{black}{$(\mu_i,\mu_j)$} & \textcolor{black}{$e_{i}e_{k}+e_{i}e_{m}+e_{j}e_{k}+e_{j}e_{m}+4e_{k}e_{m}$} & \textcolor{black}{$-2(e_{i}+e_{j})-6(e_{k}+e_{m})$} & \textcolor{black}{$8$}\tabularnewline
\hline 
\textcolor{black}{$(\mu_i,\mu_j, \mu_k)$} & \textcolor{black}{$e_{i}e_{j}+e_{i}e_{k}+e_{j}e_{k}+4(e_{i}e_{m}+e_{j}e_{m}+e_{k}e_{m})$} & \textcolor{black}{$-6(e_{i}+e_{j}+e_{k}+2e_{m})$} & \textcolor{black}{$15$}\tabularnewline
\hline 
\textcolor{black}{$(1,1,1,1)$} & \textcolor{black}{$4(e_{1}e_{2}+e_{1}e_{3}+e_{2}e_{3}+e_{1}e_{4}+e_{2}e_{4}+e_{3}e_{4})$} & \textcolor{black}{$-12(e_{1}+e_{2}+e_{3}+e_{4})$} & \textcolor{black}{$24$}\tabularnewline
\hline 
\end{tabular}
\par\end{centering}
\caption{Coefficients $u_{j}$ in (\ref{eq:translame}) for each symmetry class.
\label{tab: CoeffTable}}
\end{table}
\begin{theorem}
     The function $\Psi_{D}^{\bm{\mu}}(\bm{x})$ given by \eqref{eq:phi solution} is an eigenfunction of the Schrödinger operator \eqref{eq:Schrodinger}. The
energy  eigenvalue $E$  is given by 
\begin{align}
E & =D(D+2)\label{eq:Ejj2}
\end{align}
where $D= 2d+\sum_{i=1}^4 \mu_i$ is the degree of the harmonic polynomial $\Psi_D$.
\end{theorem}
\begin{proof}
    We have 
    \begin{equation}
    \Delta(\Psi_{D}^{\bm{\mu}}(\bm{x})) = \Delta(\Phi_{\tilde{D}}^{\bm{\mu}}(\bm{x})) + \sum_{j=1}^{4}\frac{2\mu_j}{x_j}\frac{\Phi_{\tilde{D}}^{\bm{\mu}}(\bm{x})}{\partial x_j}.\label{eq: VolkmerEqn}
    \end{equation}
    From \cite{Volkmer1999}, we know that since $\Phi_{\tilde{D}}^{\bm{\mu}}(\bm{x})$ satisfy \eqref{eq:translame} for the appropriate parameter choices in Table \ref{tab: CoeffTable}, they also satisfy the right hand side of \eqref{eq: VolkmerEqn}. Hence, $\Psi_{D}^{\bm{\mu}}(\bm{x})$ are solutions of \eqref{eq:Schrodinger}. To compute the energy for a given symmetry class $\bm{\mu}$, we note that the shift $u_2$ in Table \ref{tab: CoeffTable} is simply $U(U+2)$ where $U=\sum_{i=1}^4\mu_i$. Combining this with (\ref{eq:c2 expression}) and using $E = -c_0$ gives 
    \begin{equation}
        E = U(U+2) + 4d(d+1 + U)
    \end{equation}
    which simplifies to $E=D(D+2)$ for each class. 
\end{proof}

From (\ref{eq:phi solution}) we observe the following
\begin{prop}
For polynomial solutions to (\ref{eq:Schrodinger})
of fixed degree $D$ corresponding to energy $E=D(D+2)$, only 8 of the
16 total discrete symmetry classes can be present. If $E$ is even,
then the $(0,0,0,0)$, $(1,1,1,1)$ and all symmetry classes even
about 2 of the $x_{i}$ axes are present. Similarly, if $E$ is odd,
then the remaining 8 symmetry classes odd about and odd number of $x_i$ are present. 
\end{prop}
We refer to solutions which are even about an even
(odd) number of axes (with even (odd) energies) as ``even'' (``odd'')
solutions. In particular, the even symmetry classes are  
\begin{equation}
    \mathcal{S}_{E} = \{(0,0,0,0),(1,1,0,0),(1,0,1,0),(1,0,0,1),(0,1,1,0),(0,1,0,1),(0,0,1,1),(1,1,1,1)\} \label{eq:SE_Defn}
\end{equation}
and the odd symmetry classes are 
\begin{equation}
    \mathcal{S}_{O} = \{(0,0,0,1),(0,0,1,0),(0,1,0,0),(1,0,0,0),(1,1,1,0),(1,1,0,1),(1,0,1,1),(0,1,1,1)\}. \label{eq:SO_Defn}
\end{equation}


Using the above results, we also have the following Lemma, which shows that all eigenfunctions are obtained in this way.
\begin{lem}
The total number of eigenstates $N$ for a polynomial
solution to (\ref{eq:Schrodinger}) of degree $D$ is given by 
\begin{equation}
N=(D+1)^{2}\label{eq:Nstates}
\end{equation}

\end{lem}
\begin{proof}
Without loss of generality, we prove this for $D$ even (i.e. ``even'' solutions). From (\ref{eq:dim d}),
we know that the total number of eigenstates for the $(0,0,0,0)$
symmetry class is $\begin{pmatrix}d+2\\
2
\end{pmatrix} = \begin{pmatrix}\frac{D+4}{2}\\
2
\end{pmatrix}$. For a fixed $D$, we similarly have the total number of states for an element of the $(\mu_i,\mu_j)$ class as $\begin{pmatrix}\frac{D+2}{2}\\
2\end{pmatrix}$ and for the $(1,1,1,1)$ class we obtain $\begin{pmatrix}\frac{D}{2}\\
2\end{pmatrix}$.
The result follows since
\[
N=\begin{pmatrix}\frac{D+4}{2}\\
2
\end{pmatrix}+6\begin{pmatrix}\frac{D+2}{2}\\
2
\end{pmatrix}+\begin{pmatrix}\frac{D}{2}\\
2
\end{pmatrix}=(D+1)^{2}.
\]
In Table \ref{tab:NumStatesEllipsoidal} we show the number of eigenstates per symmetry class and energy parity.
\end{proof}

\begin{table}[H]
    \centering
    \begin{tabular}{|c|c|c|}
    \hline
        Symmetry & Number of States & Energy Parity\\
    \hline
        $(0,0,0,0)$ & $\begin{pmatrix}(D+4)/{2}\\2\end{pmatrix}$ & Even\\
        \hline 
        $(\mu_i)$ & $\begin{pmatrix}(D+3)/{2}\\2\end{pmatrix}$ & Odd \\
        \hline 
         $(\mu_i, \mu_j)$& $\begin{pmatrix}(D+2)/{2}\\2\end{pmatrix}$ & Even \\
         \hline 
        $(\mu_i,\mu_j,\mu_k)$ & $\begin{pmatrix}(D+1)/{2}\\2\end{pmatrix}$ & Odd\\
        \hline 
        $(1,1,1,1)$  & $\begin{pmatrix}D/{2}\\2\end{pmatrix}$ & Even \\
    \hline
    \end{tabular}
    \caption{Number of states for each symmetry class and energy parity.}
    \label{tab:NumStatesEllipsoidal}
\end{table}
Recall that the basis of our study is separation
of variables on the $3-$sphere. It is no surprise then that the total
energy (\ref{eq:Ejj2}) and number of eigenstates (\ref{eq:Nstates})
precisely coincide with standard results for spherical harmonics on
the sphere, see, e.g.~\cite{Axler01}.
To compute the joint spectrum of $(\lambda_{1},\lambda_{2})$
in (\ref{eq:Gen Lame}) for all 16 discrete symmetry classes, we employ
the following Lemma from \cite{Al-Rashed1985}.
\begin{lem}
Let $S_{d}$ be a Heine-Stieltjes polynomial of degree
$d$ and denote its (real) roots by $z_{1},\dots,z_{d}$. If, in  \eqref{eq:ArbGenLame}, the $\gamma_j> 0$ and $e_j \in \mathbb{R}$ with $e_j < e_{j+1}$ for all $j$, then every root is a solution to the system of equations
\begin{align}
\sum_{j=1}^{4}\frac{\gamma_{j}/2}{z_{k}-e_{j}}+\sum_{j=1,j\ne k}^{d}\frac{1}{z_{k}-z_{j}} & =0, \quad k = 1, \dots, d\,.
\label{eq:NumericalSpecEqns}
\end{align}
For each solution $z_1,\dots, z_d$ the accessory parameters $q_{j}$ are given
by 
\begin{align*}
q_{j} & =\gamma_{j}\sum_{k=1}^{d}\frac{1}{z_{k}-e_{j}}, \quad j = 1,2,3,4\,,
\end{align*}
and in terms of these the spectral parameters $(\lambda_{1},\lambda_{2})$
are 
\begin{subequations}
\begin{align}
\lambda_{1} & = - 4 \sum e_ie_j(q_k + q_m) \label{eq:lambdaEllipsoidal1}\\
\lambda_{2} & = 4 \sum e_i e_j e_k q_m \label{eq:lambdaEllipsoidal2}
\end{align}
\end{subequations}
where the sums are over all $i,j,k,m$
that are pairwise distinct.
\end{lem}
This well known Lemma is surprising in the sense that the system of equations \eqref{eq:NumericalSpecEqns} determines the eigenfunction $S_d$ and joint spectrum $(\lambda_1, \lambda_2)$ simultaneously. 
Comparing (\ref{eq:lambdaEllipsoidal1}), (\ref{eq:lambdaEllipsoidal2}) and the classical equation (\ref{eq:EllipsoidalIntegralDef}),
the degrees in $e_j$ are different. 
The reason for this is that $q_j$ for fixed energy is of degree $-1$ in $e_j$, as can be seen from \eqref{eq:c2 expression}. Note that the solution to the nonlinear equations \eqref{eq:NumericalSpecEqns} is such that automatically $\sum q_i = 0$ and $\sum q_i e_i = \alpha \beta = -d(d - 1 + \sum \gamma_i)$, combining \eqref{eq:c2 expression} and \eqref{eq:Gen FuschianInfinity}.

We can use the value of $\hbar$ to scale the energy eigenvalue to 1 in the semiclassical limit $\hbar \to 0$, and similar scalings for the other eigenvalues. After the scaling, changing $\hbar$ changes the number of states, but not the size of the image of the momentum map. 
Recall the definition of $E=\frac{\tilde{E}}{\hbar^{2}}$
and the discrete values of $E = D(D+2)$ given in (\ref{eq:Ejj2}). To scale such that $\tilde E$ is exactly 1 would require $\hbar = 1/\sqrt{ D ( D+2)}$, which is, however, not an integer. The reduced symplectic manifold is $S^2 \times S^2$ and hence compact, and therefore we require $\hbar$ to be the inverse of an integer. Thus 
\begin{align}
\hbar & =\frac{1}{\sqrt{D(D+2) + 1}} = \frac{1}{D+1}\label{eq:hbar val}
\end{align}
so that $1/\hbar^2$ is exactly the number of states in the reduced compact system.
With this definition of $\hbar$ we find
\[
    \tilde E = E \hbar^2 = \frac{D(D+2)}{(D+1)^2} = 1 - \hbar^2
\]
where the last equality is exact.

Using (\ref{eq:NumericalSpecEqns}) with the 16
forms of the generalised Lam\'{e} equation,
we produce examples of the joint spectrum shown in Fig. \ref{fig:Fig1}
with $D=18$ (a) and $D=19$ (b). We also show that the boundary of the joint spectrum
is given by the classical momentum map (solid black lines). For
their derivation, see \cite{Nguyen2023}. 

In our numerics, we seed roots $z_k$ randomly in the appropriate intervals $(e_j, e_{j+1})$ as outlined in \cite{Volkmer1999}. Standard root finding techniques like Newton's method are then used to find accurate solutions to (\ref{eq:NumericalSpecEqns}). Larger values of $D$ result in substantially longer computation time.

Since we have the freedom to choose the quadratic integrals via an affine transformation, we can give a clearer representation of the
spectrum, as shown in Fig. \ref{fig:Fig1rotated} a) and
b) with the hyperbolic-hyperbolic point centred at the origin. This clearly shows three regions with a $\mathbb{Z}^{2}$
lattice.

Note that lattice points in these regions are actually two
points, representing period doubling. Classically, this corresponds
to two tori in the pre-image of the classical momentum map. Moving
from these chambers into the fourth (bounded by the curve), we see
a halving of visible dots as we see period doubling occurring again. This phenomena has been described in more detail
in \cite{Sinitsyn2007}.
\begin{figure}[H]
\begin{centering}
\textcolor{black}{\includegraphics[width=7cm]{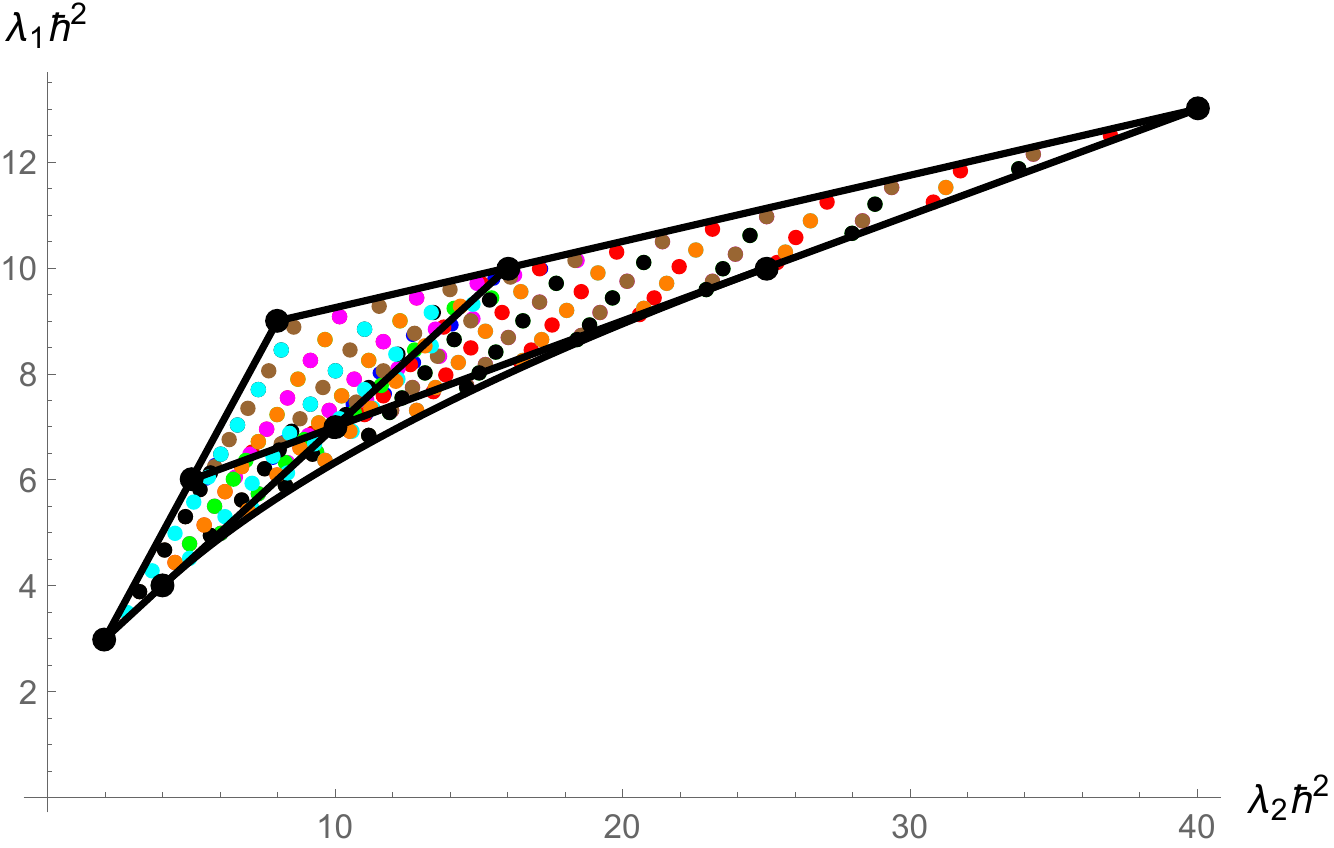}\includegraphics[width=7cm]{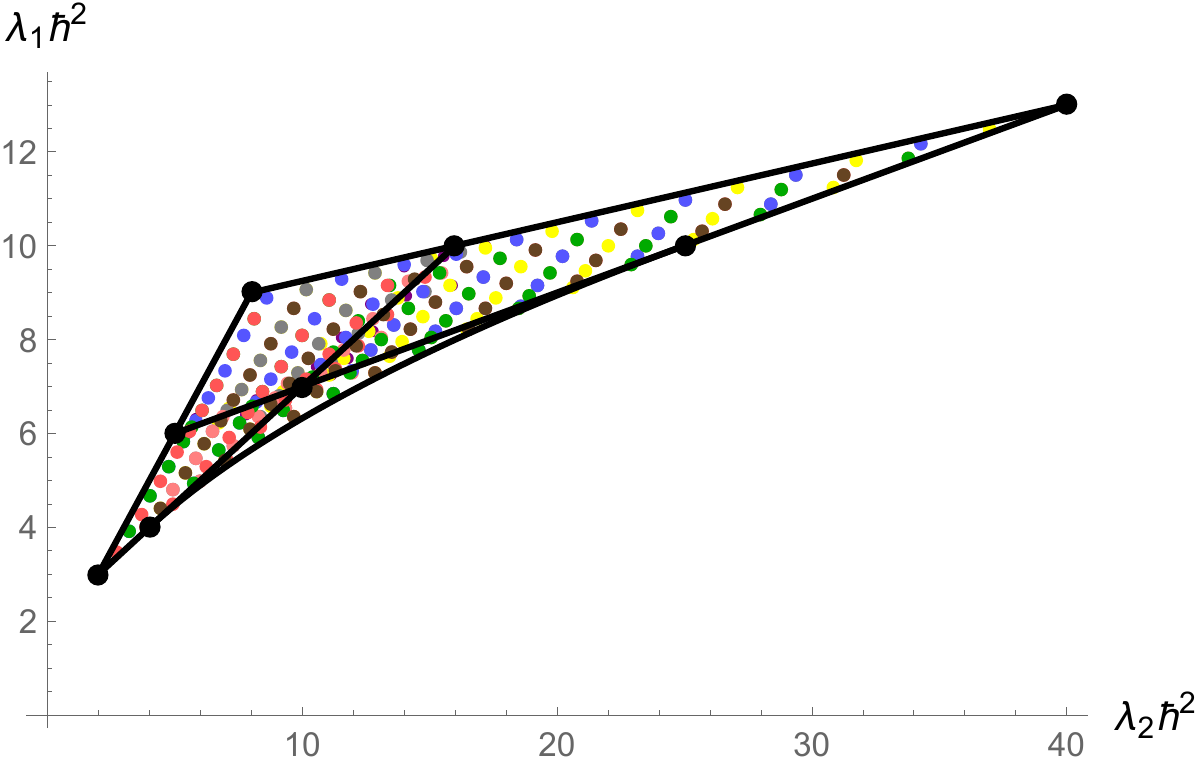}}
\par\end{centering}
\caption{a) Joint spectrum $(\lambda_{1}, \lambda_2)$ where $(e_{1},e_{2},e_{3},e_{4})=(1,2,5,8)$
with $D=18$ and b) $D=19$. A direct correspondence between the coloured dots and symmetry classes is shown in Table \ref{tab:my_label}.\label{fig:Fig1}}
\end{figure}
\begin{figure}[H]
\begin{centering}
\includegraphics[width=7cm,height=7cm]{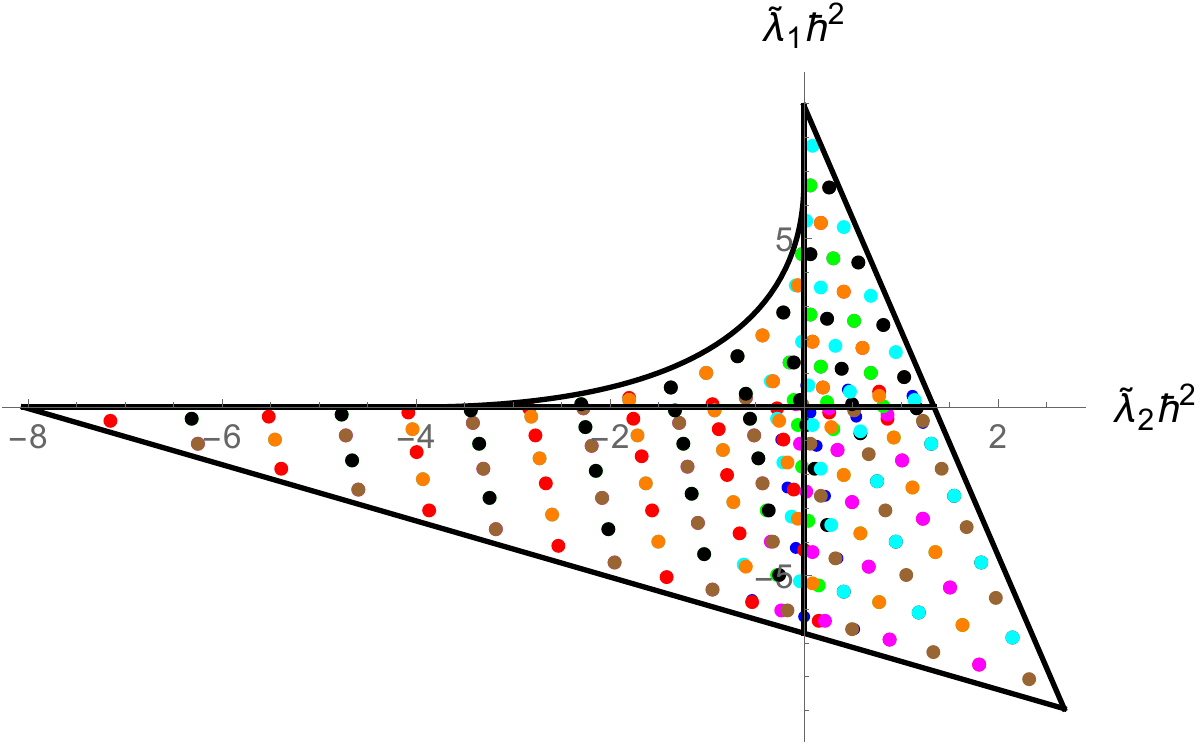}\includegraphics[width=7cm,height=7cm]{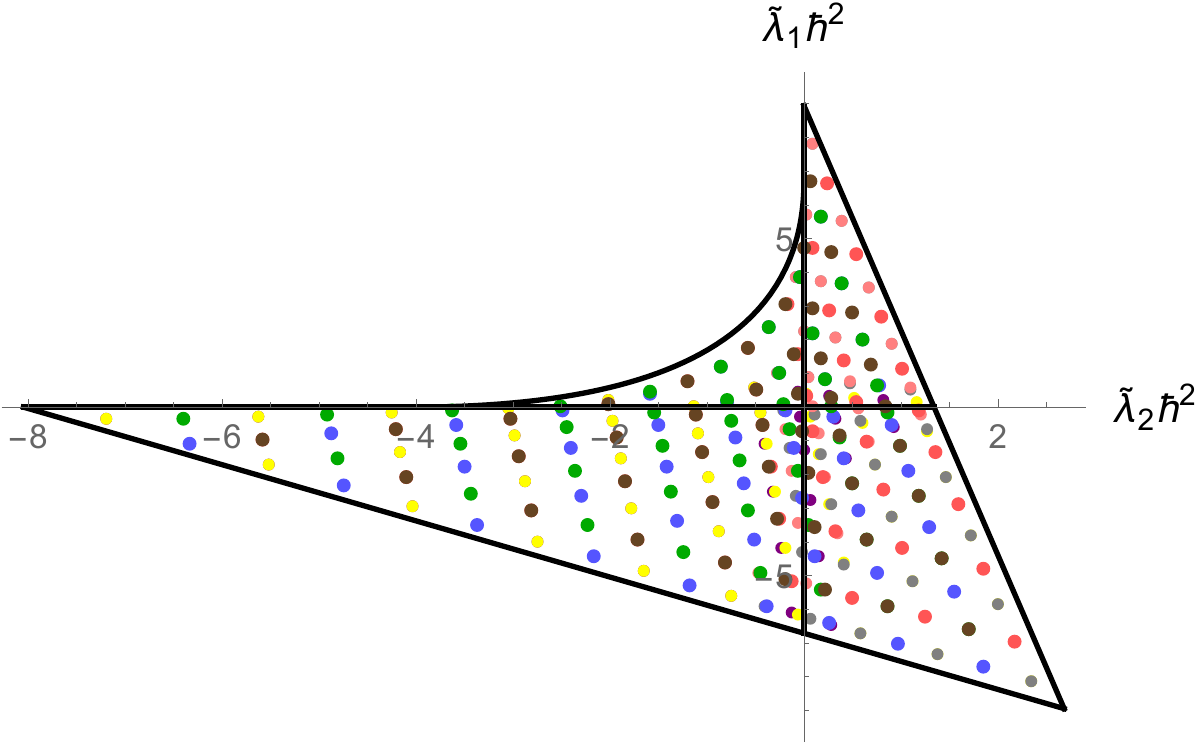}
\par\end{centering}
\caption{Spectra corresponding to Fig. \ref{fig:Fig1} a) and b) after performing an affine transformation to map the hyperbolic-hyperbolic point to the origin. \label{fig:Fig1rotated}}
\end{figure}
\begin{table}[H]
    \centering
    \begin{tabular}{|c|c|c|c|}
    \hline
    $\bm{\mu}$ & Colour a) & $\bm{\mu}$ & Colour b)\\
    \hline
       $(0,0,0,0)$  & Blue & $(1,0,0,0)$ & Purple\\
       \hline 
        $(1,1,0,0)$ & Red & $(0,1,0,0)$ & Yellow\\
        \hline 
        $(1,0,1,0)$ & Magenta & $(0,0,1,0)$ & Grey \\
        \hline 
        $(1,0,0,1)$ & Green & $(0,0,0,1)$ & Pink\\
        \hline 
        $(0,1,1,0)$ & Brown & $(1,1,1,0)$ & Light Blue\\
        \hline 
        $(0,1,0,1)$ & Black & $(1,1,0,1)$ & Dark Green\\
        \hline 
        $(0,0,1,1)$ & Cyan & $(0,1,1,1)$ & Light Red\\
        \hline 
        $(1,1,1,1)$ & Orange & $(1,0,1,1)$ & Dark Brown\\
    \hline
    \end{tabular}
    \caption{Symmetry class corresponding to each coloured eigenstate shown in Figure \ref{fig:Fig1}.}
    \label{tab:my_label}
\end{table}
\textcolor{black}{Equipped with the joint spectrum $(\lambda_{1},\lambda_{2})$,
we can investigate the corresponding action variables. Using the St\"{a}ckel
matrix (\ref{eq:Stackel El S3-1}) and the results of \cite{gurarie95}
and \cite{Nguyen2023}, we have the following Lemma. }
\begin{lem}
The three continuous, classical actions $(J_{1},J_{2},J_{3})$
are given by 
\begin{equation}
\begin{aligned}J_{1}=\frac{1}{\pi}\int_{e_{1}}^{\min(R_{1},e_{2})}p(z)dz &  & J_{2}=\frac{1}{\pi}\int_{\max(R_{1},e_{2})}^{\min(R_{2},e_{3})}p(z)dz &  & J_{3}=\frac{1}{\pi}\int_{\max(R_{2},e_{3})}^{e_{4}}p(z)dz\end{aligned}
\label{eq:ACtions ell dfef}
\end{equation}
where the momentum $p(z)$ is given by 
\[
p^{2}=\frac{-\tilde{E}z^{2}+\lambda_{1}z-\lambda_{2}}{4(z-e_{1})(z-e_{2})(z-e_{3})(z-e_{4})},
\]
and $R_{1},R_{2}$ denote the roots of $p^{2}$ with $0\le R_{1}\le R_{2}$.
\end{lem}
It was shown in \cite{Nguyen2023} that the 3 actions
are linearly dependent, continuous and satisfy
the following Lemma. 
\begin{lem}
The actions are constrained by the following relation
\begin{align}
J_{1}+J_{2}+J_{3} & =\sqrt{\tilde{E}}.\label{eq:Linear Dependence of Action}
\end{align}
\end{lem}
The fact that $H$ can be written as a function of the sum of the actions
in (\ref{eq:Linear Dependence of Action})  follows from the fact that $H$ is
superintegrable.
In Fig. \ref{fig:Fig2} a) and b) we show the actions of the even and odd spectra
in Fig. \ref{fig:Fig1} a) and b) respectively, computed using (\ref{eq:ACtions ell dfef})
with $\tilde{E}=1$. 
Similar plots for each symmetry class can be found in Appendix \ref{subsec:ActionsEllipsoida}.

The classical action variables even though continuous are complicated functions, and defining corresponding quantum mechanical operators may be difficult. Thus presenting the joint spectrum in the space of action variables is not done by quantising these functions. Instead we take the eigenvalues of the operators $(\hat H, \hat \eta_1, \hat \eta_2)$ and map them to $(J_1, J_2, J_3)$ with the previous Lemma. In this way a striking representation of the joint spectrum inside an equilateral triangle is obtained. Is shown in \cite{Nguyen2023} this triangle is rigid even when passing to degenerate coordinates systems. Note that this triangle is not Delzant (our system is not toric to begin with), but it seems to serve a similar role for our class of quantum integrable systems. Near each corner of the triangle the lattice is approximately $\mathbb{Z}^2$ with basis vectors given by the sides of the triangle. The action triangle is a natural global representation because there are three elliptic-elliptic equilibrium points in the reduced system, so that it represents three different $\mathbb{Z}^2$ lattices in one moment map.

Among the degenerate systems discussed in the next section the cylindrical coordinate system does lead to a toric system on $S^2 \times S^2$ with a square as Delzant polytope. The triangle appears as the quotient of the square by discrete symmetry group $\mathbb{Z}^2 \times \mathbb{Z}^2$ acting by reflection across the diagonals of the square, see below.

\textcolor{black}{}
\begin{figure}[H]
\begin{centering}
\textcolor{black}{\includegraphics[width=7cm]{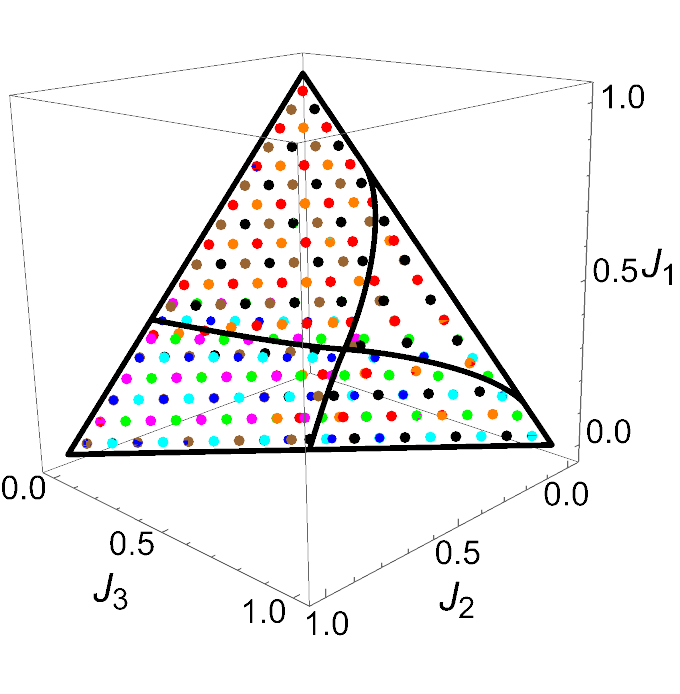}\includegraphics[width=7cm]{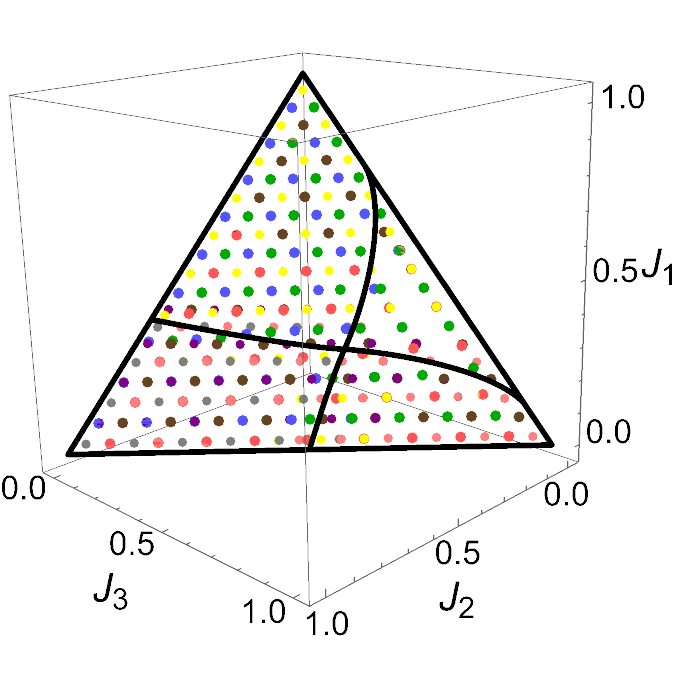}}
\par\end{centering}
\textcolor{black}{\caption{Joint spectrum in action variables in the ellipsoidal case corresponding to the joint spectra  shown in Fig \ref{fig:Fig1} a) and b) respectively. \label{fig:Fig2}}
}
\end{figure}

\section{Degenerate Systems}
We now study the degenerate systems arising from the generalised ellipsoidal coordinates $(1234)$. We originally focus on the prolate, oblate and Lam\'{e} coordinates; these correspond to the edges of the Stasheff polytope shown in Fig. \ref{fig:StasheffQuantum}. Further degenerations of these yield the spherical and cylindrical coordinates, represented by the corners of the polytope. 

We will show that the separated ODEs for these coordinates can be obtained by smoothly degenerating those of the generalised ellipsoidal coordinates (i.e.~the generalised Lam\'{e} equation). In doing this, we prove the following theorem which is the classical analogue of Theorem~3 of \cite{Nguyen2023}:

\begin{theorem}
    For each pair of parentheses in the labelled Stasheff polytope (see Fig.~\ref{fig:StasheffQuantum}) that enclose two adjacent numbers, the corresponding quantum integrable system has an $SO(2)$ symmetry and a corresponding ODE that can be transformed to a trivial ODE with trigonometric solutions. 

    For each pair of parentheses that enclose three adjacent numbers, the corresponding quantum integrable system has a global $SO(3)$ symmetry and a corresponding ODE that can be transformed to the Gegenbauer equation (hypergeometric type). This applies also when there is a pair of  parentheses inside. 

    If there are two pairs of parentheses (corresponding to the edges of the Stasheff polytope) then one ODE is of Heun type (Fuchsian with 4 regular singular points).

    If there are three pairs of parentheses (corresponding to the corners of the Stasheff polytope) then one ODE is of hypergeometric type (Fuchsian with 3 regular singular points). For spherical coordinates this can be transformed to the associated Legendre equation and for cylindrical coordinates this is the Jacobi equation.
    
    
\end{theorem}
Note that a single ODE of Heun type gives solutions for \emph{two} factors of the eigenfunction $\psi$. In the general ellipsoidal case treated in the previous section the single Fuchsian ODE with 4 finite regular singular points gives \emph{three} such factors. Any ODE of hypergeometric type (including the trivial trigonometric one) gives a solution for a single factor of the eigenfunction. Hence another way of stating which ODEs occur for the various cases is this: On the inside/edges/corners of the Stasheff polytope the highest Fuchsian ODE has 5/4/3 regular singular points, accounting for 3/2/1 of the factors of the eigenfunctions. The remaining 0/1/2 ODEs are of hypergeometric type.

\subsection{Prolate Coordinates and the Heun Equation}\label{sec:ProlateS3}

Prolate coordinates on $S^{3}$ are a degeneration
of ellipsoidal coordinates arising from limiting
the middle two semi major axes to each other, i.e. $e_{3}\to e_{2}^{+}$.
We normalise the $e_{i}$ according to $(e_{1},e_{2}=e_{3},e_{4})=(0,1,a)$.

From \cite{Kalnins1986}, an explicit representation
of prolate coordinates is 
\begin{equation}
\begin{aligned}x_{1}^{2} & =\frac{s_{1}s_{3}}{a} &  & x_{2}^{2}=-\frac{\left(s_{1}-1\right)s_{2}\left(s_{3}-1\right)}{a-1}\\
x_{3}^{2} & =\frac{\left(s_{1}-1\right)\left(s_{2}-1\right)\left(s_{3}-1\right)}{a-1} &  & x_{4}^{2}=\frac{\left(a-s_{1}\right)\left(a-s_{3}\right)}{(a-1)a}
\end{aligned}
\label{eq:prolate coord def}
\end{equation}
where $0\le s_{1},s_{2}\le1\le s_{3}\le a$ and a possible St\"{a}ckel
matrix is given by 
\begin{equation}
\sigma_{\text{pro}}=\frac{1}{4}\left(\begin{array}{ccc}
-\frac{1}{\left(s_{1}-1\right)\left(s_{1}-a\right)} & -\frac{1}{\left(s_{1}-1\right)s_{1}\left(s_{1}-a\right)} & -\frac{1-a}{\left(s_{1}-1\right){}^{2}s_{1}\left(s_{1}-a\right)}\\
0 & 0 & -\frac{1}{\left(s_{2}-1\right)s_{2}}\\
-\frac{1}{\left(s_{3}-1\right)\left(s_{3}-a\right)} & -\frac{1}{\left(s_{3}-1\right)s_{3}\left(s_{3}-a\right)} & -\frac{1-a}{\left(s_{3}-1\right){}^{2}s_{3}\left(s_{3}-a\right)}
\end{array}\right).\label{eq:stackel matrix pro}
\end{equation}

\textcolor{black}{From \cite{Nguyen2023} we have the following result.}
\begin{lem}
Separating the Hamilton-Jacobi equation in prolate
coordinates and subsequently reducing by $\sqrt{H}$ 
gives a semi-toric integrable system on $S^{2}\times S^{2}$ with
integrals
\begin{align*}
G= a\ell_{12}^{2}+a\ell_{13}^{2}+\ell_{14}^{2} &  & M=\ell_{23}^{2}
\end{align*}
and Poisson structure $B_{\bm{X},\bm{Y}}$ as
given in (\ref{eq:BXY}).
\end{lem}

To obtain the separated equations, we use the St\"{a}ckel matrix  (\ref{eq:stackel matrix pro}) and the metric
tensor for prolate coordinates \eqref{eq:prolate coord def}, and follow the same methodology as described in Lemma \ref{SeparationLemma}. For all future systems studied in this paper, this same technique is used to obtain the separated ODEs.
This gives the following. 
\begin{lem}
Separating the Schrödinger equation (\ref{eq:Schrodinger})
in prolate coordinates (\ref{eq:prolate coord def}) gives the following
separated equations
\begin{subequations}
\begin{align}
\psi_{j}^{''}+\frac{1}{2}\left(\frac{1}{s_{j}}+\frac{2}{s_{j}-1}+\frac{1}{s_{j}-a}\right)\psi_{j}^{'}+\frac{-Es_{j}^{2}+\left(\lambda+E+(a-1)m^{2}\right)s_{j}-\lambda}{4s_{j}(s_{j}-1)^{2}(s_{j}-a)}\psi_{j} & =0 & \text{for \ensuremath{j=1,3}}\label{eq:ProlateEqn1}\\
\psi_{j}^{''}+\frac{1}{2}\left(\frac{1}{s_{j}}+\frac{1}{s_{j}-1}\right)\psi_{j}^{'}+\frac{m^{2}}{4s_{j}(1-s_{j})}\psi_{j} & =0 & \text{for \ensuremath{j=2}}\label{eq:ProlateEqn2}
\end{align}
\end{subequations}
where $(m,\lambda)$ are spectral parameters. 
\end{lem}
We observe that the separated equation in (\ref{eq:ProlateEqn2}) is of hypergeometric type, but via the coordinate transformation $s_{2}=\cos^{2}\phi$ we recover the trivial ODE 
\begin{align}
\psi^{''}_2(\phi)+m^{2} & =0.\label{eq:trivial equation}
\end{align}
To ensure a smooth globally defined solution
on $S^{3}$, we enforce periodic boundary conditions.
This yields a discrete spectral parameter 
$m\in\mathbb{Z}$ representing the quantised angular momentum in the
$(x_{2},x_{3})$ plane with solutions $e^{im\phi}$ to (\ref{eq:trivial equation}). 

In \cite{Nguyen2023} we showed that the classical integrals for the degenerate systems are obtained by a smooth degeneration of the coordinates. The following Lemma shows how the ODEs \eqref{eq:ProlateEqn1} and \eqref{eq:ProlateEqn2} naturally arise from the generalised Lam\'{e} equation \eqref{eq:Gen Lame} in a similar fashion.
\begin{lem}
The separated equations for $\psi_{j}$ in (\ref{eq:ProlateEqn1}) and (\ref{eq:ProlateEqn2})
can be obtained by smoothly degenerating the generalised Lam\'{e}
equation (\ref{eq:Gen Lame}).
\end{lem}
\begin{proof}
As noted in \cite{Nguyen2023}, the canonical transformation
from ellipsoidal to prolate coordinates is given by 
\begin{align}
e_{3} & =e_{2}+\epsilon, & s_{2} & =e_{2}+\epsilon\tilde{s}_{2}, & p_{2} & =\frac{\tilde{p}_{2}}{\epsilon},\label{eq:prodegen}
\end{align}
in the limit $\epsilon\to0$ where $\tilde{s}_{2}\in[0,1]$ and $\tilde{p}_2$ is its conjugate momenta. We
normalise the semi-major axes with $(e_{1},e_{2}=e_{3},e_{4})=(0,1,a)$
and set
\begin{align*}
\lambda=\lambda_{2}, &  & m^{2}=\frac{1}{a-1}(\lambda_{1}-\lambda_{2}-E).
\end{align*}
For the equations in $s_1$ and $s_3$, the first derivative terms trivially degenerate while for the multiplicative 
term, substituting (\ref{eq:prodegen}) into (\ref{eq:Gen Lame})
and taking the limit $\epsilon\to0$ gives (\ref{eq:ProlateEqn1}).

Repeating the same process for the equations in $s_2$ gives the desired result.
\end{proof}

Note that the St\"{a}ckel matrix used to obtain the separation constants is not unique and so neither are the arising ODEs. Another choice of St\"{a}ckel matrix would result in a different (yet equivalent) set of ODEs since the separation constants would be different. The first derivative term is specific to the coordinates and so is independent of the choice of St\"{a}ckel matrix.

Equation (\ref{eq:ProlateEqn1}) is
an example of a Heun equation; a Fuchsian $4$ 
ODE with regular singularities
at $0,1,a$ and $\infty$: 
\begin{align}
W^{''}+\left(\frac{\gamma}{z}+\frac{\delta}{z-1}+\frac{\epsilon}{z-a}\right)W^{'}+\frac{\alpha\beta z-q}{z(z-1)(z-a)}W & =0\label{eq:GenHeun}
\end{align}
where $\alpha+\beta+1=\gamma+\delta+\epsilon$. The corresponding
Riemann symbol is given by 
\begin{align}
\mathcal{S}_{\text{Heun}} & =\left(\begin{aligned}0 &  & 1 &  & a &  & \infty\\
0 &  & 0 &  & 0 &  & \alpha & ;z\\
1-\gamma &  & 1-\delta &  & 1-\epsilon &  & \beta
\end{aligned}
\right).\label{eq:Heun Riemann}
\end{align}
\begin{lem}
By the change of dependent variable $\psi_j=(z-1)^{|m|/2}W_j$,
(\ref{eq:ProlateEqn1}) can be written in the form of (\ref{eq:GenHeun})
where $\gamma=\epsilon=\frac{1}{2},\delta=1+|m|,q=\frac{1}{4}\left(a|m|-\lambda\right)$
and 
\begin{align}
\alpha,\beta & =\frac{1}{2}\left(1\pm\sqrt{1+E}+\left|m\right|\right).\label{eq:alphabeta}
\end{align}
Specifically, we have 
\begin{equation}
    W_j^{''} + \left(\frac{1}{2z} + \frac{1+ |m|}{2(z-1)} + \frac{1}{2(z-a)}\right)W_j^{'} + \frac{(-E + |m|(|m|+2))z- (a|m| - \lambda)}{4z(z-1)(z-a)}W_j\label{eq:specificHeunPro}
\end{equation}
\end{lem}
The solution $H\ell(a,q;\alpha,\beta,\gamma,\delta;z)$
that corresponds to exponent $0$ about $z=0$ such that $H\ell(a,q;\alpha,\beta,\gamma,\delta;0)=1$
is called a Heun function. We use the notation from \cite{DLMF}. Assuming $\gamma\notin-\mathbb{N}$, the
Heun function can be expanded as an infinite series
\begin{align}
H\ell(a,q;\alpha,\beta,\gamma,\delta;z) & =\sum_{i=0}^{\infty}c_{i}z^{i}\label{eq:Heun solutions}
\end{align}
 where $|z|<1$ and $c_{0}=1$. Substituting (\ref{eq:Heun solutions})
into (\ref{eq:GenHeun}) gives the following Lemma taken from \cite{ronveaux95}.
\begin{lem}
\label{LemmaProlate}The coefficients $c_{i}$ in
(\ref{eq:Heun solutions}) satisfy the following three term recurrence relation
\begin{align}
A_{i}c_{i+1}-(B_{i}+q)c_{i}+C_{i}c_{i-1} & =0\label{eq:3 term}
\end{align}
 where, for $i\ge1$:
\begin{align}
A_{i} & =a(i+1)(i+\gamma)\nonumber \\
B_{i} & =i\left[(i-1+\gamma)(a+1)+a\delta+\epsilon\right]\label{eq:Heun3term}\\
C_{i} & =(i-1+\alpha)(i-1+\beta)\nonumber 
\end{align}
 subject to the conditions $c_{0}=1$ and 
\begin{equation}
a\gamma c_{1}-qc_{0} =0.
\end{equation}
\end{lem}

From (\ref{eq:Heun3term}) it is clear that $\alpha=-d$ where $d\in\mathbb{N}$ forces $C_{d+1}=0$ and the three
term recurrence (\ref{eq:3 term}) truncates (\ref{eq:Heun solutions})
to a polynomial of degree $d$. These solutions, analytic at
all three finite singularities $z=0,1,a$ are known as Heun polynomials
and will be denoted by $Hp(a,q;\alpha,\beta,\gamma,\delta;z)$.
In the general ellipsoidal case Heine-Stieltjes polynomials were computed by solving a non-linear system of equations for the roots of the eigenfunctions. In the Heun case the above recursion gives a different and simpler method for computing the eigenfunctions and eigenvalues. 
In the remainder of this section we show how to obtain the corresponding harmonic polynomials in $\mathbb{R}^4$ and how to compute the various discrete symmetry classes of eigenfunctions.


Using a technique similar to the ellipsoidal case and outlined by \cite{Volkmer1999},
we have the following Lemma.
\begin{lem}
\label{prolateincartesianlemma}Let $Hp_{d}(z)$
be a Heun polynomial solution to \eqref{eq:specificHeunPro} of degree $d$ and
denote by $z_{1},\dots,z_{d}$ its roots. The product $Hp_d(s_1)Hp_d(s_3)$, expressed in the original Cartesian coordinates, is given by 
\begin{align}
Hp_d(s_1)Hp_d(s_3) & =\prod_{k=1}^{d}z_{k}(z_{k}-a)\left[\frac{x_{1}^{2}}{z_{k}}-r^2+x_{4}^{2}\frac{1-a}{z_{k}-a}\right]\label{eq:procombinedwfcartesian-1}
\end{align}
where $r^2 = x_1^2 + x_2^2 + x_3 ^2 + x_4 ^2$.
\end{lem}
\begin{proof}
Firstly, we claim that 
\begin{align}
(\theta-s_{1})(\theta-s_{3}) & =\theta(\theta-a)\left(r^2-\frac{x_{1}^{2}}{\theta}+\frac{(a-1)}{\theta-a}x_{4}^{2}\right)\label{eq:HeinProof-1}
\end{align}
for all $\theta\ne 0,a$.
Both the left and right hand sides of (\ref{eq:HeinProof-1}) are
monic quadratic functions of $\theta$. Further, both expressions
have roots at $\theta=s_{1},s_{3}$. This is easily verified for the
right hand side by using the coordinate transform (\ref{eq:prolate coord def}) and the condition $r^2=1$.

Let $Hp_{d}(z)$ be a Heine-Stieltjes polynomial of degree $d$ with
roots $z_{1},z_{2},\dots,z_{d}$ in factored form, i.e. 
\begin{equation}
Hp_{d}(z) =(z_{1}-z)(z_{2}-z)\dots(z_{d}-z).\label{eq:heinroots-1}
\end{equation}
Combining (\ref{eq:HeinProof-1}) with (\ref{eq:heinroots-1}) gives the desired result. \end{proof}

Note that we have solved the canonical Heun equation
(\ref{eq:GenHeun}). To write the combined solution to (\ref{eq:ProlateEqn1}) and (\ref{eq:ProlateEqn2}),
we include the normalisation factor by recalling $\psi_{1}=(1-s_{1})^{|m|/2}W_{1}$.
For the $s_{3}$ equation, since $s_{3}\ge1$ we have $\psi_{3}=(s_{3}-1)^{|m|/2}W_{3}$.
Hence, the combined solution to \eqref{eq:Schrodinger} in prolate coordinates is the product
\begin{equation}
\psi_{1}\psi_{2}\psi_{3}=(1-s_{1})^{|m|/2}(s_{3}-1)^{|m|/2}Hp_{d}(s_{1})Hp_{d}(s_{2})e^{im\phi}.\label{eq:preProlateSol}
\end{equation}
This yields the following Lemma.
\begin{lem}
Expressed in Cartesian coordinates, the product $\psi_1\psi_2\psi_3$ in \eqref{eq:preProlateSol} can be written as the following $\tilde{D}$ degree homogeneous polynomial where $\tilde{D} = 2d + |m|$: 
\begin{equation}
\Phi_{\tilde{D}}(\bm{x})=(a-1)^{|m|/2}\left(x_{2}+i \sign(m)x_{3}\right)^{|m|}\prod_{k=1}^{d}z_{k}(z_{k}-a)\left[\frac{x_{1}^{2}}{z_{k}}-r^2+x_{4}^{2}\frac{1-a}{z_{k}-a}\right]\label{eq:FinalPSID}
\end{equation}
and $m\in[-\tilde{D},\tilde{D}]$ is an integer.
\end{lem}
\begin{proof}
Using the definition of prolate coordinates in
(\ref{eq:prolate coord def}) we obtain
\[
(1-s_{1})^{m/2}(s_{3}-1)^{m/2}=\left[(a-1)(x_{2}^{2}+x_{3}^{2})\right]^{m/2}.
\]
We also know that $e^{im\phi}=\left(\frac{x_{2}+ix_{3}}{\sqrt{x_{2}^{2}+x_{3}^{2}}}\right)^{m}$. Combining this with (\ref{eq:procombinedwfcartesian-1}) and (\ref{eq:preProlateSol})
gives the product in (\ref{eq:FinalPSID}). For homogeneity, given that $r^2 = \sum_{i=1}^{4}x_{i}^2$, it is clear that all terms in \eqref{eq:FinalPSID} are of degree $\tilde{D} = 2d + |m|$.
\end{proof}
Like in the ellipsoidal case, there are $16$ discrete symmetry classes of $\Phi_{\tilde{D}}$ corresponding to parity about each of the $x_i$ axes. However, we note that a simple parity flip of either $x_2$ or $x_3$ for the the complex phase $(x_2+ix_3)^m$ would not necessarily yield a materially different wave function depending on the parity of $m$.

To address this, we note that for odd $m$ we have $\Re(x_2+ix_3)^{m} = x_2P(x_2^2,x_3^2)$ and $\Im(x_2+ix_3)^{m} = x_3Q(x_2^2,x_3^2)$ where $P$ and $Q$ are some polynomials of degree $(m-1)/2$. These therefore give solutions which are odd about the $x_2$ $(x_3)$ axes and even about $x_3$ $(x_2)$ respectively. Similarly, we note that for even $m$ we have $\Re(x_2+ix_3)^{m} = \tilde{P}(x_2^2,x_3^2)$ and $\Im(x_2+ix_3)^{m}=x_2x_3\tilde{Q}(x_2^2,x_3^2)$ where $\tilde{P}$ and $\tilde{Q}$ are polynomials of degree $m/2$ and $(m-2)/2$ respectively. The former give solutions which are even about both $x_2$ and $x_3$ and the latter those which are odd about both of these axes.

In the following discussion, we only focus on the discrete symmetries about the $x_1$ and $x_4$ axes. This is because symmetries about the $x_2$ and $x_3$ axes are encoded in taking the real and imaginary components of the wave function, in conjunction with a parity choice of $m$, as just described. This is similar in other degenerate systems (oblate, spherical, etc).

Let $\bm{\mu} = (\mu_1, \mu_4)$ where $\mu_i \in \{0, 1\}$ and $\mu_i$ being $1$ ($0$) denotes a solution odd (even) about the $x_i$ axis.   
For these symmetries, we consider the change of dependent variable in \eqref{eq:specificHeunPro} given by $\phi_1 = s_1^{\mu_{1}/2}(a-s_1)^{\mu_{1}/2}W_1$ and $\phi_3 = s_3^{\mu_{4}/2}(a-s_3)^{\mu_{4}/2}W_3$
which corresponds to multiplication of the total wavefunction by $x_1^{\mu_1}x_4^{\mu_4}$. Doing so gives a resulting Heun equation with parameters shown in Table \ref{tab: Heunpro Params}.

Denote by $Hp_{d}^{\bm{\mu}}(z)$ solutions to \eqref{eq:GenHeun} corresponding to the $\bm{\mu}$ symmetry class. Let the product $Hp_{d}^{\bm{\mu}}(s_1)Hp_{d}^{\bm{\mu}}(s_3)e^{im\phi}$, when converted back to Cartesian coordinates, be given by $\Phi_{\tilde{D}}^{\bm{\mu}}(\bm{x})$. As with the ellipsoidal case, we set 
\begin{equation}
    \Psi_{D}^{\bm{\mu}}(\bm{x}) \coloneqq x_{1}^{\mu_1}x_{4}^{\mu_4}\Phi_{\tilde{D}}^{\bm{\mu}}(\bm{x})
\end{equation}
where $D \coloneqq \tilde{D} + \mu_1+\mu_4$ and $m\in[-D, D]$ is an integer for all symmetry classes. Note that \eqref{eq:FinalPSID} is a special case of $\Psi_{D}^{\bm{\mu}}(\bm{x})$ with $\bm{\mu} = (0,0)$. 
\begin{lem}
    The $\Psi_{D}^{\bm{\mu}}(\bm{x})$ are $D$ degree, homogeneous harmonic polynomials and the energy eigenvalue is given by $E=D(D+2)$. 
\end{lem}
\begin{proof}
    The proof that $\Psi_{D}^{\bm{\mu}}(\bm{x})$ satisfy \eqref{eq:Schrodinger} is identical to the ellipsoidal case. For the energy eigenvalue, we note that a necessary condition for polynomial solutions is for the three-term recurrence relation \eqref{eq:3 term} to vanish. This forces $C_i = 0$ in \eqref{eq:Heun3term}, leading to the condition \begin{equation}
        \alpha = -d + \frac{1}{2}(\mu_1+\mu_4). \label{eq:alphaProCond}
    \end{equation} 
    Substituting \eqref{eq:alphaProCond} into the expression for $\alpha$ in \eqref{eq:alphabeta} gives the desired result.  
\end{proof}
From \eqref{eq:alphaProCond} we observe that for the $(0,0)$ and $(1,1)$ symmetry classes, $\alpha$ is a negative integer while for the $(0,1)$ and $(1,0)$ classes, $\alpha$ is a negative half integer. Further, we note that fixing a value of $E$ (in turn fixing $D$) enforces a parity relationship between $E$ and $m$ for a given symmetry class. For even (odd) $E$, $m$ must be even (odd) for the $(0,0)$ and $(1,1)$ classes while $m$ must be odd (even) for the others.
\begin{table}
\begin{centering}
\begin{tabular}{|c|c|c|c|c|c|c|}
\hline 
\textcolor{black}{$(\mu_1,\mu_4)$} & \textcolor{black}{$\tilde{\alpha}$} & \textcolor{black}{$\tilde{\beta}$} & \textcolor{black}{$\tilde{\gamma}$} & \textcolor{black}{$\tilde{\delta}$} & \textcolor{black}{$\tilde{\epsilon}$} & \textcolor{black}{$\tilde{q}$}\tabularnewline
\hline 
\textcolor{black}{$(0,0)$} & \textcolor{black}{$\alpha$} & \textcolor{black}{$\beta$} & \textcolor{black}{$\gamma$} & \textcolor{black}{$\delta$} & \textcolor{black}{$\epsilon$} & \textcolor{black}{$q$}\tabularnewline
\hline 
\textcolor{black}{$(1,0)$} & \textcolor{black}{$\alpha+1-\gamma$} & \textcolor{black}{$\beta+1-\gamma$} & \textcolor{black}{$2-\gamma$} & \textcolor{black}{$\delta$} & \textcolor{black}{$\epsilon$} & \textcolor{black}{$(1-\gamma)(a\delta+\epsilon)+q$}\tabularnewline
\hline 
\textcolor{black}{$(0,1)$} & \textcolor{black}{$\alpha+1-\epsilon$} & \textcolor{black}{$\beta+1-\epsilon$} & \textcolor{black}{$\gamma$} & \textcolor{black}{$\delta$} & \textcolor{black}{$2-\epsilon$} & \textcolor{black}{$\gamma(1-\epsilon)+q$}\tabularnewline
\hline 
\textcolor{black}{$(1,1)$} & \textcolor{black}{$2+\alpha-\gamma-\epsilon$} & \textcolor{black}{$2+\beta-\gamma-\epsilon$} & \textcolor{black}{$2-\gamma$} & \textcolor{black}{$\delta$} & \textcolor{black}{$2-\epsilon$} & \textcolor{black}{$q+2+a\delta-\epsilon-\gamma(1+a\delta)$}\tabularnewline
\hline 
\end{tabular}
\par\end{centering}
\textcolor{black}{\caption{Parameters of (\ref{eq:GenHeun}) for the various symmetry classes.\label{tab: Heunpro Params} }
}
\end{table}

As in the ellipsoidal case, we can compute the total number of states for a fixed energy $E$ as well as that for
each symmetry class.
\begin{lem}
Fix an energy $E=D(D+2)$. Then the total
number of eigenstates is given by 
\begin{align}
N & =(D+1)^{2}.\label{eq:TotalStatesPro}
\end{align}
The total number of eigenstates for each symmetry class is shown in
Table \ref{tab:Pro Num States}.
\end{lem}
\begin{proof}
Assume $D$ is even and consider
the $(\mu_1,\mu_4)=(0,0)$ symmetry class; the other cases are proven
in a similar manner. We observe that $|m|$ ranges over the $D+1$
values $-D, -D+2, \dots, 0,\dots,D-2,D$ since $|m|$ is required to
be an even integer based on our condition for $\alpha$. For a fixed energy and the relationship between
$d$ and $|m|$ this enforces, we have a total of $\frac{D-|m|}{2}+1$
states for a given $m$.

\textcolor{black}{Summing the total number of states over all possible
values of $m$ gives 
\begin{align*}
N_{(0,0)}^{even}= & \left(\frac{D}{2}+1\right)+2\times\left(\frac{D}{2}+\left(\frac{D}{2}-1\right)+\dots+1\right)=\left(\frac{D}{2}+1\right)^{2}.
\end{align*}
Repeating a similar procedure for all other symmetry classes gives
the desired result. Regardless of the parity of $D$, summing over
each column of Table \ref{tab:Pro Num States} gives (\ref{eq:TotalStatesPro}).}
\end{proof}
\begin{table}
\begin{centering}
\textcolor{black}{}%
\begin{tabular}{|c|c|c|}
\hline 
\textcolor{black}{$(\mu_1,\mu_4)$} & \textcolor{black}{$E$ even} & \textcolor{black}{$E$ odd}\tabularnewline
\hline 
\textcolor{black}{$(0,0)$} & \textcolor{black}{$\left(\frac{D}{2}+1\right)^{2}$} & \textcolor{black}{$\left(\frac{D+1}{2}\right)\left(\frac{D+3}{2}\right)$}\tabularnewline
\hline 
\textcolor{black}{$(1,0)$} & \textcolor{black}{$\left(\frac{D}{2}+1\right)\frac{D}{2}$} & \textcolor{black}{$\left(\frac{D+1}{2}\right)^{2}$}\tabularnewline
\hline 
\textcolor{black}{$(0,1)$} & \textcolor{black}{$\left(\frac{D}{2}+1\right)\frac{D}{2}$} & \textcolor{black}{$\left(\frac{D+1}{2}\right)^{2}$}\tabularnewline
\hline 
\textcolor{black}{$(1,1)$} & \textcolor{black}{$\left(\frac{D}{2}\right)^{2}$} & \textcolor{black}{$\left(\frac{D+1}{2}\right)\left(\frac{D-1}{2}\right)$}\tabularnewline
\hline 
\end{tabular}
\par\end{centering}
\textcolor{black}{\caption{Number of states for each symmetry class and each parity of $E$.
\label{tab:Pro Num States}}
}
\end{table}

As with the ellipsoidal case, we have $\hbar$ as the inverse of an integer, namely $\hbar=\frac{1}{D+1}$. The joint spectrum $(m\hbar,\lambda\hbar^{2})$ is obtained for fixed $\tilde{E}$ and total degree $d$ by numerically computing the eigenvalues of the matrix (\ref{eq:Heun3term}) to obtain $d$ eigenvalues $q$ and hence $\lambda$ for given $m$. 
We use Mathematica's \texttt{Eigenvalue[]} command and the numerics is fast and accurate for moderately large $D=2d+|m|\approx20$. Since the matrices \eqref{eq:Heun3term} are not self adjoint, substantially larger $d$  results in non-trivial numerical error.

An example of the joint spectrum is shown in Fig. \ref{fig: Prolate Spectrum and Actions}
a) below. We have chosen $D=20$ giving $\hbar=\frac{1}{21}$
and a total of $(20+1)^{2}=441$ total eigenstates. The blue, orange,
red and cyan points correspond to the $(0,0),(1,0),(0,1),(1,1)$ symmetry
classes respectively. Since we have chose an even value of $E$ (and thus $D$) we observe states of the $(0,0)$ and $(1,1)$ symmetry classes only along even $m$. Similarly, along odd $m$ only the $(1,0)$ and $(0,1)$ states are present. Had we chosen an odd $E$ then orange/red dots would be seen along even $m$ and vice versa.
We also show the classical boundary of the momentum map in black. As mentioned, this is a semi-toric system and an isolated rank-2 critical point of focus-focus variety is located at $(m\hbar,\lambda\hbar^{2})=(0,1)$. This is shown as the magenta dot.

From the same reasoning used to obtain (\ref{eq:ACtions ell dfef}),
we obtain the following formulae for the actions
\begin{equation}
\begin{aligned}J_{1,pro}=\frac{1}{\pi}\int_{0}^{\min(r_{1},1)}p_{1}ds_{1} &  & J_{2,pro} & =\frac{1}{\pi}\int_{0}^{1}p_{2}ds_{2} &  &  & J_{3,pro}=\frac{1}{\pi}\int_{\max(1,r_{2})}^{a}p_{3}ds_{3}\end{aligned}
\label{eq:action prol}
\end{equation}
where, using (\ref{eq:stackel matrix pro}), for $j=1,3$
\begin{equation}
p_{j}^{2} =\frac{\tilde{E}s_{j}^{2}+\left(\lambda+\tilde{E}+(a-1)m^{2}\right)s_{j}-\lambda}{4s_{j}(s_{j}-1)^{2}(s_{j}-a)}
\end{equation}
and $(r_{1},r_{2})$ denote the roots of $p_{j}^{2}$ with $0\le r_{1}\le1\le r_{2}\le a$.
Note that $J_{2}$ simplifies to $m=\left|\ell_{23}\right|$. Like
for the ellipsoidal system, the prolate actions also satisfy (\ref{eq:Linear Dependence of Action}).
The action map for the prolate system is shown in Fig. \ref{fig: Prolate Spectrum and Actions} b) where the magenta dot corresponding to the focus-focus point is
located at $\frac{2}{\pi}(\sin^{-1}(\frac{1}{a}),0,\frac{\pi}{2}-\sin^{-1}(\frac{1}{a}))$.
Eigenstates with $m=0$ have $J_1 = 0$ and are hence on the boundary of the triangle. The reason for this is that the discrete symmetry reduced action variables in the quantum setting would require to impose Neumann boundary conditions.

\begin{figure}[H]
\begin{centering}
\includegraphics[width=7.5cm,height=7.5cm,keepaspectratio]{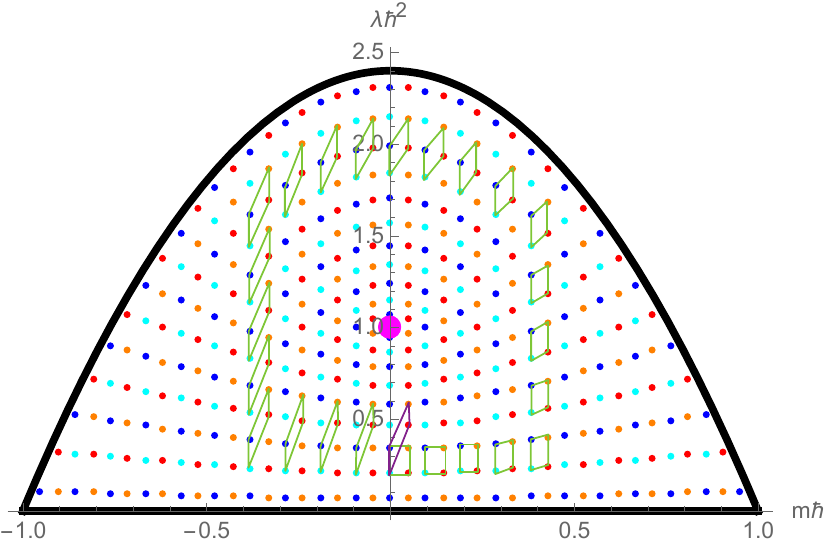}\includegraphics[width=7.5cm,height=7.5cm,keepaspectratio]{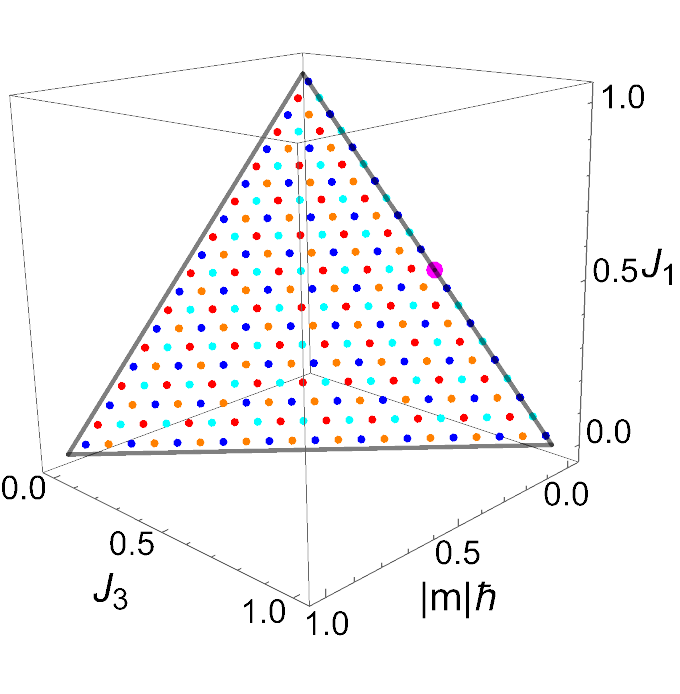}
\par\end{centering}
\caption{a) Prolate spectrum with $D=20$, $a=2.4$ and a total of $21^{2}$ eigenstates.
Different symmetry classes are represented with different colours. Note the focus-focus critical value shown in magenta. 
b) Action map corresponding to the joint spectrum. \label{fig: Prolate Spectrum and Actions}}
\end{figure}

Since the classical integrable system is semi-toric and has a focus-focus point, we expect the corresponding joint spectrum to exhibit quantum
monodromy. This is shown in Fig. \ref{fig: Prolate Spectrum and Actions}
a). A unit cell is parallel transported around the focus-focus critical
value. Denote the basis vectors of this unit cell as $v_{1}$ (vertical)
and $v_{2}$ (horizontal). As the cell completes a full loop, we observe
a basis transformation where $v_{1}$ remains unchanged but $v_{2}$
is updated to $v_{2}+2v_{1}$. In a similar vein to our previous work
\cite{spheroidal}, we have the following
basis transformation
\begin{equation}
\begin{pmatrix}v_{1}^{'}\\
v_{2}^{'}
\end{pmatrix} =\begin{pmatrix}1 & 0\\
\omega & 1
\end{pmatrix}\begin{pmatrix}v_{1}\\
v_{2}
\end{pmatrix}
\end{equation}
where $\omega=2$. As was observed in \cite{spheroidal},
when considering the symmetry classes one at a time, $\omega=1$ whereas
combining them two at a time gives $\omega=2$. Finally, in Fig.
\ref{fig:a)-Projection-of} we show the projection of the joint spectrum in the action variables onto either the $(m,J_{1})$ or $(m,J_{3})$ axes (and plotting the signed $m$ instead of $|m|$). The resulting polygons are two different representations of the semi-toric polygon invariant with the cut above and below the focus-focus point.

\begin{figure}[H]
\begin{centering}
\includegraphics[width=7.5cm,height=7.5cm,keepaspectratio]{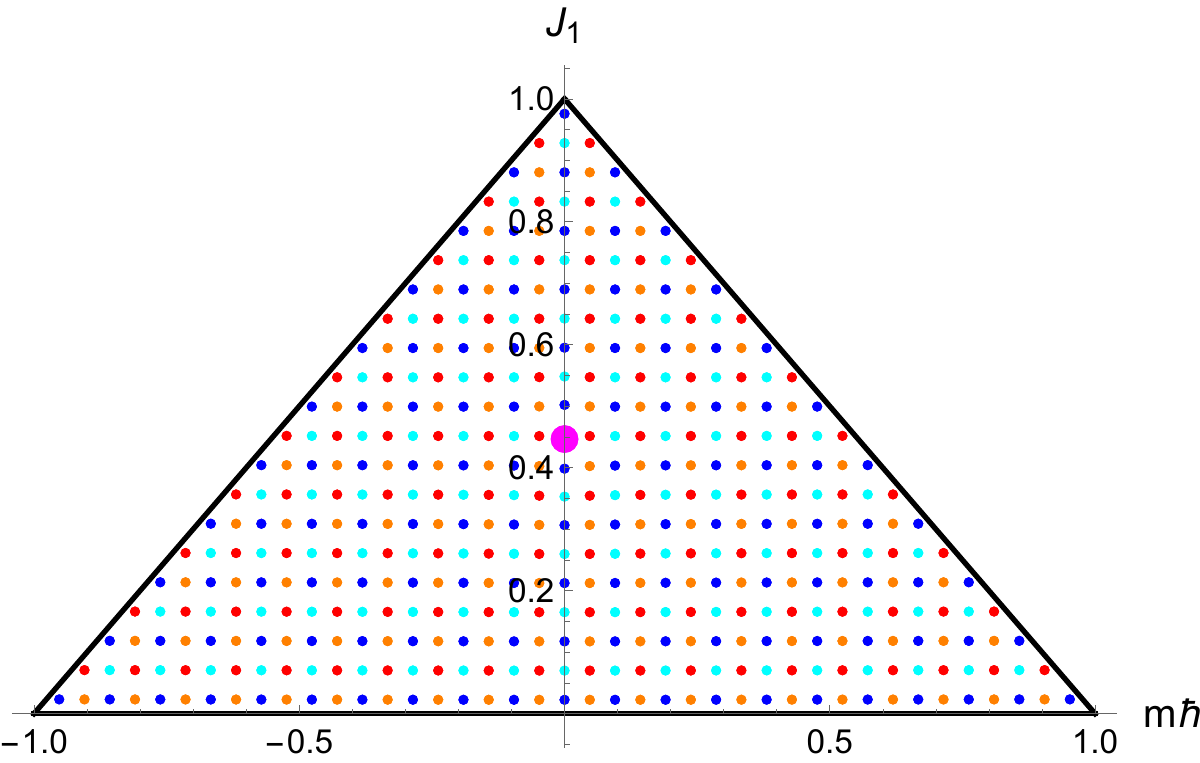}\includegraphics[width=7.5cm,height=7.5cm,keepaspectratio]{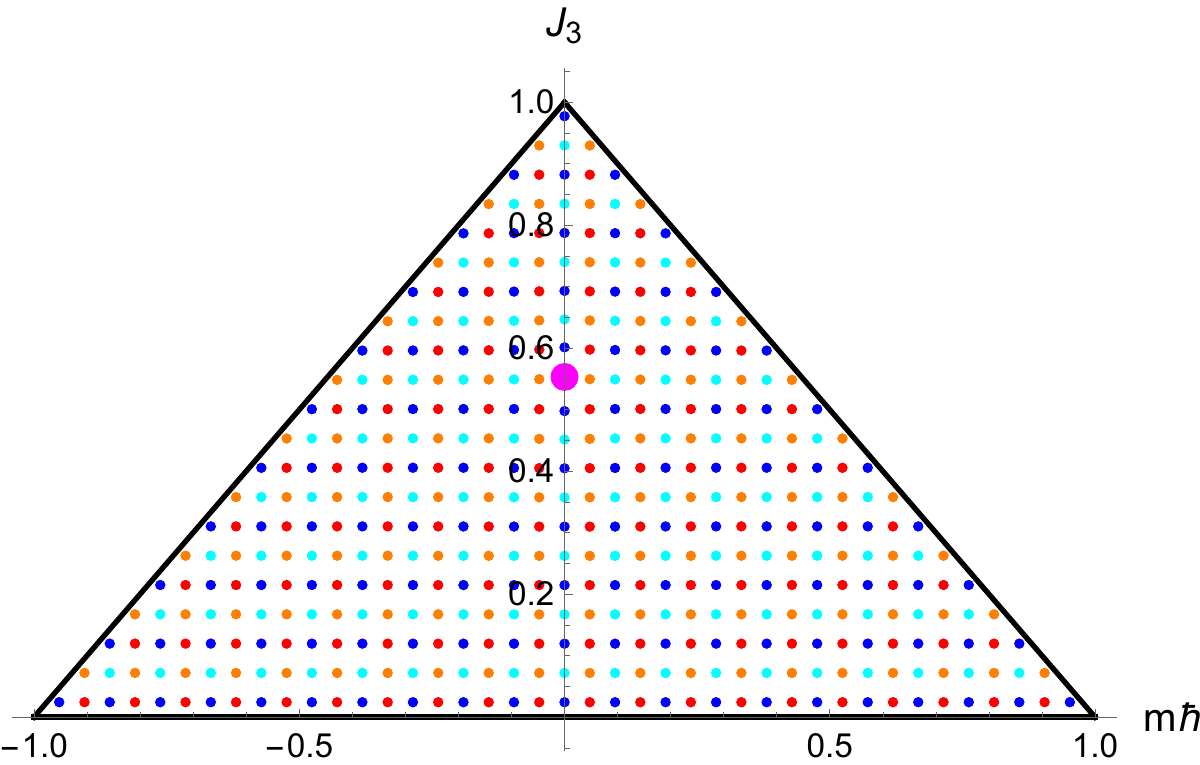}
\par\end{centering}
\caption{a) The semi-toric polygon in the prolate case $D=20$, $a=2.4$; projection of the actions onto the $(m,J_{1})$ axes, b) Projection
onto the $(m,J_{3})$ axes. \label{fig:a)-Projection-of}}

\end{figure}

\subsection{Oblate Coordinates}\label{sec:OblateS3}

Degenerating either the two largest $(e_{3}=e_{4})$
or smallest $(e_{1}=e_{2})$ semi major axes gives rise to the oblate
coordinates on $S^{3}$. They are equivalent under sending $e_{i}\to-e_{i}$
and then translating and scaling.
In this paper we only address the former case and normalise the $e_{i}$
with $(e_{1},e_{2},e_{3}=e_{4})=(0,1,a)$. 

From \cite{Kalnins1986}, an explicit representation
of oblate coordinates is 
\begin{equation}
\begin{aligned}x_{1}^{2} & =\frac{s_{1}s_{2}}{a} &  & x_{2}^{2}=\frac{-\left(s_{1}-1\right)\left(s_{2}-1\right)}{a-1}\\
x_{3}^{2} & =\frac{\left(s_{1}-a\right)\left(s_{2}-a\right)s_{3}}{a\left(a-1\right)} &  & x_{4}^{2}=\frac{\left(s_{1}-a\right)\left(s_{2}-a\right)\left(1-s_{3}\right)}{a\left(a-1\right)}
\end{aligned}
\label{eq:prolate coord def-1}
\end{equation}
where $0\le s_{1},s_{2}\le1\le s_{3}\le a$.
A possible St\"{a}ckel matrix is given by 
\begin{equation}
\Phi_{\text{obl}}=\frac{1}{4}\left(\begin{array}{ccc}
-\frac{1}{\left(s_{1}-1\right)\left(s_{1}-a\right)} & -\frac{1}{\left(s_{1}-1\right)s_{1}\left(s_{1}-a\right)} & \frac{(1-a)a}{\left(s_{1}-1\right)s_{1}\left(s_{1}-a\right){}^{2}}\\
-\frac{1}{\left(s_{2}-1\right)\left(s_{2}-a\right)} & -\frac{1}{\left(s_{2}-1\right)s_{2}\left(s_{2}-a\right)} & \frac{(1-a)a}{\left(s_{2}-1\right)s_{2}\left(s_{2}-a\right){}^{2}}\\
0 & 0 & -\frac{1}{\left(s_{3}-1\right)s_{3}}
\end{array}\right).\label{eq: stackel oblate}
\end{equation}

From \cite{Nguyen2023} and similar to the prolate
system, we have the following.
\begin{lem}
Separating the Hamilton Jacobi equation in oblate
coordinates using the Stäckel matrix \eqref{eq: stackel oblate} and subsequently reducing by $H$ gives a two degree of
freedom integrable system on $S^{2}\times S^{2}$ with integrals 
\begin{align*}
G_{\text{obl}}= a\ell_{12}^{2}+\ell_{13}^{2}+\ell_{14}^{2} &  & M=\ell_{34}^{2}
\end{align*}
and Poisson structure $B_{\bm{X},\bm{Y}}$ as
given in (\ref{eq:BXY}).
\end{lem}
Using (\ref{eq: stackel oblate}) and a computation similar to the prolate case we have the following
Lemma.
\begin{lem}
Separating the Schrödinger equation (\ref{eq:Schrodinger})
in oblate coordinates (\ref{eq:prolate coord def-1}) gives the following
ODEs
\begin{subequations}
\begin{align}
\psi_{j}^{''}+\frac{1}{2}\left(\frac{1}{s_{j}}+\frac{1}{s_{j}-1}+\frac{2}{s_{j}-a}\right)\psi_{j}^{'}+\frac{-Es_{j}^{2}+\left(\lambda+aE-(a-1)m^{2}\right)s_{j}-a\lambda}{4s_{j}(s_{j}-1)(s_{j}-a)^{2}}\psi_{j} & =0 & \text{for \ensuremath{j=1,2}}\label{eq:OblateEqn1}\\
\psi_{j}^{''}+\frac{1}{2}\left(\frac{1}{s_{j}}+\frac{1}{s_{j}-1}\right)\psi_{j}^{'}+\frac{m^{2}}{4s_{j}(1-s_{j})}\psi_{j} & =0 & \text{for \ensuremath{j=3}}\label{eq:OblateEqn2}
\end{align}
\end{subequations}
where $(m,\lambda)$ are spectral parameters.
\end{lem}
Again, the equation for $\psi_{3}$ in \eqref{eq:OblateEqn2} can be converted
to the trivial equation (\ref{eq:trivial equation}) using the same
transformation as for prolate. This yields a discrete spectral
parameter $m\in\mathbb{Z}$ which represents quantised angular momentum
in the $(x_{3},x_{4})$ plane with solutions $e^{im\phi}$.
As with the prolate system, the separated ODEs can also be obtained by degenerating the generalised Lam\'{e} equation.
\begin{lem}
The separated equations for the oblate system (\ref{eq:OblateEqn1}) and (\ref{eq:OblateEqn2})
smoothly degenerate from those of the ellipsoidal system.
\end{lem}
\begin{proof}
The proof is identical to that of prolate but uses
the transformation 
\begin{align*}
(e_{4},s_{3},p_{3}) & =(e_{3}+\epsilon,e_{3}+\epsilon\tilde{s}_{3},\frac{\tilde{p}_{3}}{\epsilon})
\end{align*}
and taking the limit as $\epsilon\to 0$. For more detail, see \cite{Nguyen2023}.
\end{proof}
Equation (\ref{eq:OblateEqn1}) is an example of a Heun equation. 
\begin{lem}
\label{Lemma2}By the change of
dependent variable $\psi_j=(z-a)^{|m|/2}W_j$, (\ref{eq:OblateEqn1}) can
be written in the form of (\ref{eq:GenHeun}) where $\gamma=\delta=\frac{1}{2},\epsilon=1+|m|,q=\frac{1}{4}\left(|m|-\lambda\right)$
and $\alpha,\beta$ as in (\ref{eq:alphabeta}). In particular, we have 
\begin{equation}
    W_j^{''} + \left(\frac{1}{2z} + \frac{1}{2(z-1)} + \frac{1 + |m|}{z-a}\right)W_j^{'} + \frac{(-E + |m|(|m|+2))z- (|m| - \lambda)}{4z(z-1)(z-a)}W_j. \label{eq:oblateHeunSpecificBase}
\end{equation}
\end{lem}
Let $Hp_d(z)$ be a Heun polynomial solution to \eqref{eq:oblateHeunSpecificBase} of degree $d$ and let $z_{1},\dots,z_{d}$ be its roots. Many of the results for prolate apply for the oblate
system which we state in the following Lemma. 
\begin{lem}
    The combined solution to \eqref{eq:OblateEqn1} and \eqref{eq:OblateEqn2} is the product 
    \begin{equation}
    \psi_{1}\psi_{2}\psi_{3}=(a-s_{1})^{m/2}(a-s_{2})^{m/2}Hp_{d}(s_{1})Hp_{d}(s_{2})e^{im\phi}.\label{eq:HeunOblateSol}
    \end{equation}
    Expressed in the original Cartesian coordinates using (\ref{eq:prolate coord def-1}), \eqref{eq:HeunOblateSol} is a homogeneous polynomial of total degree $\tilde{D} = 2d + |m|$ given by
\begin{equation}
\Phi_{\tilde{D}}=\left(a(a-1)\right)^{m/2}\left(x_{3}+i\text{\normalfont sign}(m)x_{4}\right)^{|m|}\prod_{k=1}^{d}z_{k}(z_{k}-1)\left(r^2-\frac{x_{1}^{2}}{z_{k}}a-\frac{a-1}{z_k-1}x_{2}^{2}\right).\label{eq:oblatecombinedwf}
\end{equation}
\end{lem}
\begin{proof}
The proof follows the same logic as used for the prolate system. To obtain (\ref{eq:oblatecombinedwf}),
we note the use of the following identity:
\[
(\theta-s_{1})(\theta-s_{2})=\theta(\theta-1)\left(r^2-\frac{x_{1}^{2}}{\theta}a-\frac{a-1}{\theta-1}x_{2}^{2}\right),
\]
which holds since both sides are monic quadratic polynomials in $\theta$
that vanish at $\theta=s_{1},s_{2}$. Using the definition of oblate coordinates in (\ref{eq:prolate coord def-1})
and observing that $e^{im\phi}=\left(\frac{x_{3}+ix_{4}}{\sqrt{x_{3}^{2}+x_{4}^{2}}}\right)^{m}$,
the result follows. 
\end{proof}
As with the prolate system, we only consider the 4 discrete symmetry classes about the $x_{1}$ and $x_{2}$ planes. Let $\bm{\mu} = (\mu_1, \mu_2)$ where $\mu_i \in \{0, 1\}$ and $\mu_1$ ($\mu_2$) being $1$ ($0$) denotes a solution odd (even) about the $x_1$ ($x_2$) axis. To investigate these symmetries, consider the change of dependent variable in \eqref{eq:oblateHeunSpecificBase} given by $\phi_1 = s_{1}^{\mu_1/2}(a-s_1)^{\mu_{1}/2}W_1$ and $\phi_2 = s_{2}^{\mu_{2}/2}(a-s_2)^{\mu_2/2}W_2$
which corresponds to multiplication by $x_{1}^{\mu_{1}}x_{2}^{\mu_{2}}$. This change of variables leads to a Heun equation \eqref{eq:GenHeun} with parameters given in Table \ref{tab: HeunOblParams}.
As in the prolate example, denote by $Hp_{d}^{\bm{\mu}}(z)$ solutions to \eqref{eq:GenHeun} for the symmetry class $\bm{\mu}$ and let the product $Hp_{d}^{\bm{\mu}}(s_1)Hp_{d}^{\bm{\mu}}(s_2)e^{im\phi}$, when converted back into Cartesian coordinates, be given by $\Phi_{\tilde{D}}^{\bm{\mu}}(\bm{x})$. Set \begin{equation}
    \Psi_{D}^{\bm{\mu}}(\bm{x})\coloneqq x_{1}^{\mu_1}x_{2}^{\mu_2}\Phi_{\tilde{D}}^{\bm{\mu}}(\bm{x})
\end{equation}
where $D= \tilde{D} + \sum_{i=1}^{2}\mu_i$. Again, \eqref{eq:oblatecombinedwf} is a special case with $\bm{\mu} = (0,0)$. Following the same argument used in prolate, we have the following Lemma.
\begin{lem}
    The $\Psi_{D}^{\bm{\mu}}(\bm{x})$ are $D$ degree homogeneous, harmonic polynomials and the energy eigenvalue is given by $E=D(D+2)$.  
\end{lem}
\begin{table}
\begin{centering}
\textcolor{black}{}%
\begin{tabular}{|c|c|c|c|c|c|c|}
\hline 
\textcolor{black}{$(\mu_1,\mu_2)$} & \textcolor{black}{$\tilde{\alpha}$} & \textcolor{black}{$\tilde{\beta}$} & \textcolor{black}{$\tilde{\gamma}$} & \textcolor{black}{$\tilde{\delta}$} & \textcolor{black}{$\tilde{\epsilon}$} & \textcolor{black}{$\tilde{q}$}\tabularnewline
\hline 
\textcolor{black}{$(0,0)$} & \textcolor{black}{$\alpha$} & \textcolor{black}{$\beta$} & \textcolor{black}{$\gamma$} & \textcolor{black}{$\delta$} & \textcolor{black}{$\epsilon$} & \textcolor{black}{$q$}\tabularnewline
\hline 
\textcolor{black}{$(1,0)$} & \textcolor{black}{$\alpha+1-\gamma$} & \textcolor{black}{$\beta+1-\gamma$} & \textcolor{black}{$2-\gamma$} & \textcolor{black}{$\delta$} & \textcolor{black}{$\epsilon$} & \textcolor{black}{$q+(1-\gamma)(a\delta+\epsilon)$}\tabularnewline
\hline 
\textcolor{black}{$(0,1)$} & \textcolor{black}{$\alpha+1-\delta$} & \textcolor{black}{$\beta+1-\delta$} & \textcolor{black}{$\gamma$} & \textcolor{black}{$2-\delta$} & \textcolor{black}{$\epsilon$} & \textcolor{black}{$q+(1-\gamma)(a\delta+\epsilon)$}\tabularnewline
\hline 
\textcolor{black}{$(1,1)$} & \textcolor{black}{$\alpha+2-\gamma-\delta$} & \textcolor{black}{$\beta+2-\gamma-\delta$} & \textcolor{black}{$2-\gamma$} & \textcolor{black}{$2-\delta$} & \textcolor{black}{$\epsilon$} & \textcolor{black}{$q+2a(1-\gamma)+2(1-\gamma)\epsilon$}\tabularnewline
\hline 
\end{tabular}
\par\end{centering}
\textcolor{black}{\caption{Parameters of (\ref{eq:GenHeun}) for the various symmetry classes
of the oblate system.\label{tab: HeunOblParams}}
}
\end{table}
For a given energy the total number of eigenstates
is computed using the same methodology as for the prolate case.
\begin{lem}
For a fixed value of $D$, the total number of eigenstates is given by $N=(D+1)^{2}$. The number
of states for each symmetry class given in Table \ref{tab: HeunOblParams}
is the same as that shown in Table \ref{tab:Pro Num States}.
\end{lem}
Using the three term recurrence \eqref{eq:Heun3term} and the same methodology for prolate, we find the joint spectrum. In Fig. \ref{fig:OblateSpecandActions} a) we show the joint spectrum $(m\hbar,\lambda\hbar^{2})$
for $D=20$ and $a=2.4$. As in the prolate system, we have $\hbar=\frac{1}{D+1}$
and a total of $441$ total eigenstates. The blue, red, orange and
cyan points correspond to the $(0,0),(1,0),(0,1)$ and $(1,1)$ symmetry classes
respectively.

We also show the classical outline of the momentum
map in black. Unlike the prolate system, this is not a semi-toric
system. One interesting feature of the joint spectrum is the existence
of period doubling. A quantum consequence of this
is the near-degeneracy in the upper chamber. 
\begin{figure}[H]
\begin{centering}
\includegraphics[width=7.5cm,height=7.5cm,keepaspectratio]{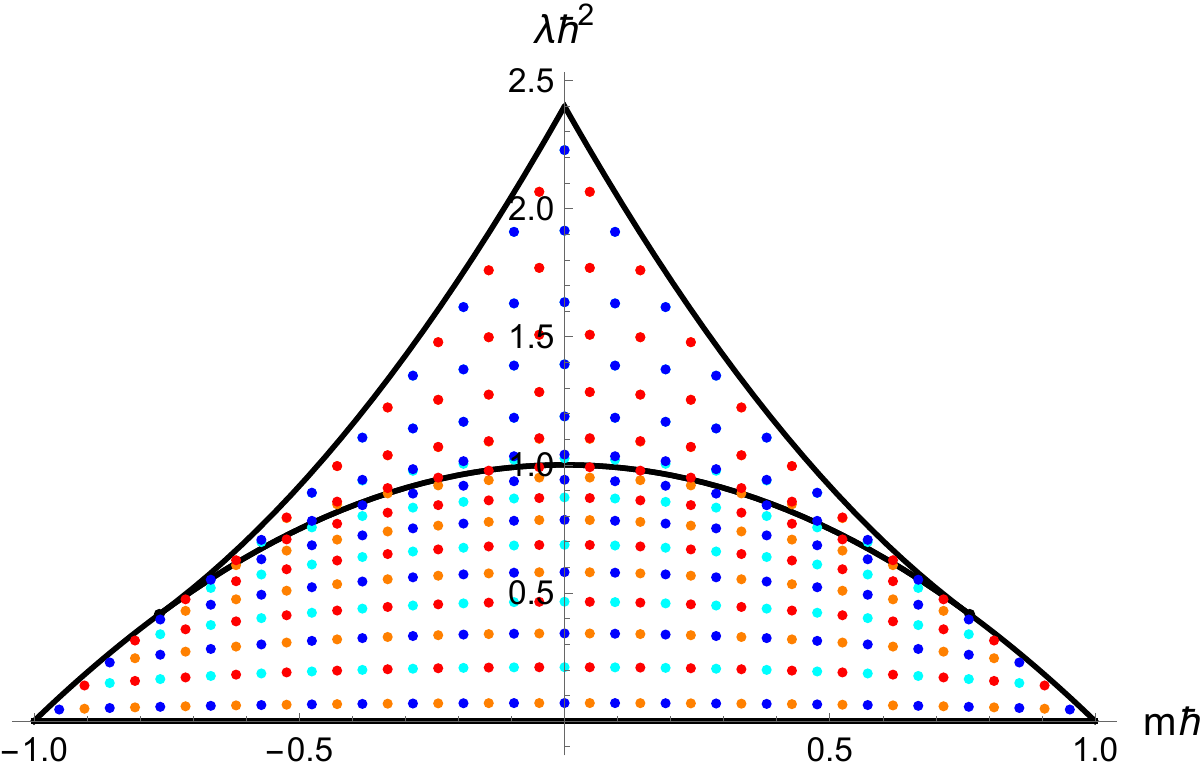}\includegraphics[width=7.5cm,height=7.5cm,keepaspectratio]{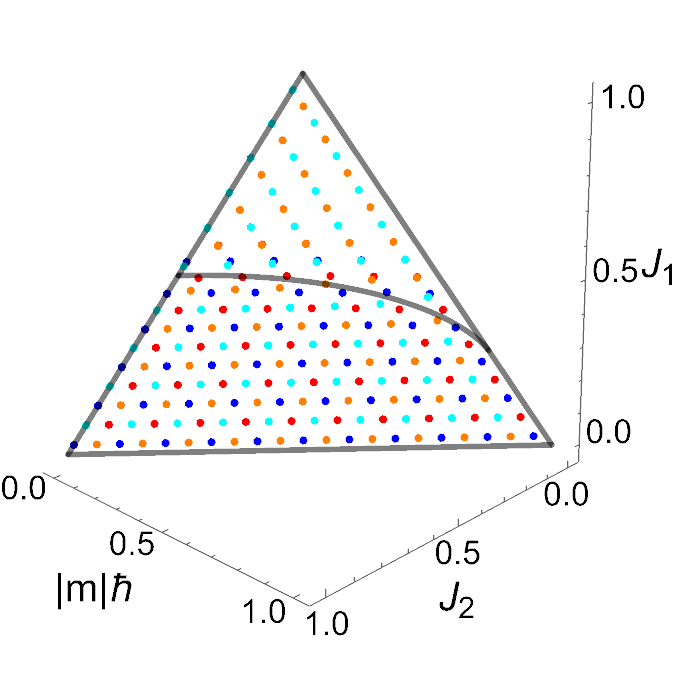}
\par\end{centering}
\caption{a) Joint spectrum of the oblate system $(m\hbar,\lambda\hbar^2)$ for $D=20$ and $a=2.4$ showing
a total of 441 states. b) Corresponding action map. Note
that light blue (cyan) and orange states are hidden by the nearby
degenerate red and blue states in the upper chamber. \label{fig:OblateSpecandActions}}
\end{figure}
To compute the actions, we use the same methodology
as for the prolate system. Doing so gives the following three continuous
actions:
\begin{align*}
J_{1} & =\frac{2}{\pi}\int_{0}^{\min(r_{1},1)}p_{1}ds & J_{2} & =\frac{2}{\pi}\int_{\max(r_{1},1)}^{\min(r_{2},a)}p_{2}ds & J_{3} & =|\ell_{34}|
\end{align*}
where 
\[
p_{i}^{2}=\frac{-\tilde{E}s_{i}^{2}+(\tilde{E}a+g-(a-1)m^{2})s_{i}-ag}{4s_{i}(s_{i}-1)(s_{i}-a)^{2}}
\]
 for $i=1,2$, $r_{2}\ge1$ and $0\le r_{1}\le r_{2}\le a.$ The action
map corresponding to the joint spectrum in Fig. \ref{fig:OblateSpecandActions} a) is shown in b). Again,
we note that $J_{1}+J_{2}+J_{3}=1$. 

\subsection{Lam\'{e} Coordinates}\label{sec:LameS3}

The Lam\'{e} coordinates are unique compared to the other singly degenerate coordinate systems (oblate, prolate)
since they are\textcolor{black}{{} an extension
of ellipsoidal coordinates from $S^{2}$ onto $S^{3}$. Consequently,
while prolate and oblate coordinates have symmetry group $SO(2)$,
the Lam\'{e} coordinates have the larger symmetry group $SO(3)$.}

As with oblate, there are two possible (yet equivalent)
ways to define Lam\'{e} coordinates. 
Either the largest three $e_i$ become equal, or the smallest three.
We choose the former and define
Lam\'{e} coordinates by 
\begin{align}
x_{1}^{2} & =s_{1} & x_{2}^{2} & =\frac{(1-s_{1})(s_{2}-f_{1})(s_{3}-f_{1})}{(f_{2}-f_{1})(f_{3}-f_{1})}\nonumber \\
x_{3}^{2} & =\frac{(s_{1}-1)(f_{2}-s_{2})(f_{2}-s_{3})}{(f_{1}-f_{2})(f_{2}-f_{3})} & x_{4}^{2} & =\frac{(s_{1}-1)(f_{3}-s_{2})(f_{3}-s_{3})}{(f_{2}-f_{3})(f_{3}-f_{1})}\label{eq:Lame def}
\end{align}
where $0\le s_{1}\le1$ and $0\le f_{1}\le s_{2}\le f_{2}\le s_{3}\le f_{3}$ with real parameters $f_i$.

A possible St\"{a}ckel matrix is given by 
\begin{equation}
\sigma_{\text{Lamé}}=\frac{1}{4}\begin{pmatrix}\frac{1}{s_{1}(1-s_{1})} & \frac{-1}{s_{1}(1-s_{1})^{2}} & 0\\
0 & \frac{1}{(f_{3}-s_{2})(s_{2}-f_{2})} & \frac{1}{(f_{3}-s_{2})(s_{2}-f_{2})(s_{2}-f_{1})}\\
0 & \frac{1}{(f_{3}-s_{3})(s_{3}-f_{2})} & \frac{1}{(f_{3}-s_{3})(s_{3}-f_{2})(s_{3}-f_{1})}
\end{pmatrix}\label{eq:StackelLame}
\end{equation}
 with corresponding integrals obtained by separation as
\begin{align*}
F & =\ell_{12}^{2}+\ell_{13}^{2}+\ell_{14}^{2}, & G & =f_{1}\ell_{34}^{2}+f_{2}\ell_{24}^{2}+f_{3}\ell_{23}^{2}.
\end{align*}
As in the other coordinates, reduction by $H$ gives a two degree of freedom integrable system $(F,G)$ on $S^2\times S^2$ with Poisson structure $B_{\bm{X}, \bm{Y}}$.
The integral $F$ is a result of the $SO(3)$
symmetry. Using (\ref{eq:StackelLame}) and the same methodology used in the previous systems, we obtain
the following separated ODEs
\begin{subequations}
\begin{align}
\psi_{1}^{''}+\frac{1}{2}\left(\frac{1}{s_{1}}+\frac{3}{s_{1}-1}\right)\psi_{1}^{'}+\frac{f-Es_{1}}{4s_{1}(1-s_{1})^{2}}\psi_{1} & =0\label{eq:raw equations}\\
\psi_{k}^{''}+\frac{1}{2}\left(\frac{1}{z-f_{1}}+\frac{1}{z-f_{2}}+\frac{1}{z-f_{3}}\right)\psi_{k}^{'}+\frac{(f-E)z+g}{4(z-f_{1})(z-f_{2})(z-f_{3})}\psi_{k} & =0 & k=2,3\label{eq:raw equations2}
\end{align}
\end{subequations}
The separated equations can also be obtained by smoothly degenerating the generalised Lam\'{e} equation.
\begin{lem}
    Equations \eqref{eq:raw equations} and \eqref{eq:raw equations2} arise from smoothly degenerating the separated generalised Lam\'{e} equations in \eqref{eq:Gen Lame}.
\end{lem}

\begin{proof}
    From \cite{Nguyen2023}, it is known that the transformation 
    \begin{align}
        \left(e_1,e_2, e_3, e_4\right) &= \left(-\frac{1}{\epsilon}, f_1, f_2, f_3\right)\label{eq:LameCoordTrans}\\
        \left(s_1, p_1\right) &= \left(f_2 - \frac{s_1}{\epsilon}, \epsilon p_1 \right)\label{eq:LameCoordTrans2}
    \end{align}
    defines a canonical transformation that smoothly degenerates the classical ellipsoidal system into that for Lam\'{e}. Further, applying \eqref{eq:LameCoordTrans} and taking the limit as $\epsilon\to 0$ gives 
    \begin{equation}
        (\eta_1, \eta_2) = \left(-\frac{2H - F}{\epsilon}, -\frac{G}{\epsilon}\right).\label{eq:IntegralsLameDegen}
    \end{equation}
    Substituting \eqref{eq:LameCoordTrans}, \eqref{eq:LameCoordTrans2} and \eqref{eq:IntegralsLameDegen} into \eqref{eq:Gen Lame} and and considering dominating terms of $\epsilon$ gives the degenerated ODEs \eqref{eq:LameCoordTrans} and \eqref{eq:LameCoordTrans2}.
\end{proof}
Equation \eqref{eq:raw equations} has regular singularities
at $s_{1}=0,1,\infty$ meaning that it is of the hypergeometric type. We
observe that the form as written in (\ref{eq:raw equations}) is not
canonical. The roots of the indicial equation at $s_{1}=1$ are
\begin{equation}
r_\pm=\frac{1}{4}\left(-1\pm\sqrt{1+4E-4f}\right),\label{eq:rpm}
\end{equation}
where $r_{-}<0<r_{+}$. 
To change (\ref{eq:raw equations}) into canonical
form, we make the change of dependent variable $\psi_{1}\to y(1-s_{1})^{r_{+}}$
followed by the change of independent variable $s_{1}=x_{1}^2$, yielding the well known Gegenbauer equation
\begin{equation}
    (1-x_1^2)y^{''}-(2u+1)x_1y^{'}+n(n+2u)y=0 \label{eq:GegenbauerEqn}
\end{equation}
where $u=\frac{1}{2} (1+ \sqrt{1+4E-4f})$ and 
\begin{equation}
    n = -\frac{1}{2}\left(\sqrt{4E-4f+1}-\sqrt{4E+4}+1\right).\label{eq:JacobiPolynomial condition}
\end{equation}
From \cite{AS, Morse53}, it is known that polynomial solutions to \eqref{eq:GegenbauerEqn} occur when $n$ is an integer. These are known as the Gegenbauer polynomials $C_{n}^{u}(x_1)$ where $n$ denotes the degree of the polynomial.

The equation in the $s_{2},s_{3}$ variables is
known as the Lam\'{e} equation, a special case of the Heun equation. As with the previous degenerate cases, we simplify by mapping $(f_{1},f_{2},f_{3})$ to $(0,1,a)$
where $a>1$. The parameters $(f_1, f_2, f_3)$ in \eqref{eq:raw equations2} can be mapped to $(0,1,a)$ by a M\"obius transformation, then the Riemann symbol is given by
\[
\mathcal{S}_{2,3}=\left(\begin{aligned}0 &  & 1 &  & a &  & \infty\\
0 &  & 0 &  & 0 &  & \alpha & ;z\\
1/2 &  & 1/2 &  & 1/2 &  & \beta
\end{aligned}
\right)
\]
where $\alpha, \beta=\frac{1}{4}\left(1\pm\sqrt{1+4E-4f}\right)$
and $q=g/4$. The Lam\'{e} equation is typically
written as follows \cite{Bateman:100233}:
\begin{equation}
W^{''}+\frac{1}{2}\left(\frac{1}{\xi}+\frac{1}{\xi-1}+\frac{1}{\xi-k^{-2}}\right)W^{'}+\frac{bk^{-2}-\nu(\nu+1)\xi}{4\xi(\xi-1)(\xi-k^{-2})}W=0\label{eq:Lame equation}
\end{equation}
 where $k=\frac{1}{\sqrt{a}}$, $b=g/a$ and $\nu=\frac{1}{2}\left(\sqrt{4E-4f+1}-1\right)=2\beta-1$
with $\nu\ge-1/2$ and $0<k<1$. 
Polynomial solutions of the Lam\'{e} equation are known as Lam\'{e} polynomials \cite{Arscott2014} and we
denote them as $Lp_{d}(z)$ where the subscript $d$ signifies the solution's
degree. Note that $\nu$ being a non negative integer is a necessary condition for polynomial
solutions. 
\begin{lem}
\label{First Lame Energy Sol}Solutions to (\ref{eq:raw equations2}) are polynomial iff
$E=\tilde{D}(\tilde{D}+2)$ where $\tilde{D}=2d+n$ is an integer.
\end{lem}
\begin{proof}
From \cite{Volkmer1999} we know that polynomial
solutions are only obtainable for (\ref{eq:raw equations2}) if 
\begin{equation}
\frac{f-E}{4}=-d(d+\frac{1}{2}).\label{eq:fErelation}
\end{equation}
Combining (\ref{eq:fErelation})
and (\ref{eq:JacobiPolynomial condition}) gives the result.
\end{proof}
Using \eqref{eq:fErelation}, we know that $r_{+}$ in \eqref{eq:rpm} simplifies
to $d$, $\nu=2d, \alpha = -d, \beta=2d+\frac{1}{2}$ and $u=2d+1$. Denote
the full solution $\psi=\psi_{1}\psi_{2}\psi_{3}$ to \eqref{eq:raw equations} and \eqref{eq:raw equations2} as
\begin{equation}
\psi=(1-s_{1})^{d}C_{n}^{2d+1}(x_1)Lp_{d}(s_{2})Lp_{d}(s_{3}).\label{eq:PsiD}
\end{equation}
We have the following Lemma.
\begin{lem}\label{lem:HomogenProof}
Expressed in the original Cartesian coordinates,
$\psi$ is a homogeneous polynomial of degree $\tilde{D}$. Specifically,
we can rewrite (\ref{eq:PsiD}) as 
\begin{equation}
\Phi_{\tilde{D}}(\bm{x})=r^{n}C_{n}^{2d+1}\left(\frac{x_{1}}{r}\right)\prod_{i=1}^{3}\prod_{k=1}^{d}(z_{k}-f_{i})\prod_{k=1}^{d}\sum_{i=1}^{3}\frac{x_{i+1}^{2}}{z_{k}-f_{i}}\label{eq:LameinCartesian}
\end{equation}
where the $z_{k}$ are the roots of the Lam\'{e} polynomials $Lp_{d}$, $(f_1,f_2,f_3)=(0,1,a)$ and $r^2=\sum_{i}x_{i}^{2}$.
\end{lem}
\begin{proof}
From the definition of the Lam\'{e} coordinates
(\ref{eq:Lame def}), it is easily seen that 
\[
(z_{k}-s_{2})(z_{k}-s_{3})=\frac{(z_{k}-f_{1})(z_{k}-f_{2})(z_{k}-f_{3})}{1-x_{1}^{2}}\left(\frac{x_{2}^{2}}{z_{k}-f_{1}}+\frac{x_{3}^{2}}{z_{k}-f_{2}}+\frac{x_{4}^{2}}{z_{k}-f_{3}}\right).
\]
Taking the product over all $d$ roots and using $r^2=1$ along with $s_1=x_1^2$ gives (\ref{eq:LameinCartesian}).
To prove homogeneity, we note that 
\begin{equation}
    C^{2d+1}_{n}(x_1)=r^{n}C^{2d+1}_{n}(\frac{x_1}{r})
\end{equation}
and so the argument of the Gegenbauer polynomial will be of degree 0. We observe that $r^{n}$ will be of degree $n$ in the $x_i$ and $C^{2d+1}_n(\frac{x_1}{r})$ will only consist of terms whose powers have the same parity as $n$. Thus, we have shown that \eqref{eq:LameinCartesian} is homogeneous and polynomial.
\end{proof}
As with the previous systems, let $\bm{\mu} = (\mu_2,\mu_3,\mu_4)$ where $\mu_i\in \{0,1\}$ denote a symmetry class where $\mu_i=0$ represent the wavefunction odd about the $x_i$ axis and even otherwise. Consider the change of dependent variable 
\begin{align}
    \phi_j & = \prod_{i=1}^3 (s_j - f_i)^{\mu_{i+1}/2}\psi_j
\end{align}
where $j=2,3$ in \eqref{eq:raw equations2}. This is equivalent to multiplying the original solution by $x_{2}^{\mu_2}x_{3}^{\mu_3}x_{4}^{\mu_4}$. Note that performing this change of variables in the Lam\'{e} equation \eqref{eq:raw equations2} results in a Heun equation of the form \eqref{eq:GenHeun} with parameters for each symmetry class $\bm{\mu}$ as shown in Table \ref{tab:LameTable}. 

For the Gegenbauer equation \eqref{eq:GegenbauerEqn}, we note that the independent variable is $x_1$ and so parity about this axis is determined simply by the parity of the Gegenbauer polynomial $C_n^u(x_1)$, which is determined by $n$. We let $\mu_1=1$ if $n$ is odd and $\mu_1=0$ otherwise. 

For a given symmetry class $\bm{\mu}$, denote by $Lp_{d}^{\bm{\mu}}(z)$ solutions to \eqref{eq:GenHeun} with parameters given in Table \ref{tab:LameTable}. 
We consider the product $\Phi_{\tilde{D}}^{\bm{\mu}}(\bm{x}) = C_{n}^{2d+1+U}(x_1)Lp_{d}^{\bm{\mu}}(s_2)Lp_{d}^{\bm{\mu}}(s_3)$ where $U=\sum_{i=1}^3 \mu_{i+1}$. Finally, we set 
\begin{equation}
    \Psi_{D}^{\bm{\mu}}(\bm{x}) \coloneqq x_{2}^{\mu_2}x_{3}^{\mu_3}x_{4}^{\mu_4}\Phi_{\tilde{D}}^{\bm{\mu}}(\bm{x}) \label{eq: GenLameCoordSol}
\end{equation}
where $D = n + \tilde{D} + U$. As before, \eqref{eq:LameinCartesian} is a special case of \eqref{eq: GenLameCoordSol} with $\bm{\mu} = (0,0,0)$.
\begin{lem}
    The $\Psi_{D}^{\bm{\mu}}(\bm{x})$ are homogeneous polynomial eigenfunctions of the original Laplacian \eqref{eq:Schrodinger} and the energy is given by $E = D(D+2)$.
\end{lem}
\begin{proof}
    Proving $\Psi_{D}^{\bm{\mu}}(\bm{x})$ is an eigenfunction follows the same argument as the previous cases. To compute the energy for a given symmetry class $\bm{\mu}$, we combine the parameters of Table \ref{tab:LameTable} with the results of \cite{Volkmer1999} to get a more general form of \eqref{eq:fErelation}, namely 
    \begin{equation}
        \frac{f-E + U(U+1)}{4}= -d(d+1/2 + U). \label{eq:generalf-e}
    \end{equation}
    Combining \eqref{eq:generalf-e} and \eqref{eq:JacobiPolynomial condition} gives the result for $E$. Similarly, we have that $u=2d+1+U$.
\end{proof}
\begin{table}[H]
\centering
\begin{tabularx}{\textwidth}{|p{1.7cm}|p{2.3cm}|p{2.3cm}|p{0.8cm}|p{0.8cm}|p{0.8cm}|p{4.8cm}|}
\hline 
$\bm{\mu}$& $\tilde{\alpha}$ & $\tilde{\beta}$ & $\tilde{\gamma}$ & $\tilde{\delta}$ & $\tilde{\epsilon}$ & $\tilde{q}$ \\
\hline 
$(0,0,0)$ & $\alpha$ & $\beta$ & $\gamma$ & $\delta$ & $\epsilon$ & $q$ \\
\hline 
$(1,0,0)$ & $\alpha+1-\gamma$ & $\beta+1-\gamma$ & $2-\gamma$ & $\delta$ & $\epsilon$ & $q+(1-\gamma)(a\delta+\epsilon)$ \\ 
\hline 
$(0,1,0)$ & $\alpha+1-\delta$ & $\beta+1-\delta$ & $\gamma$ & $2-\delta$ & $\epsilon$ & $q-a\gamma(\delta-1)$ \\
\hline 
$(0,0,1)$ & $\alpha+1-\epsilon$ & $\beta+1-\epsilon$ & $\gamma$ & $\delta$ & $2-\epsilon$ & $q+\gamma(1-\epsilon)$ \\
\hline 
$(1,1,0)$ & $\alpha-\gamma-\delta+2$ & $\beta-\gamma-\delta+2$ & $2-\gamma$ & $2-\delta$ & $\epsilon$ & $q-a(\gamma+\delta-2)-\gamma\epsilon+\epsilon$ \\
\hline 
$(1,0,1)$ & $\alpha-\gamma-\epsilon+2$ & $\beta-\gamma-\epsilon+2$ & $2-\gamma$ & $\delta$ & $2-\epsilon$ & $q-\gamma(a\delta+1)+a\delta-\epsilon+2$ \\
\hline 
$(0,1,1)$ & $\alpha-\delta-\epsilon+2$ & $\beta-\delta-\epsilon+2$ & $\gamma$ & $2-\delta$ & $2-\epsilon$ & $q+\gamma(-a\delta+a-\epsilon+1)$ \\
\hline 
$(1,1,1)$ & $\alpha-\gamma-\delta-\epsilon+3$ & $\beta-\gamma-\delta-\epsilon+3$ & $2-\gamma$ & $2-\delta$ & $2-\epsilon$ & $q-a(\gamma+\delta-2)-\gamma-\epsilon+2$ \\
\hline 
\end{tabularx}
\caption{Symmetry class $\bm{\mu}$ and corresponding Heun (Lam\'{e}) equation parameters for which the resulting polynomial is an eigenfunction. \label{tab:LameTable}}
\end{table}
From the above, it is clear that for a fixed energy $E$, only 8 of the 16
discrete symmetry classes for the Lam\'{e} system can be present.
If $E$ is even, then the symmetry classes $\mathcal{S}_{E}$ given in \eqref{eq:SE_Defn} are present. Conversely, if $E$ is odd, we have the remaining
8 symmetry classes $\mathcal{S}_{O}$ from \eqref{eq:SO_Defn}. Note that while we have omitted $\mu_1$ from the definition of $\bm{\mu}$ in our discussion of the Lam\'{e} systems so far, this symmetry is accounted for in the parity of $n$ which is set by the value of $D$. When addressing all symmetry classes below, we will write $\tilde{\bm{\mu}}=(\mu_1;\mu_2,\mu_3,\mu_4)$ to highlight this fact.

We have the following observation.
\begin{lem}
    Chebyshev polynomials of the second kind are the eigenfunction corresponding to the joint spectrum point $(f,g) = (1-\hbar^2,0)$.
\end{lem}
\begin{proof}
    For $\bm{\mu} = (0,0,0)$, we have from \eqref{eq:generalf-e} the relation $f = E-4d(d+1/2)$. When $d=0$ one obtains $f=E$ which gives $f=1-\hbar^2$ (after normalisation by $\hbar$) and $g=0$. This gives Chebyshev polynomials of the second kind $U_{n}(x_1)$ since $U_{n}(x_1)=C_{n}^{1}(x_1).$ Note that in the limit $\hbar\to 0$ this joint spectrum point becomes the degenerate point $(1,0)$ for the classical system.
\end{proof}

The Lam\'{e} equation (\ref{eq:Lame equation})
can be transformed by the change of independent variable $\xi=\text{sn}^{2}(z,k)$
into the following form:
\begin{equation}
W^{''}+(b-\nu(\nu+1)k^{2}\text{sn}^{2}(z,k))W=0.\label{eq:LameJacobiForm}
\end{equation}
This form of the Lam\'{e} equation has regular
singular points at $2pK+(2q+1)iK^{'}$ where $K,K^{'}$ are the complete
elliptic integrals of the first kind with moduli $k$ and $\sqrt{1-k^{2}}$
respectively and $p,q\in\mathbb{Z}$.

From \cite{Arscott2014}, it is known that there are $8$ polynomial solutions to (\ref{eq:LameJacobiForm}) if $\nu$ is a non-negative integer. These are shown in Table \ref{tab:LameJacobiTable} where $m=0,1,\dots,\nu$ and $Ec^m_d,Es^m_d$ are known as the Lam\'{e} functions \cite{erdelyi1955higher}. Since $\xi = s_i = \text{sn}^2(z,k)$, we are able to connect these eigenfunctions to $Lp_{d}^{\bm{\mu}}(z)$ as given in Table \ref{tab:LameTable}.

The total number of states for each symmetry class
is shown in Table \ref{tab:Lame Symmetry State Table} below. Here we use the notation $(\mu_1; \mu)$ to represent a state with parity given by $\mu_1$ about the $x_1$ axis and odd about one remaining axis given by $\mu$. E.g. $(1;\mu)$ represents the classes $(1;0,0,1), (1;0,1,0), (1;1,0,0)$. Similarly, we have $(\mu_1;\mu,\nu)$ and so forth. The number of states are computed using similar methods to those shown in the prolate case.

Summing the relevant entries for a given parity of the energy shows
that the total number of states for a given energy $E=D(D+2)$
is given by $(D+1)^{2}$. To compute the joint spectrum for a given
$E$, we use the parameters shown in Table \ref{tab:LameTable}
and the recurrence relationship given in Lemma \ref{LemmaProlate}.
An example for even and odd energies is shown in Fig. \ref{fig:LameEvenOdd spectrum}.
As with the other cases such as ellipsoidal and prolate, we have $\hbar=\frac{1}{D+1}$. As with the prolate and oblate systems, since the matrices involved are not self adjoint, larger values of $D$ lead to numerical instability.
\begin{table}
\begin{centering}
\textcolor{black}{}%
\begin{tabular}{|c|c|}
\hline 
\textcolor{black}{Lam\'{e} Polynomial} & \textcolor{black}{Corresponding Symmetry $\bm{\mu}$}\tabularnewline
\hline 
\textcolor{black}{$Ec_{2d}^{2m}(z,k^{2})$} & \textcolor{black}{$(0,0,0)$}\tabularnewline
\hline 
\textcolor{black}{$Ec_{2d+1}^{2m+1}(z,k^{2})$} & \textcolor{black}{$(1,0,0)$}\tabularnewline
\hline 
\textcolor{black}{$Es_{2d+1}^{2m+1}(z,k^{2})$} & \textcolor{black}{$(0,1,0)$}\tabularnewline
\hline 
\textcolor{black}{$Ec_{2d+1}^{2m}(z,k^{2})$} & \textcolor{black}{$(0,0,1)$}\tabularnewline
\hline 
\textcolor{black}{$Es_{2d+1}^{2m+2}(z,k^{2})$} & \textcolor{black}{$(1,1,0)$}\tabularnewline
\hline 
\textcolor{black}{$Ec_{2d+2}^{2m+1}(z,k^{2})$} & \textcolor{black}{$(1,0,1)$}\tabularnewline
\hline 
\textcolor{black}{$Es_{2d+2}^{2m+1}(z,k^{2})$} & \textcolor{black}{$(0,1,1)$}\tabularnewline
\hline 
\textcolor{black}{$Es_{2d+3}^{2m+2}(z,k^{2})$} & \textcolor{black}{$(1,1,1)$}\tabularnewline
\hline 
\end{tabular}
\par\end{centering}
\caption{Lam\'{e} polynomials solutions to \eqref{eq:LameJacobiForm} and their corresponding symmetries given in Table \ref{tab:LameTable}.\label{tab:LameJacobiTable}}
\end{table}
\begin{table}[H]
\begin{centering}
\textcolor{black}{}%
\begin{tabular}{|c|c|c|}
\hline 
 & \textcolor{black}{Number of States} & \textcolor{black}{$E$}\tabularnewline
\hline 
\textcolor{black}{$(1;\mu)$} & \textcolor{black}{$D(D+2)/8$} & \textcolor{black}{Even}\tabularnewline
\hline 
\textcolor{black}{$(1;1,1,1)$} & \textcolor{black}{$D(D-2)/8$} & \textcolor{black}{Even}\tabularnewline
\hline 
\textcolor{black}{$(0;0,0,0)$} & \textcolor{black}{$(D+2)(D+4)/8$} & \textcolor{black}{Even}\tabularnewline
\hline 
\textcolor{black}{$(0;\mu,\nu)$} & \textcolor{black}{$D(D+2)/8$} & \textcolor{black}{Even}\tabularnewline
\hline 
\textcolor{black}{$(1;0,0,0)$} & \textcolor{black}{$(D+1)(D+3)/8$} & \textcolor{black}{Odd}\tabularnewline
\hline 
\textcolor{black}{$(1;\mu,\nu)$} & \textcolor{black}{$(D+1)(D-1)/8$} & \textcolor{black}{Odd}\tabularnewline
\hline 
\textcolor{black}{$(0;\mu)$ } & \textcolor{black}{$(D+1)(D+3)/8$} & \textcolor{black}{Odd}\tabularnewline
\hline 
\textcolor{black}{$(0;1,1,1)$} & \textcolor{black}{$(D+1)(D-1)/8$} & \textcolor{black}{Odd}\tabularnewline
\hline 
\end{tabular}
\par\end{centering}
\caption{Number of states per symmetry class and corresponding energy parity. \label{tab:Lame Symmetry State Table}}
\end{table}
\begin{figure}[H]
\begin{centering}
\textcolor{black}{\includegraphics[width=7.5cm,height=7.5cm,keepaspectratio]{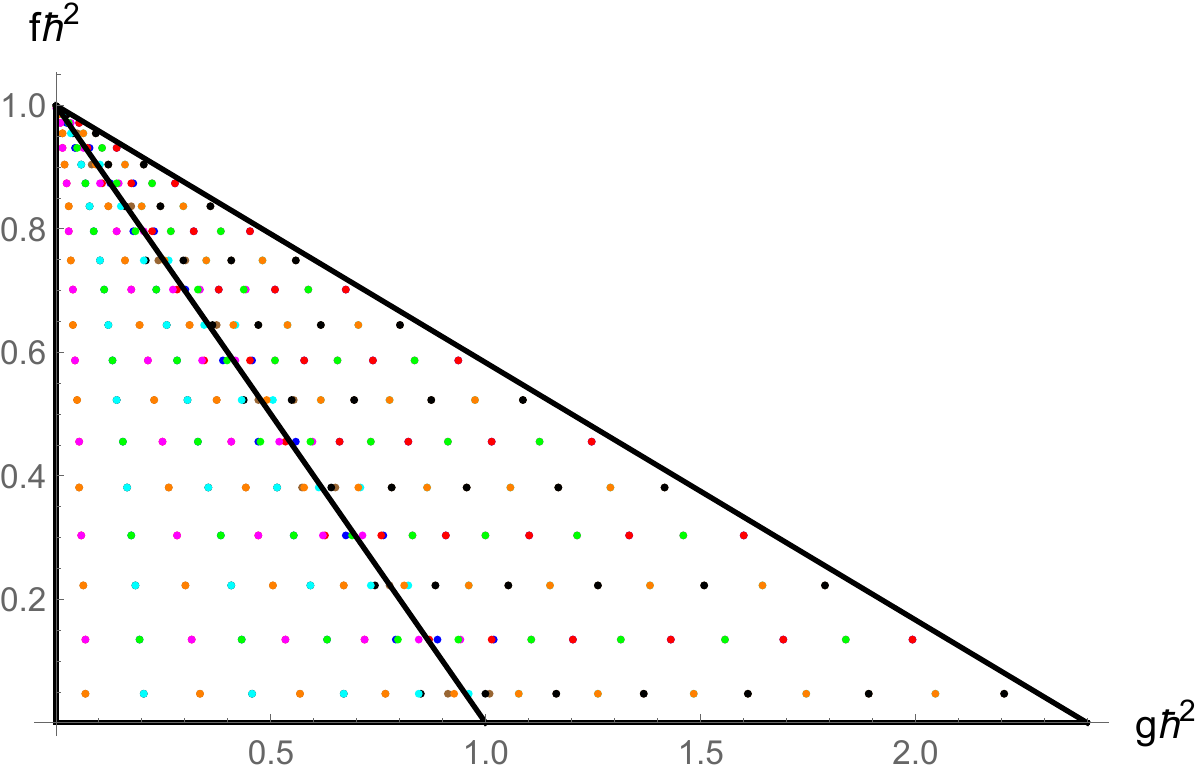}\includegraphics[width=7.5cm,height=7.5cm,keepaspectratio]{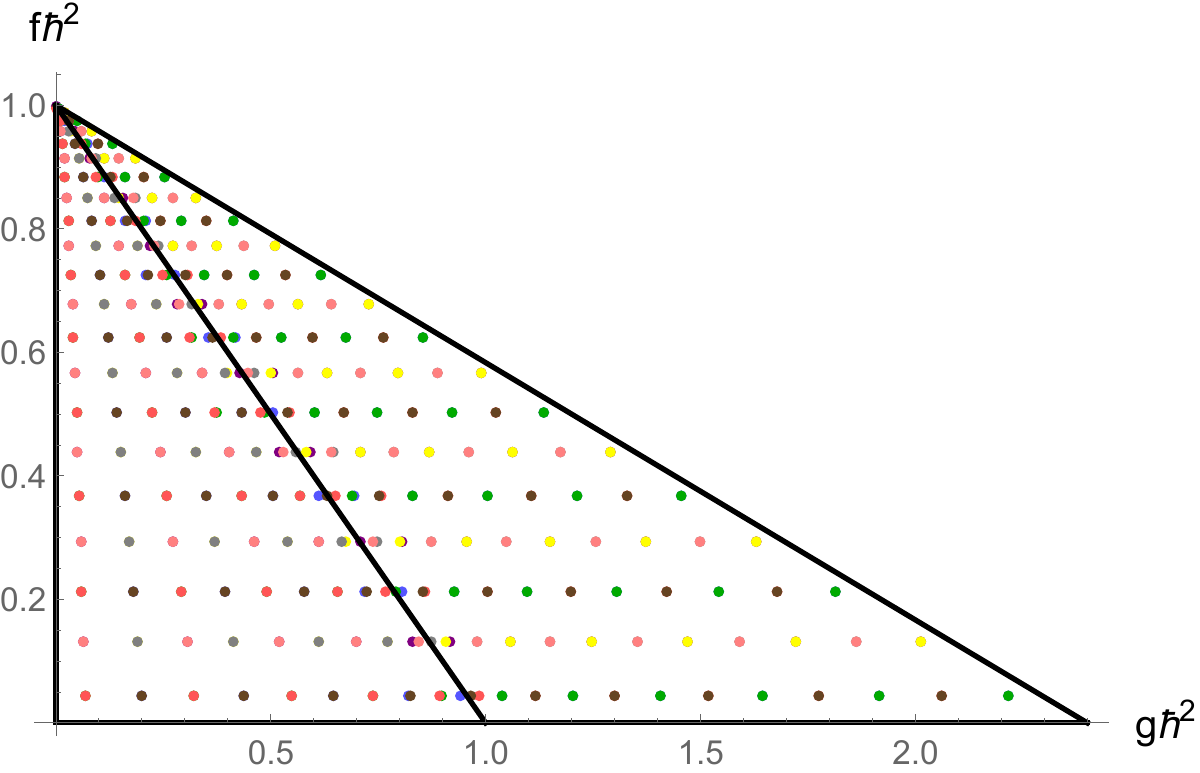}}
\par\end{centering}
\caption{Joint spectrum $(f,g)$ in the Lam\'{e}  case with $(f_1,f_2,f_3)=(0,1,2.4)$ for a) even energy with $D=20$
and b) odd energy with $D=21$ thereby representing all $16$ symmetry classes. Correspondence between the coloured dots and symmetry class are the same for the ellipsoidal system shown in Table \ref{tab:my_label}. \label{fig:LameEvenOdd spectrum}}
\end{figure}

\textcolor{black}{The actions of the Lam\'{e} system are given by
the following 
\[
\begin{aligned}J_{1}=\frac{2}{\pi}\int_{0}^{f}p_{1}ds, &  & J_{2}=\frac{2}{\pi}\int_{f_{1}}^{\min(r_{2},f_{2})}p_{2}ds, &  & J_{3}=\frac{2}{\pi}\int_{\max(f_{2},r_{2})}^{f_{3}}p_{3}ds\end{aligned}
\]
where $r_{2}=\frac{g}{1-f}$ and 
\[
\begin{aligned}p_{1}^{2}=\frac{f-\tilde{E}s_{1}}{4(s_{1}-1)^{2}s_{1}} &  & p_{i}^{2}=\frac{(f-\tilde{E})s_{i}+g}{4(s_{i}-f_{3})(s_{i}-f_{2})(s_{i}-f_{1})}\end{aligned}
\]
for $i=2,3.$ The first action simplifies to 
\begin{equation}
J_{1}=1-\sqrt{1-f}.\label{eq:J1 Lame Action}
\end{equation}
The image of spectra shown in Fig. \ref{fig:LameEvenOdd spectrum},
under the action map is shown in Fig. \ref{fig:LameActions}.}
Unlike the prolate and oblate cases, there are no states on the boundary.
There is a single state located near the top corner $(f,g)=(1-\hbar^2,0)$.
\textcolor{black}{}
\begin{figure}[H]
\begin{centering}
\textcolor{black}{\includegraphics[width=7.5cm,height=7.5cm,keepaspectratio]{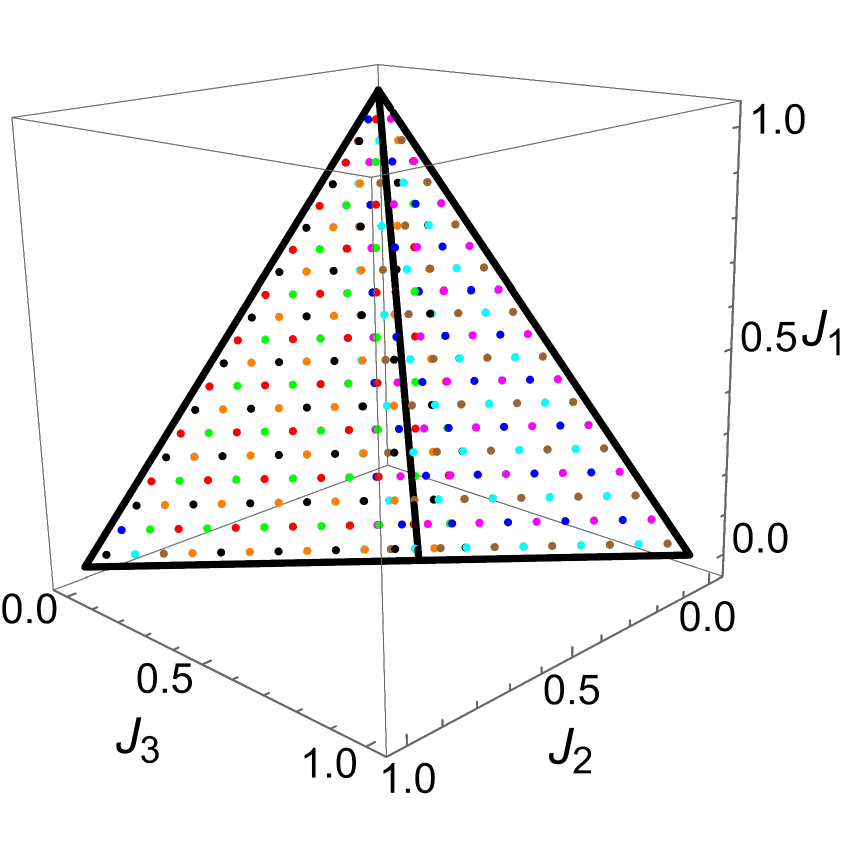}\includegraphics[width=7.5cm,height=7.5cm,keepaspectratio]{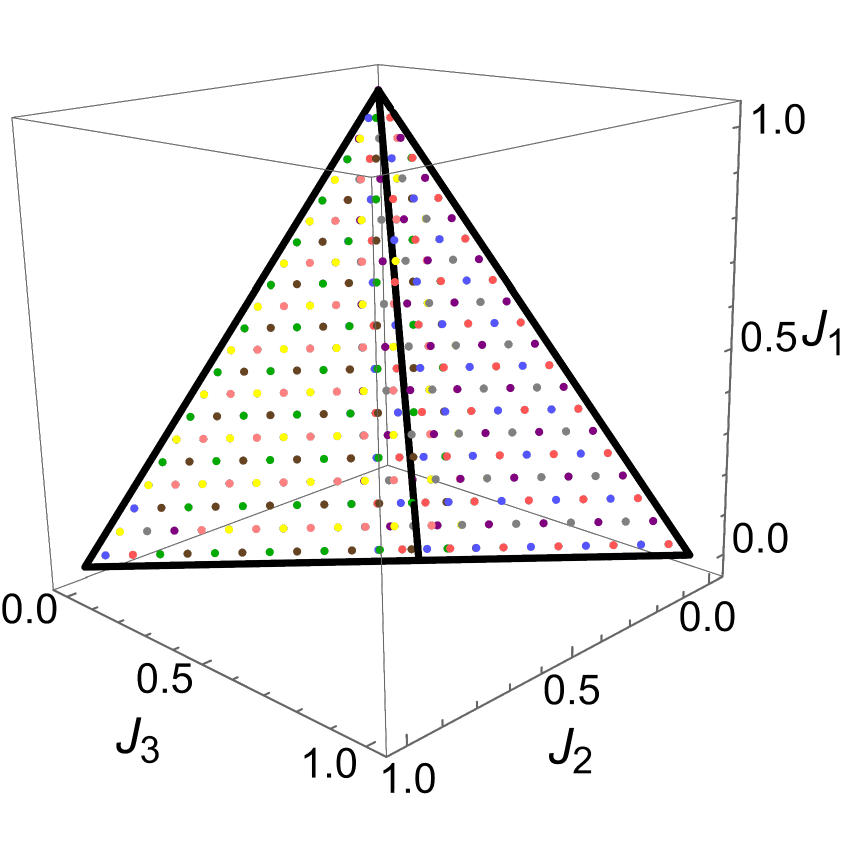}}
\par\end{centering}
\caption{Joint spectrum in the Lam\'{e} case in action variables corresponding to the joint spectra shown in Fig. \ref{fig:LameEvenOdd spectrum} a) and b).\label{fig:LameActions}}
\end{figure}

\subsection{Spherical Coordinates}\label{sec:SphericalS3}

Spherical coordinates (also known as poly-spherical
coordinates) are a further degeneration of prolate, oblate or Lam\'{e}
coordinates and possess both an $SO(2)$ and $SO(3)$ symmetry. Unlike
 prolate coordinates, spherical coordinates - despite having a global $S^{1}$
action - are not semi-toric since they have a degenerate point. This
point corresponds to the critical values at which the global action
is not differentiable, and the pre-image of the critical value is a sphere. We will see that there is a simple eigenfunction built from Chebyshev polynomials that is related to this sphere. 

Spherical coordinates can be obtained from the Lam\'{e}
case by setting either $f_{1}=f_{2}$ or $f_{2}=f_{3}$. These result
in different coordinate systems, respectively called the $12$ and
$23$ spherical systems and are related to each other via a permutation of coordinates. Similarly, setting $a=1$ in both the prolate and
oblate cases gives the $23$ and $12$ spherical systems respectively. See \cite{Nguyen2023} for more details. 

In this paper we focus on the $23$ spherical system,
defined by 
\begin{align}
x_{1}^{2} & =s_{1} & x_{2}^{2} & =(1-s_{1})s_{2}\nonumber \\
x_{3}^{2} & =(1-s_{1})(1-s_{2})s_{3} & x_{4}^{2} & =(1-s_{1})(1-s_{2})(1-s_{3}),\label{eq:23 spherical def}
\end{align}
 where $0\le s_{k}\le1$ for all $k=1,2,3$. Unlike the previously
described systems, the spherical coordinates are relatively simple
and separation can be done "by hand", without the use of a St\"{a}ckel matrix. The classical integrals obtained by separation are $(F, \ell_{34}^2)$ where $F=\ell_{12}^2 + \ell_{13}^2 + \ell_{14}^2$. Note this is the same $F$ as in the Lam\'{e} system. The separated ODEs are 
\begin{subequations}
\begin{align}
\psi_{1}^{''}+\frac{1}{2}\left(\frac{3}{s_{1}-1}+\frac{1}{s_{1}}\right)\psi_{1}^{'}+\frac{\lambda-Es_{1}}{4s_{1}(1-s_{1})^{2}}\psi_{1} & =0\label{eq:Spherical Separation1} \\
\psi_{2}^{''}+\left(\frac{1}{2s_{2}}+\frac{1}{s_{2}-1}\right)\psi_{2}^{'}+\frac{(\lambda-E)(s_{2}-1)-m}{4s_{2}(1-s_{2})^{2}}\psi_{2} & =0\label{eq:Spherical Separation2}\\
\psi_{3}^{''}+\frac{1}{2}\left(\frac{1}{s_{3}}+\frac{1}{s_{3}-1}\right)\psi_{3}^{'}+\frac{m}{4s_{3}(1-s_{3})}\psi_{3} & =0.\label{eq:Spherical Separation3}
\end{align}
\end{subequations}

The equation for $\psi_{3}$ \eqref{eq:Spherical Separation3}
can be converted to the trivial equation \eqref{eq:trivial equation} with the coordinate transformation $s_3 = \cos^2(\phi)$. Using this and periodic
boundary conditions, we find that $m\in\mathbb{Z}$ and $\psi_{3}=e^{im\phi}$
where $\phi=\arctan(\frac{x_{3}}{x_{4}})$.

As for the previous degenerate systems, we have the following Lemma.

\begin{lem}
    The system of separated equations in \eqref{eq:Spherical Separation1}, \eqref{eq:Spherical Separation2}, \eqref{eq:Spherical Separation3} can be obtained from either those of the Lam\'{e} system \eqref{eq:raw equations}, \eqref{eq:raw equations2} or those of the oblate system \eqref{eq:OblateEqn1}, \eqref{eq:OblateEqn2}. 
\end{lem}
\begin{proof}
    For the Lam\'{e} system \eqref{eq:raw equations}, \eqref{eq:raw equations2}, the transformation $(s_3,p_3)\to (f_2 + \tilde{s_3} \epsilon, \tilde{p_3}/\epsilon)$ is canonical. This, coupled with $f_3 \to f_2 + \epsilon$ and a corresponding transformation of the integrals followed by the normalization $(f_1,f_2)=(0,1)$ yields the result. For the oblate direction, a similar method with $(a, s_3, p_3) = (1+\epsilon, 1+ \epsilon \tilde{s_3}, \tilde{p_3}/\epsilon)$ in \eqref{eq:OblateEqn1} and \eqref{eq:OblateEqn2} gives the result.  
\end{proof}
Next, we turn our attention to the non trivial equations.
For $\psi_{1}$ in \eqref{eq:Spherical Separation1}, we observe that it is of hypergeometric type. It
has 3 regular singularities - two finite poles at $s_{1}=0,1$ and
the third at infinity. To assist in finding polynomial solutions,
we transform the equation for $\psi_{1}$ using the change of dependent
variable $\psi_{1}=(1-s_{1})^{\frac{1}{4}(-1+X)}W_{1}$ where $X=\sqrt{1+4E-4\lambda}$,
followed by the change of independent variable $s_{1}=x_1^2$.
This yields the Gegenbauer equation 
\begin{equation}
(1-x_1^2)W_1^{''}-2x_1(2u+1)W_1^{'}+n(n+2u)W_1=0\label{eq:GegSph1}
\end{equation}
where $u = \frac{1}{2}(1+X)$ and 
\begin{equation}
    n = -\frac{1}{2}(1-2\sqrt{1+E}+X).\label{eq:n1 eq}
\end{equation}
Gegenbauer polynomial solutions to \eqref{eq:GegSph1}, denoted by $C_{n}^{u}(x_1)$, are obtained when $n_1$ is a non negative integer.

Repeating a similar process for $\psi_{2}$ in \eqref{eq:Spherical Separation2}, we
again note the equation is hypergeometric with regular singularities
at $0,1$ and infinity. The change of independent variable $s_2 = x^2 \coloneqq \frac{x_2^2}{\tilde{r}^2}$ where $\tilde{r}\coloneqq x_2^2 + x_3^2 + x_4^2$ gives the Associated Legendre equation written as follows
\begin{equation}
    (1-x^2)\psi_2^{''} - 2x\psi_2^{'} + \left(\ell(\ell+1)-\frac{m^2}{1-x^2}\right)\psi_2 = 0 \label{eq:AssocLegendreFirstAppearance}
\end{equation}
where $\ell(\ell+1) = E-\lambda$, i.e. 
\begin{equation}
    \ell = \frac{1}{2}(-1+\sqrt{1+4E-4\lambda}). \label{eq:ellS3Spherical}
\end{equation}
In \eqref{eq:AssocLegendreFirstAppearance}, $\ell$ and $m$ are referred to as the degree and order respectively of the Associated Legendre equation. Non trivial non singular solutions $P_{\ell}^{m}(x)$ are yielded when both $\ell$ and $m$ are integers. Further, said solutions are polynomial only if $m$ is an even integer. 
For polynomial solutions, both $n$ and $\ell$
from (\ref{eq:n1 eq}) and (\ref{eq:ellS3Spherical}) respectively must be
integers. Solving these simultaneously for $E$ gives 
\begin{equation}
    E=D(D+2) \label{eq:EforSphericalBase}
\end{equation}
where $D=\ell + n$. Using \eqref{eq:EforSphericalBase} with \eqref{eq:ellS3Spherical} gives analytic values of the eigenvalue $\lambda$ as 
\begin{equation}
    \lambda = D(D+2)-\ell(\ell+1)\label{eq:g base}
\end{equation}
Using (\ref{eq:g base}), we also have 
\begin{equation}
X=2\ell+1\label{eq:XSimp}
\end{equation}
along with $u=\ell+1$.
Combining the above, we have the following Lemma. 
\begin{lem}
For a fixed value of $D=n+\ell$, the combined wavefunction
solution to (\ref{eq:Schrodinger}), expressed in Cartesian coordinates, is the homogeneous polynomial given by the product 
\begin{equation}
\Psi_{D}(\bm{x})=r^{n}\tilde{r}^{\ell}\left(\frac{x_{3}+ix_{4}}{\sqrt{x_3^2+x_4^2}}\right)^{m}C_{n}^{\ell+1}\left(\frac{x_{1}}{r}\right)P_{\ell}^{m}\left(\frac{x_2}{\tilde{r}}\right)\label{eq:JacobiSphericalCartesian}
\end{equation}
where $\tilde r^2 = r^2 - x_1^2$.
\end{lem}
\begin{proof}
First, we note as before $e^{im\phi}= (\frac{x_3+ix_4}{\sqrt{x_3^2+x_4^2}})^m$. 
The associated Legendre polynomials $P_{\ell}^{m}(x)$ can be written in the form (see,~e.g.,~\cite[14.7]{DLMF})
\begin{equation}
    P^m_\ell(x) = \frac{(-1)^m}{2^\ell \ell!} (1 - x^2)^{m/2} \frac{d^{\ell+m}}{dx^{\ell+m}} (x^2 - 1)^\ell \,.
\end{equation}
Thus the denominator $\sqrt{x_3^2 + x_4^2}^{m}$ in the Lemma cancels with $\tilde r^m ( 1 - x_2^2/\tilde r^2)^{m/2}$. The remaining terms produce a homogeneous degree $\ell - m$ polynomial in $x_2$ and $\tilde r^2 - x_2^2 = x_3^2 + x_4^2$.
We note that $r^{n}C_{n}^{\ell+1}(\frac{x_1}{r})$  is homogeneous of degree $n$ and so it follows that \eqref{eq:JacobiSphericalCartesian} is homogeneous of degree $D=n+\ell$. 
\end{proof}
Let $\bm{\mu}=(\mu_1,\mu_2)$ where $\mu_i\in \{0,1\}$ and $\mu_1$ $(\mu_2)$ being $1$ $(0)$ denotes a solution odd (even) about the $x_1$ $(x_2)$ axis. As with the Lam\'{e} system, we note that both the Gegenbauer and Associated Legendre equations have independent variables $x_1$ and $x_2$ respectively and thus their parity about these axes will be determined by that of $n$ and $\ell-m$ respectively. Consequently, a fixed value of $D$ enforces a relationship between $n,\ell$ and $m$.

We note that $\ell=0$ forces $D=n$ and $m=0$. Consequently, we recover Chebyshev polynomials of the second kind as $C_{n}^{1}(x_1) = U_n(x_1)$. The corresponding state on the joint spectrum is $(m,\lambda)=(0,1-\hbar^2)$,
which becomes the classical degenerate point in the limit of $\hbar\to0$. Depending on the parity of $E$, this corresponds to either the $(0,0)$ symmetry (if $E$ is even) or $(1,0)$ if $E$ is odd.

Fixing an energy $E$ allows us to compute the joint spectrum for the spherical system. An example of the joint spectrum $(m,\lambda)$ is shown in Fig.
\ref{fig:SphericalSpecAct} a). The degenerate point is shown in
magenta at $(0,1)$.

From \cite{Nguyen2023}, it is known that the actions
of the spherical system are given as follows:
\begin{align}
J_{1}=1-\sqrt{1-f} &  & J_{2}=\sqrt{1-f} - m &  & J_{3}=m\label{eq:actionSpherical}
\end{align}
where $f$ is a value of the integral $F$. 
In Fig. \ref{fig:SphericalSpecAct} b) we show the corresponding action map to the joint spectrum shown in a). 
\begin{figure}[H]
\begin{centering}
\textcolor{black}{\includegraphics[width=7.5cm,height=7.5cm,keepaspectratio]{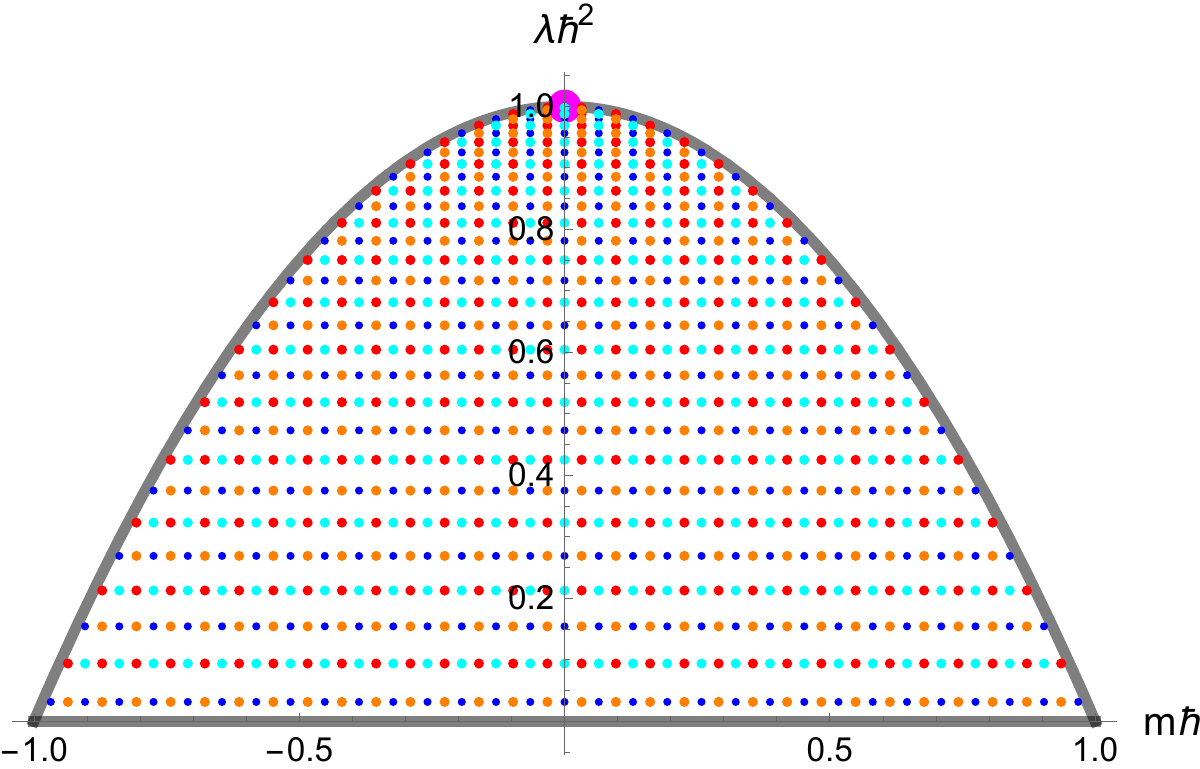}\includegraphics[width=7.5cm,height=7.5cm,keepaspectratio]{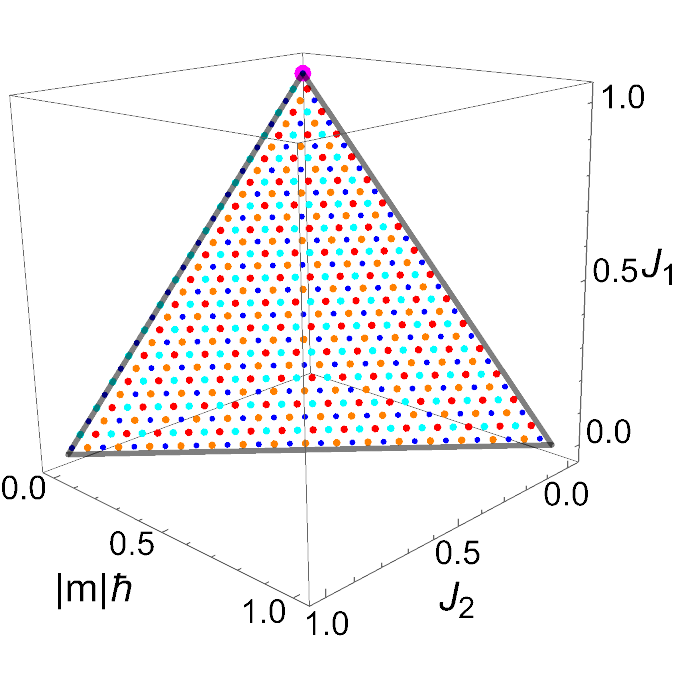}}
\par\end{centering}
\caption{a) The joint spectrum in the spherical case with $D=30$. Blue, red, orange and cyan represent
the $(0,0),(1,0),(0,1)$ and $(1,1)$ symmetries respectively. There
are $31^{2}=961$ states in total. b) The corresponding spectrum in action variables
(\ref{eq:actionSpherical}). \label{fig:SphericalSpecAct}}
\end{figure}
\subsection{Cylindrical Coordinates}\label{sec:CylindricalS3}

\textcolor{black}{Cylindrical (or Hopf, or doubly cylindrical) coordinates are a further degeneration
of oblate coordinates, found by setting $e_{1}=e_{2}$ in addition
to $e_{3}=e_{4}$. The relationship between Cartesian and cylindrical
coordinates is given by 
\begin{equation}
\begin{aligned}x_{1}^{2} & =s_{1}s_{2} & x_{2}^{2} & =s_{2}(1-s_{1})\\
x_{3}^{2} & =s_{3}(1-s_{2}) & x_{4}^{2} & =(1-s_{2})(1-s_{3}).
\end{aligned}
\label{eq:cylindricalcoorddef}
\end{equation}
The separated equations are given by 
\begin{subequations}
\begin{align}
\psi_{1}^{''}+\frac{1}{2}\left(\frac{1}{s_{1}}+\frac{1}{s_{1}-1}\right)\psi_{1}^{'}+\frac{m_{1}^{2}}{4s_{1}(1-s_{1})}\psi_{1} & =0\label{eq:cyl1}\\
\psi_{2}^{''}+\left(\frac{1}{s_{2}-1}+\frac{1}{s_{2}}\right)\psi_{2}^{'}+\frac{m_{1}(s_{2}-1)-(m_{2}+E(s_{2}-1))s_{2}}{4s_{2}^{2}(1-s_{2})^{2}}\psi_{2} & =0\label{eq:cyl2}\\
\psi_{3}^{''}+\frac{1}{2}\left(\frac{1}{s_{3}}+\frac{1}{s_{3}-1}\right)\psi_{3}^{'}+\frac{m_{2}^{2}}{4s_{3}(1-s_{3})}\psi_{3} & =0\label{eq:cyl3}
\end{align}
\end{subequations}
where $m_{1},m_{2}$ are spectral parameters. The equations in $\psi_{1}$
and $\psi_{3}$ in \eqref{eq:cyl1} and \eqref{eq:cyl3} are reducible to the trivial equation \eqref{eq:trivial equation} with the trigonometric transformation used for the previous systems. This, along
with periodic boundary conditions, give $m_{1},m_{2}\in\mathbb{Z}$
with solutions $\psi_{1}=e^{im_{1}\phi_{1}}$ and $\psi_{3}=e^{im_{2}\phi_{3}}$
where $\phi_{1}=\arctan(\frac{x_{2}}{x_{1}})$ and $\phi_{3}=\arctan(\frac{x_{4}}{x_{3}})$. }

In \eqref{eq:cyl2}, the equation in $\psi_{2}$ is of hypergeometric type. The change
of dependent variable $\psi_{2}=s_{2}^{\frac{\left|m_{1}\right|}{2}}(1-s_{2})^{\frac{\left|m_{2}\right|}{2}}y$
and independent $s_{2}=\frac{1-x}{2}$ variables yields 
\[
y^{''}+\left(\frac{1+\left|m_{1}\right|}{x-1}+\frac{1+\left|m_{2}\right|}{x+1}\right)y^{'}+\frac{-E+2\left|m_{1}\right|\left(\left|m_{2}\right|+1\right)+m_{1}^{2}+m_{2}^{2}+2\left|m_{2}\right|}{4(x^{2}-1)}y=0
\]
 whose Riemann symbol is given by 
\[
\mathcal{S}_{\text{Cyl}}=\left(\begin{aligned}-1 &  & 1 &  & \infty\\
0 &  & 0 &  & \frac{1}{2}\left(-\sqrt{E+1}+\left|m_{1}\right|+\left|m_{2}\right|+1\right) & ;x\\
-\left|m_{2}\right| &  & -\left|m_{1}\right| &  & \frac{1}{2}\left(\sqrt{E+1}+\left|m_{1}\right|+\left|m_{2}\right|+1\right)
\end{aligned}
\right).
\]

Polynomial solutions to $\psi_{2}$ are obtained when an index at infinity vanishes. This leads to the condition
\begin{equation}
    \frac{1}{2}\left(-\sqrt{E+1}+\left|m_{1}\right|+\left|m_{2}\right|+1\right)=-d
\end{equation}
which gives the quantised energy
\[
E=D(D+2)
\]
where $D=2d+\left|m_{1}\right|+\left|m_{2}\right|$. These solutions, denoted by $P_d^{(|m_1|,|m_2|)}(x)$, are the well known Jacobi polynomials.

The combined wavefunction solution to \eqref{eq:Schrodinger} is therefore
\begin{equation}
    \psi = P_d^{(|m_1|,|m_2|)}(x)e^{im_1\phi_1}e^{im_2\phi_3}. \label{eq:combCylSol}
\end{equation}
\begin{lem}
For a given $D=2d + |m_1| + |m_2|$, the total wave function \eqref{eq:combCylSol} in the original Cartesian
coordinates is a harmonic homogeneous polynomial of degree $D$ given by
\begin{equation}
\Psi_{D}=r^{2d}(x_{1}+i\sign(m_1)x_{2})^{|m_{1}|}(x_{3}+i\sign(m_2)x_{4})^{|m_{2}|}P_{d}^{(\left|m_{1}\right|,\left|m_{2}\right|)}\left(\frac{x_3^{2}+x_4^{2}-x_{1}^{2}-x_{2}^{2}}{r^2}\right).\label{eq:cylindrical eqn}
\end{equation}
\end{lem}
\begin{proof}
    Using the definition of cylindrical coordinates in \eqref{eq:cylindricalcoorddef}, we note that $s_2 = x_1^2 + x_2^2$. Since $x=1-2s_2$ and $r^2 = \sum_i x_i^2$, the argument of the Jacobi polynomial follows. Similarly, rewriting the exponential terms in \eqref{eq:cylindrical eqn} using this definition and cancelling terms gives the final result.
    To prove homogeneity, we note that the argument of the Jacobi polynomial is of degree $0$ and so the final degree of $\Psi_D$ will be $|m_1|+|m_2| + 2d = D$.
\end{proof}
Note that for $m_{1}=m_{2}=0$ we recover the
Legendre polynomials $P_d^0(x)=P_d^{(0,0)}(x)$.
For $m_{1}=m_{2}$ we recover Gegenbauer
polynomials $C_d^l(x) / P_d^{(l+1/2, l+1/2)} = const$, i.e.~they are the same up to normalisation,
see, e.g.~\cite{szegő1975orthogonal}.

Since there are two continuous $S^{1}$ symmetries, the discrete symmetries of (\ref{eq:cylindrical eqn})
are represented by the parity of $m_1$ and $m_2$.
Note that $D - 2d = |m_1| + |m_2|$ has the same parity as $D$ and $E$. In particular this implies that $m_1 = m_2 = 0$ is only possible for even $D$ and $E$.

An example of the joint spectrum is given in Fig. \ref{fig:CylindricalPics}
a) with $D=20$. The total number of eigenstates is $(D+1)^{2}=441$. 

Two of the spectral parameters are trivial angular momenta: $(J_{1},J_{3})=(\left|m_{1}\right|,\left|m_{2}\right|)$.
As $J_{1}+J_{2}+J_{3} = 1$, the remaining action $J_2$ is easily found.
In Fig. \ref{fig:CylindricalPics} b) an example of the action
map for the corresponding spectrum in a) is shown. 

\begin{figure}[H]
\begin{centering}
\includegraphics[width=7.5cm,height=7.5cm,keepaspectratio]{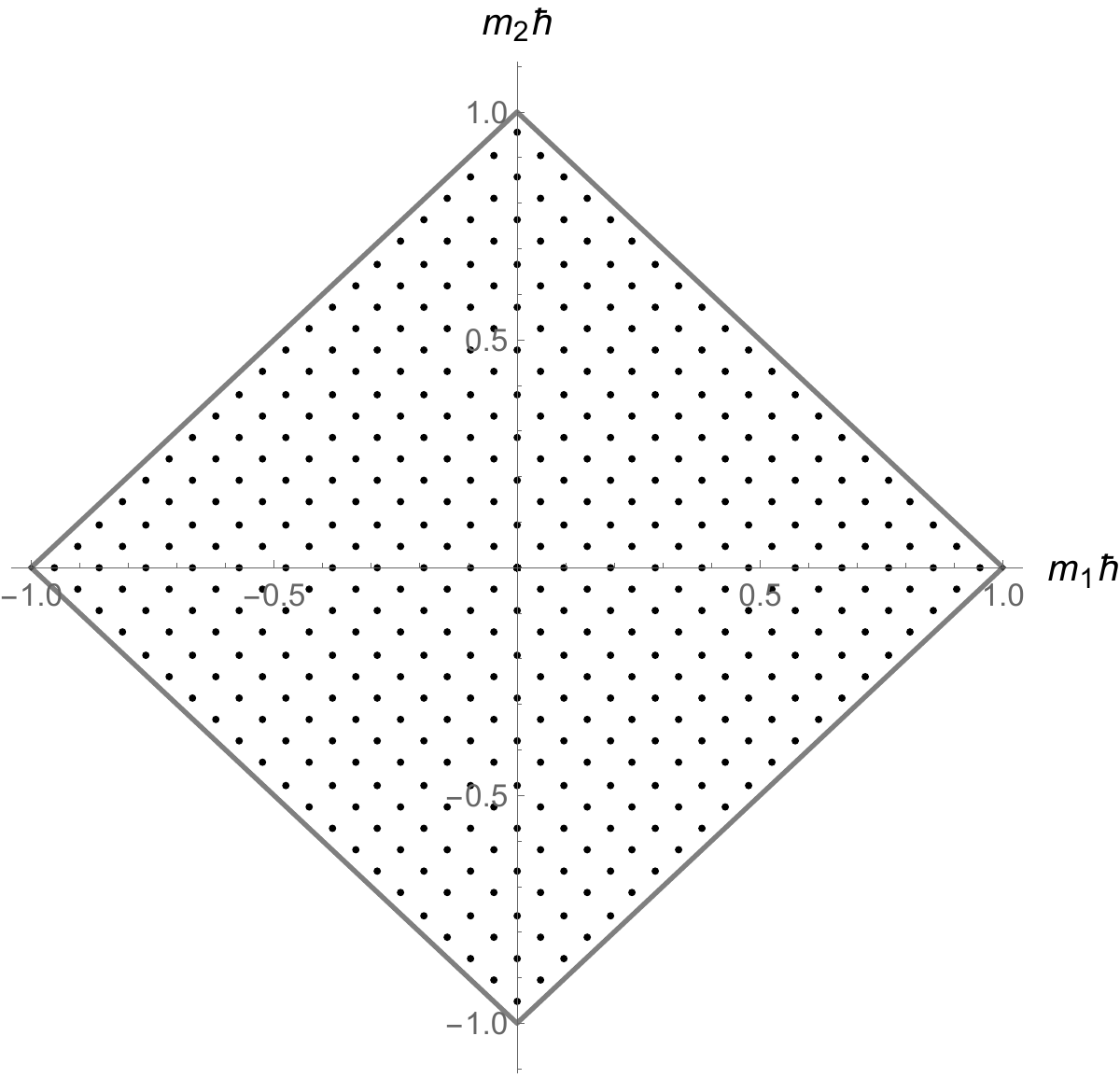}\includegraphics[width=7.5cm,height=7.5cm,keepaspectratio]{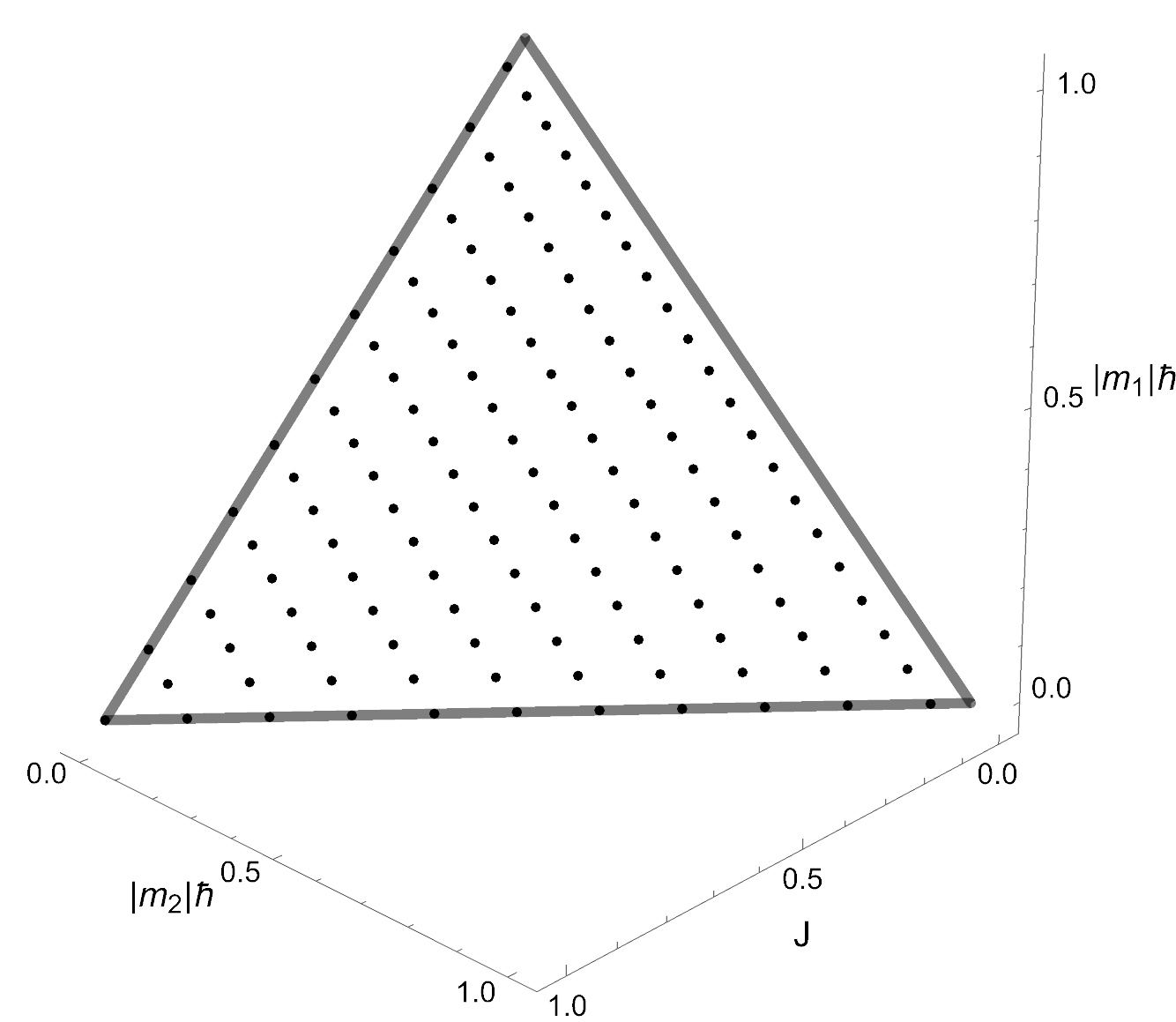}
\par\end{centering}
\caption{a) Joint spectrum with $D=20$. b) Corresponding action map.\label{fig:CylindricalPics}}

\end{figure}

\appendix 

\section{Appendix}\label{sec:AppendixS3}

\subsection{Ellipsoidal Coordinates on $S^{2}$}
In these remaining sections we compare our results on $S^3$ to the more familiar $S^2$.
Ellipsoidal coordinates on $S^{2}$ are defined
as follows: 
\begin{equation}
\begin{aligned}x_{1}^{2} & =\frac{(s_{1}-e_{1})(s_{2}-e_{1})}{(e_{2}-e_{1})(e_{3}-e_{1})} &  & x_{2}^{2}=\frac{(s_{1}-e_{2})(s_{2}-e_{2})}{(e_{1}-e_{2})(e_{3}-e_{2})} &  & x_{3}^{2}=\frac{(s_{1}-e_{3})(s_{2}-e_{3})}{(e_{1}-e_{3})(e_{2}-e_{3})}.\end{aligned}
\label{eq:S2 ellipsoidal def}
\end{equation}
where $0\le e_1 \le s_1 \le e_2 \le s_2 \le e_3$.
Separating \eqref{eq:Schrodinger} in these coordinates gives the following separated equations
\begin{equation}
\psi_{i}^{''}+\frac{1}{2}\left(\frac{1}{s_{i}-e_{1}}+\frac{1}{s_{i}-e_{2}}+\frac{1}{s_{i}-e_{3}}\right)\psi_{i}^{'}+\frac{-Es_{i}+\lambda}{4(s_{i}-e_{1})(s_{i}-e_{2})(s_{i}-e_{3})}\psi_{i}=0\label{eq:EllipsoidalS2}
\end{equation}
for $i=1,2$. 
Equations \eqref{eq:EllipsoidalS2} for $i=1,2$ are the Lam\'{e} equations studied in section \ref{sec:LameS3} and so we denote their polynomial solutions by $Lp_d(s_i)$. 
Classically, the integrals obtained
by separation are
\[
(I_{1},I_{2})=(\ell_{12}^{2}+\ell_{13}^{2}+\ell_{23}^{2},e_{3}\ell_{12}^{2}+e_{2}\ell_{13}^{2}+e_{1}\ell_{23}^{2}).
\]
A combined solution to \eqref{eq:Schrodinger} in ellipsoidal coordinates is given by 
\begin{equation}
    \psi = Lp_d(s_1)Lp_d(s_2)\label{eq:S2EllipsoidalWF}
\end{equation}
where there are $d = n_1 + n_2$ total roots of $\psi$ with $n_1$ occurring in the interval $[e_1, e_2]$ and $n_2$ in $[e_2,e_3]$.
\begin{lem}
Written in Cartesian coordinates, $\psi$ in \eqref{eq:S2EllipsoidalWF} is a homogeneous $\tilde{\ell}\coloneqq 2d$ degree polynomial
\begin{equation}
\Phi_{\tilde{l}}=\left(\prod_{i=1}^{3}\prod_{k=1}^{d}(z_{k}-e_{i})\right)\prod_{k=1}^{d}\sum_{i=1}^{3}\frac{x_{i}^{2}}{z_{k}-e_{i}}\label{eq:baseS2 Ellipsoidal}
\end{equation}
where the $z_{k}$ are solutions
of
\[
\sum_{i=1}^{3}\frac{1/4}{z_{k}-e_{i}}+\sum_{j=1,j\ne k}^{d}\frac{1}{z_{k}-z_{j}}=0.
\]
\end{lem}
Inspecting (\ref{eq:baseS2 Ellipsoidal}), it is
clear there are 8 symmetry classes in total (even/odd parity about
each of the $x_{i}$ axes). Let $\bm{\mu} \coloneqq (\mu_1, \mu_2,\mu_3)$ where $\mu_i \in \{0,1\}$ represent a symmetry class of the solution \eqref{eq:baseS2 Ellipsoidal} where $\mu_i=1$ signifies a solution odd about the $x_i$ axis and even otherwise. As with the $S^3$ ellipsoidal case, we consider these symmetries by performing the change of variables 
\begin{equation}
     \phi_j = \prod_{i=1}^3(s_j-e_i)^{\mu_i/2}\psi_{j} \label{eq:baseS2transform}
\end{equation}
in \eqref{eq:EllipsoidalS2}. Let solutions to the transformed \eqref{eq:EllipsoidalS2} for a given symmetry class $\bm{\mu}$ be given by $\Phi_{\tilde{l}}^{\bm{\mu}}(\bm{x})$ and set 
\begin{equation}
    \Psi_{\ell}^{\bm{\mu}}(\bm{x}) \coloneqq x_1^{\mu_1}x_2^{\mu_2}x_3^{\mu_3}\Phi_{\tilde{l}}^{\bm{\mu}}(\bm{x})
\end{equation}
We then have the following Lemma.
\begin{lem}
    The $\Psi_{\ell}^{\bm{\mu}}(\bm{x})$ are eigenfunctions of the Schrödinger equation \eqref{eq:Schrodinger} and the energy is given by $E = \ell(\ell+1)$ where $\ell \coloneqq 2d + \sum_{i=1}^3\mu_i$.
\end{lem}
Similar to the $S^{3}$ ellipsoidal and Lam\'{e}
system, for a fixed value of $\ell$, only four of these symmetries appear
at any time. When $\ell$ is even, these are the $(0,0,0), (1,1,0), (1,0,1), (0,1,1)$ classes (blue, purple, orange, gray respectively in Fig. \ref{fig:S2 Ellipsoidal Plots} a)) while for odd $\ell$ we have the $(1,0,0), (0,1,0), (0,0,1), (1,1,1)$ classes (red, green, cyan, brown). Note that $E$ is always even. 

We demonstrate an example of this in Fig. \ref{fig:S2 Ellipsoidal Plots}
a) and b) with $\ell=20$ and $\ell=19$ respectively and $(e_1, e_2, e_3) = (0, 1, 2.4)$. The
joint spectrum is obtained using the same three term recurrence relations from Lemma \ref{LemmaProlate}
used to find the prolate, oblate and Lam\'{e} spectra.
\begin{figure}[H]
\begin{centering}
\includegraphics[width=7.5cm,height=7.5cm,keepaspectratio]{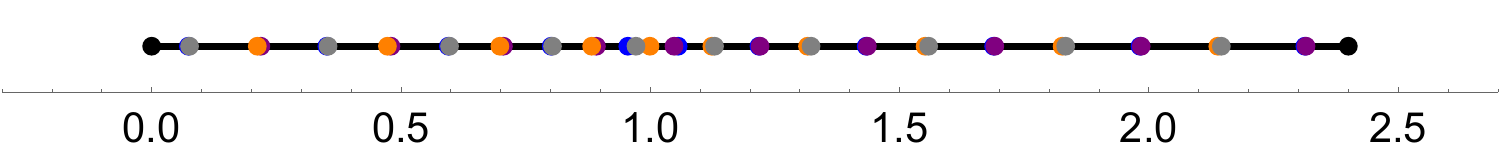}\includegraphics[width=7.5cm,height=7.5cm,keepaspectratio]{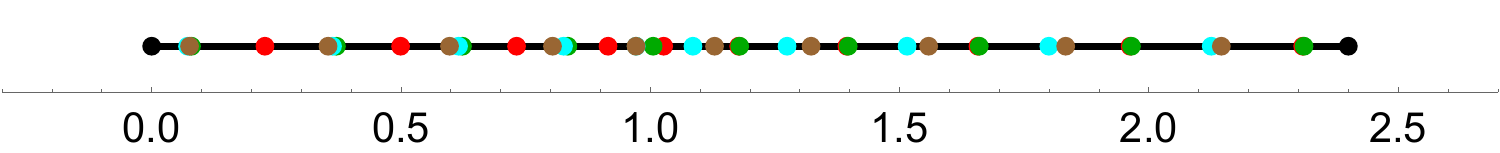}
\par\end{centering}
\textcolor{black}{\caption{Example spectra with a) $\ell=20$ and b) $\ell=19$ respectively and $(e_1, e_2, e_3) = (0, 1, 2.4)$. \label{fig:S2 Ellipsoidal Plots}}
}
\end{figure}
For ellipsoidal coordinates on the $2-$sphere,
consider a fixed value of the energy on $(E_{2}),$ the corresponding
eigenvalues $\lambda$, as well as an energy on the $3-$sphere $(E_{3})$
and eigenvalues for the \textcolor{black}{Lam\'{e}}
system $(f,g)$. Comparing (\ref{eq:EllipsoidalS2}) and (\ref{eq:raw equations})
for $k=1,2$, it is clear that 
\[
\begin{aligned}E_{2}=f-E_{3} &  & \lambda & =g.\end{aligned}
\]
This can be geometrically interpreted as follows: each joint spectrum
for the ellipsoidal $S^{2}$ system (such as those shown in Fig.
\ref{fig:S2 Ellipsoidal Plots}) corresponds to a horizontal slice
of the Lam\'{e} joint spectrum (see Fig. \ref{fig:LameEvenOdd spectrum}), where the slice is specified for a given combination of $(f,E_{3})$.

\subsection{Spherical Coordinates on $S^{2}$}
Spherical coordinates on $S^{2}$ are defined as follows:
\begin{align}
x_{1}^{2}=s_{1} &  & x_{2}^{2}=(1-s_{1})s_{2} &  & x_{3}^{2}=(1-s_{1})(1-s_{2})
\end{align}
where $0\le s_1 \le s_2 \le 1$.
Separating the Hamilton-Jacobi equation gives classical
integrals $(I_{1},I_{2})=(\ell_{12}^{2}+\ell_{13}^{2}+\ell_{23}^{2},\ell_{12}^{2})$
while doing so for the  Schrödinger equation \eqref{eq:Schrodinger} yields the following
ODEs
\begin{subequations}
\begin{align}\psi_{1}^{''}+\frac{1}{2}(\frac{1}{s_{1}}+\frac{2}{s_{1}-1})\psi_{1}^{'}+\frac{-E s_{1}+(E-m^{2})}{4s_{1}(s_{1}-1)^{2}}\psi_{1} & =0\label{eq:Spherical1 S2}\\
\psi_{2}^{''}+\frac{1}{2}(\frac{1}{s_{2}}+\frac{1}{s_{2}-1})\psi_{2}^{'}+\frac{m}{4s_{2}(s_{2}-1)}\psi_{2} & =0 \label{eq:Spherical2 S2}
\end{align}
\end{subequations}
where $(E,m)$ are spectral parameters. The change of independent variable $s_2^2=\cos^2\phi$ transforms \eqref{eq:Spherical2 S2} to the trivial equation \eqref{eq:trivial equation} with solution $e^{im\phi}$ and the associated boundary conditions force $m\in\mathbb{Z}$. 

As with the previous degenerate systems, we have the
following Lemma.
\begin{lem}
The spherical ODEs (\ref{eq:Spherical1 S2}), (\ref{eq:Spherical2 S2}) can be
obtained by degenerating those arising from ellispoidal coordinates
(\ref{eq:baseS2 Ellipsoidal})
\end{lem}
For the equation in $\psi_{1}$ \eqref{eq:Spherical1 S2}, the change of independent variable $s_1 = x_1^2$ transforms \eqref{eq:Spherical1 S2} into the form of the associated Legendre equation given in \eqref{eq:AssocLegendreFirstAppearance}
where $E =\ell(\ell +1)$ is always even and $\ell$ is an integer. Polynomial solutions are this given by the Associated Legendre polynomials $P^m_{\ell}(x_1)$. 
We have the following Lemma.
\begin{lem}
For a given energy $E=\ell(\ell+1)$ the wave function of the Schrödinger equation (\ref{eq:Schrodinger}) in spherical
coordinates on $S^{2}$ is one of the following degree $\ell$ harmonic homogeneous polynomials
\begin{equation}
    \Psi_{\ell} = r^{\ell}
   \left(\frac{x_2+ix_3}{\sqrt{x_2^2+x_3^2}}\right)^m
P_{\ell}^m\left(\frac{x_1}{r}\right), \label{eq:phi S2}
\end{equation}
where $m = -\ell, \dots, \ell$.
\end{lem}
\begin{proof}
    The form of \eqref{eq:phi S2} is obtained by recognising $e^{im\phi} = \left(\frac{x_2+ix_3}{\sqrt{x_2^2+x_3^2}}\right)^m$. Harmonicity is clear and homogeneity follows by noting that $P_{\ell}^m$ is degree $\ell$ in its argument but the argument $x_1/r$ is degree 0 in terms of the Cartesian coordinates. Moreover, $x_2^2 + x_3^2 = r^2 - x_1^2$ and so the denominator $(x_2^2 + x_3^2)^{m/2}$ cancels with a factor in $r^m P_{\ell}^m(x_1/r)$ and the remaining polynomial in $x_1/r$ is of degree $\ell - m$.
\end{proof}
For discrete symmetries reflecting parity about the $x_1$ axis, we let $\mu\in \{0,1\}$ denote the solution which is even $(0)$ or odd $(1)$ about the $x_1$ axis. The symmetry is even when $\ell - m$ is even, and odd otherwise. 
We note that, like for the spherical $S^3$ system, for a fixed $\ell$ and fixed symmetry class only every 2nd $m$ is allowed so that the parity of $\ell - m$ remains fixed.

\subsection{\label{subsec:ActionsEllipsoida}Ellipsoidal $S^{3}$ Quantised Actions}
Here, we show the quantised actions for each symmetry class $\bm{\mu}$ originally presented in Fig. \ref{fig:Fig2} a) and b). The number of eigenstates per symmetry class is given in Table \ref{tab:NumStatesEllipsoidal}. For a legend to convert between colour and symmetry class, see Table \ref{tab:my_label}. 
\begin{figure}[H]
\begin{centering}
\includegraphics[width=3cm,height=3cm,keepaspectratio]{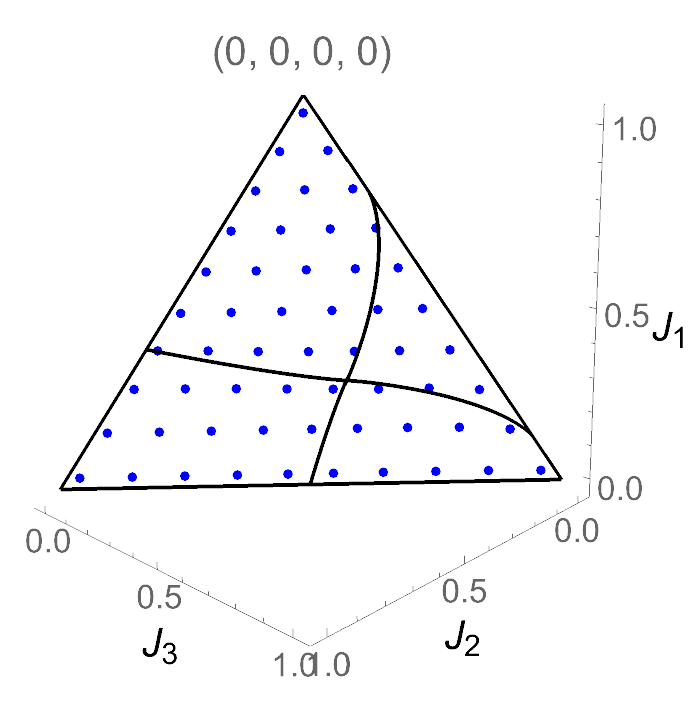}\includegraphics[width=3cm,height=3cm,keepaspectratio]{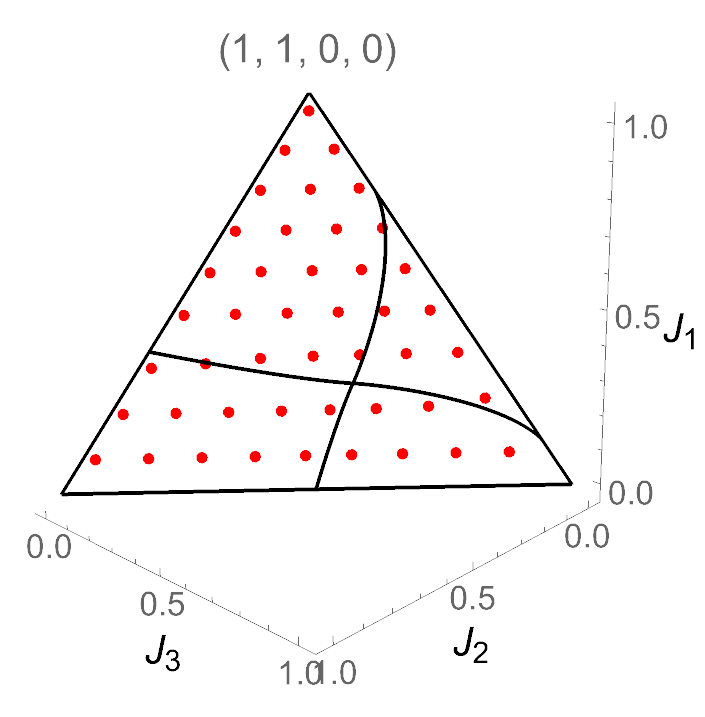}\includegraphics[width=3cm,height=3cm,keepaspectratio]{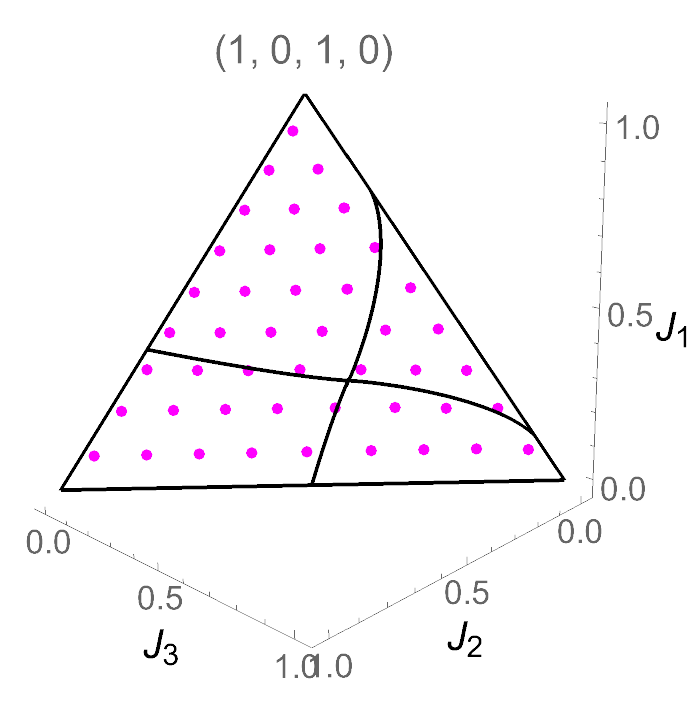}\includegraphics[width=3cm,height=3cm,keepaspectratio]{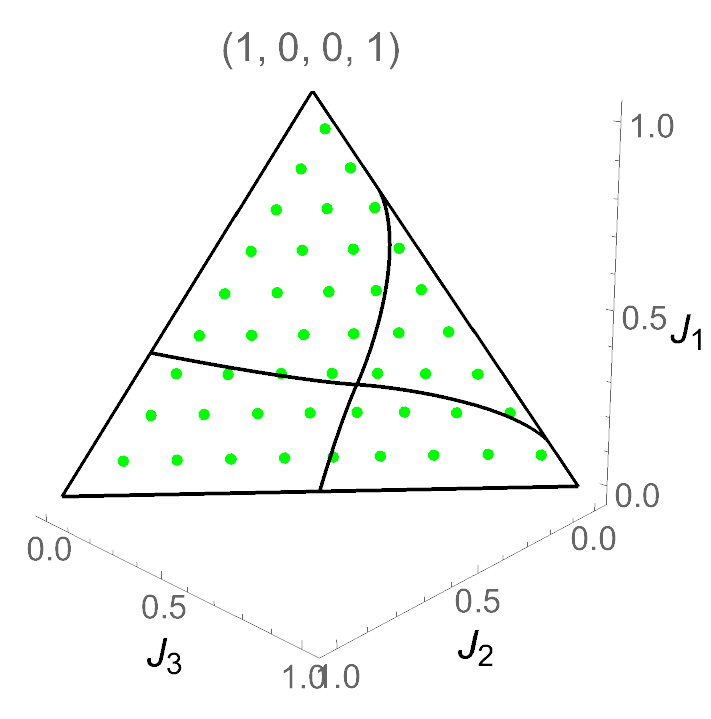}\\
\includegraphics[width=3cm,height=3cm,keepaspectratio]{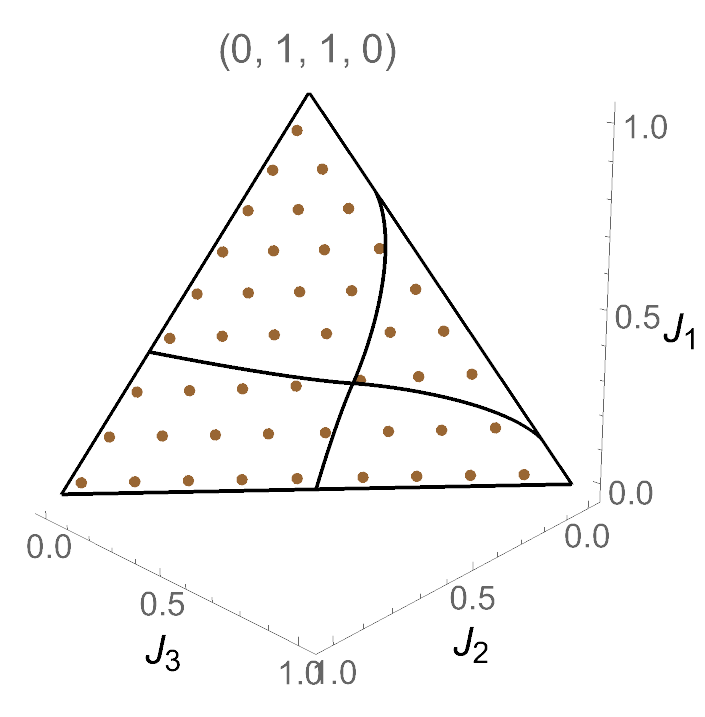}\includegraphics[width=3cm,height=3cm,keepaspectratio]{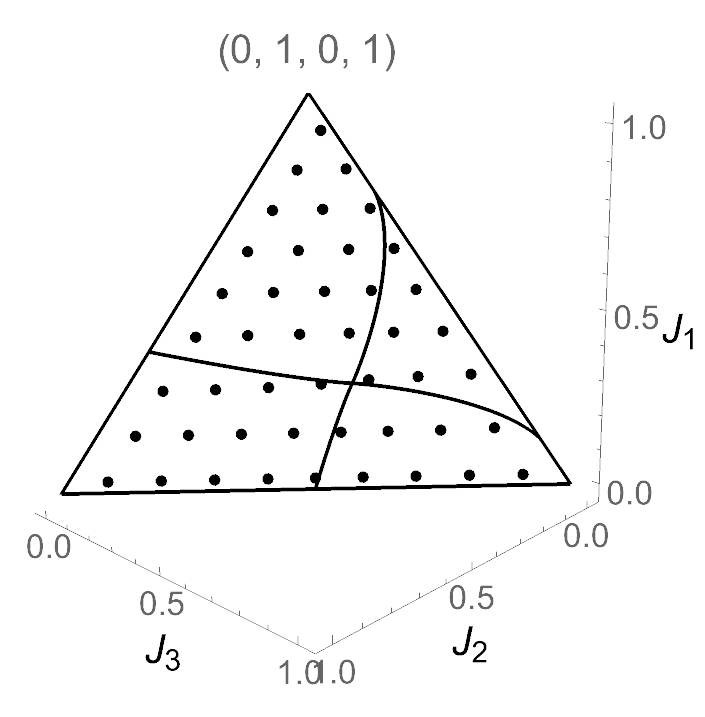}\includegraphics[width=3cm,height=3cm,keepaspectratio]{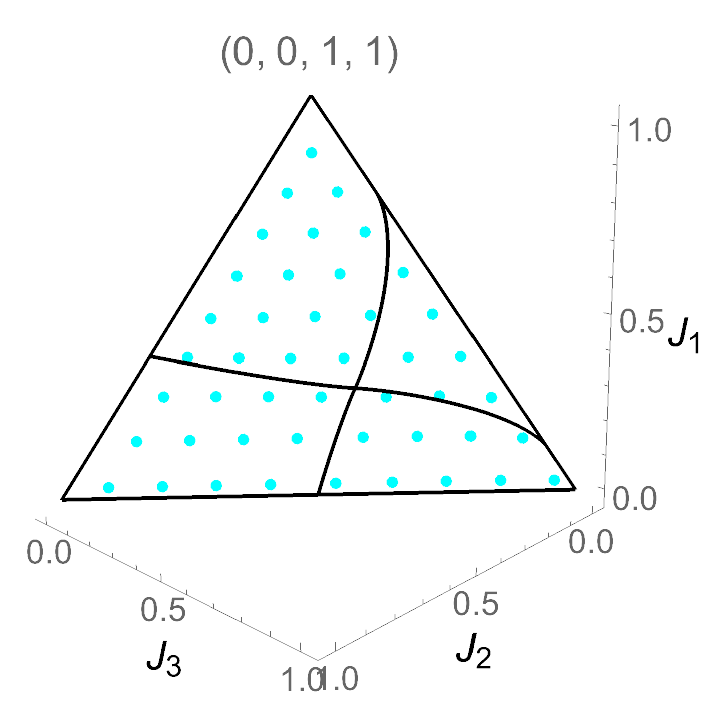}\includegraphics[width=3cm,height=3cm,keepaspectratio]{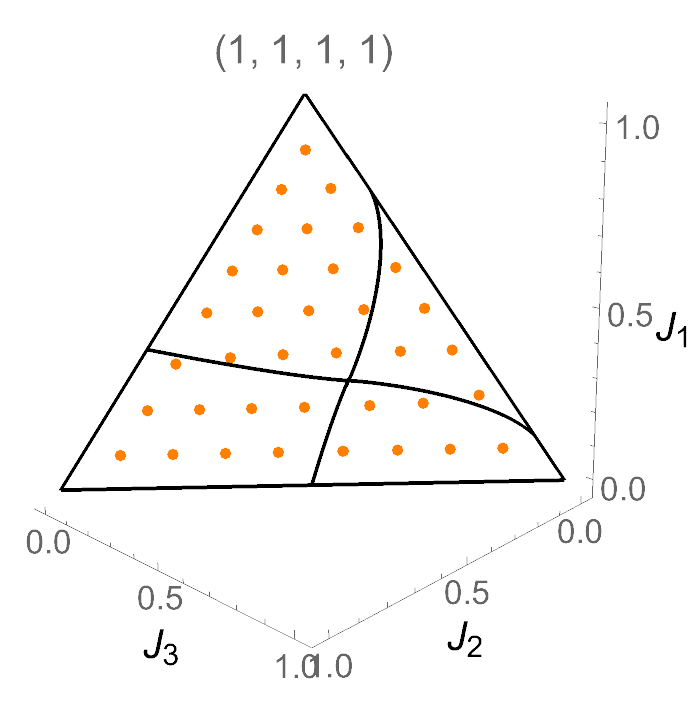}
\par\end{centering}
\caption{Actions for each symmetry class shown in Fig. \ref{fig:Fig1} a).}

\end{figure}
\begin{figure}[H]
\begin{centering}
\includegraphics[width=3cm,height=3cm,keepaspectratio]{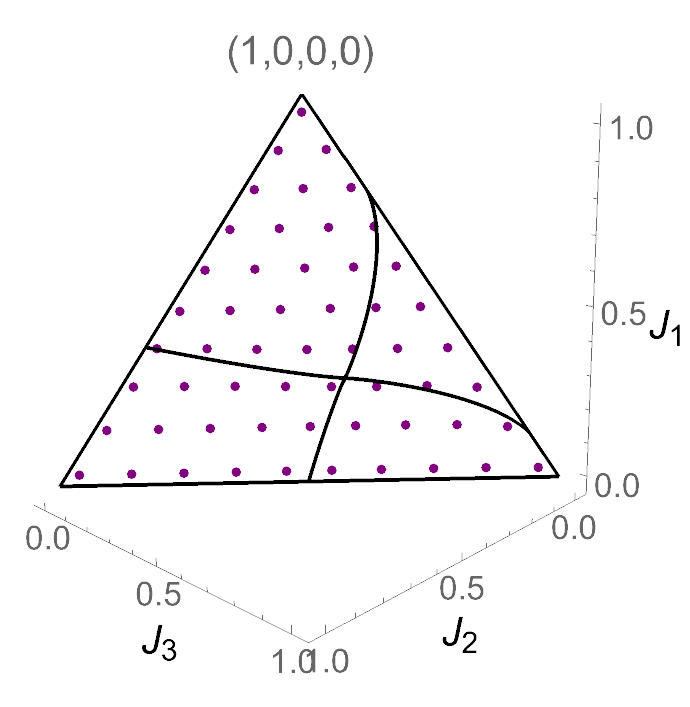}\includegraphics[width=3cm,height=3cm,keepaspectratio]{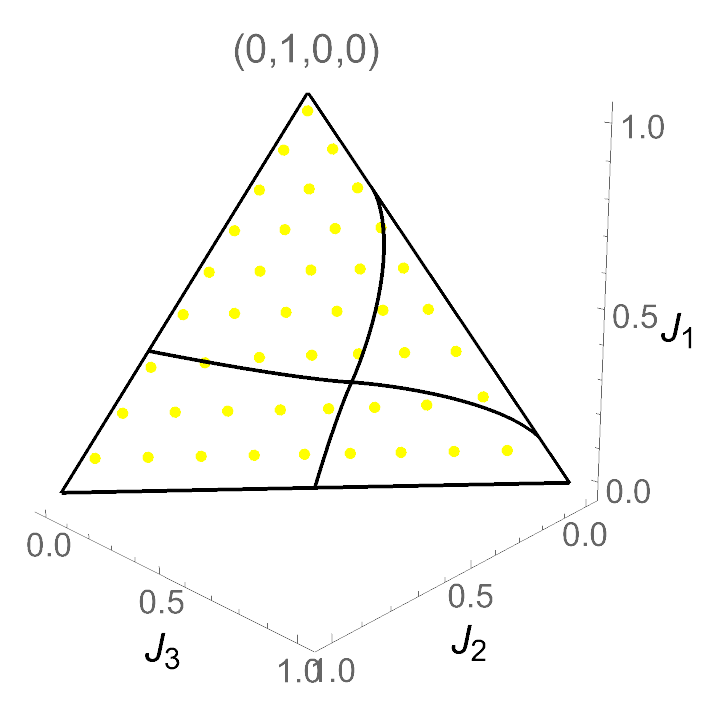}\includegraphics[width=3cm,height=3cm,keepaspectratio]{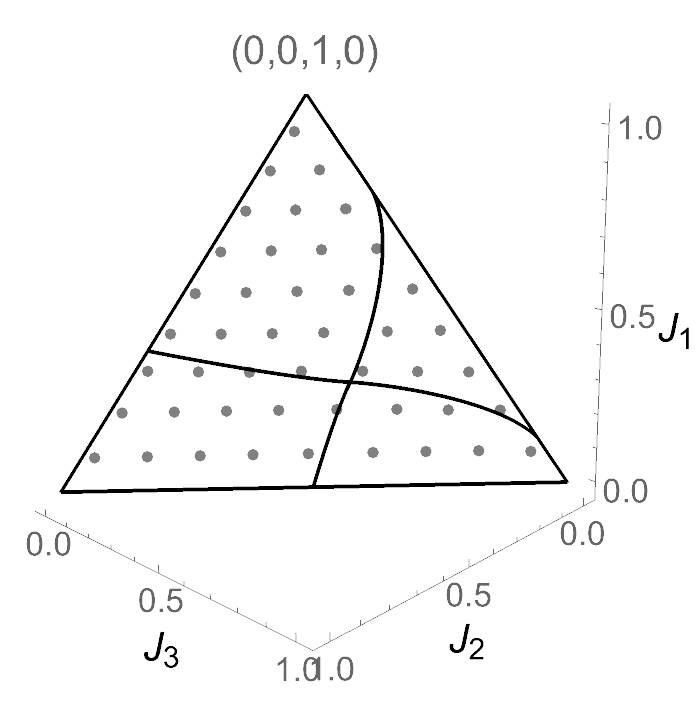}\includegraphics[width=3cm,height=3cm,keepaspectratio]{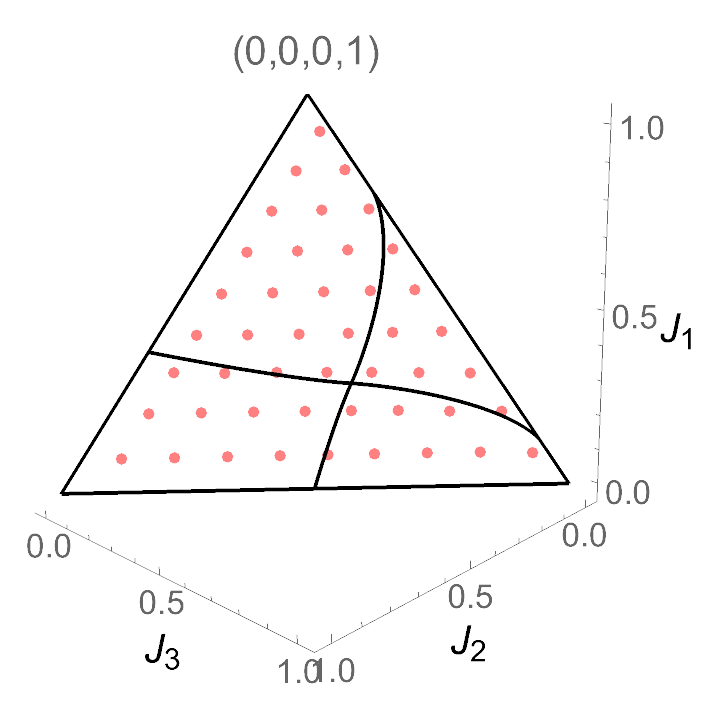}\\
\includegraphics[width=3cm,height=3cm,keepaspectratio]{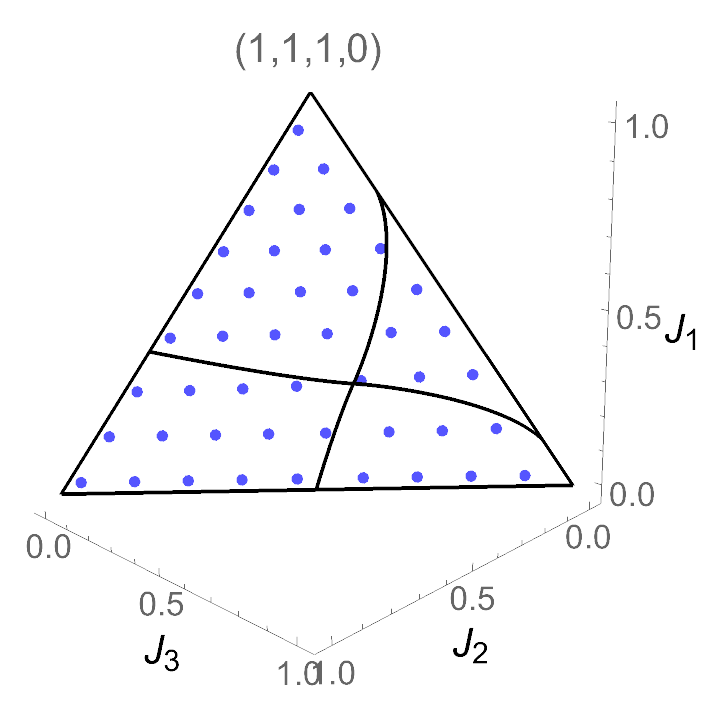}\includegraphics[width=3cm,height=3cm,keepaspectratio]{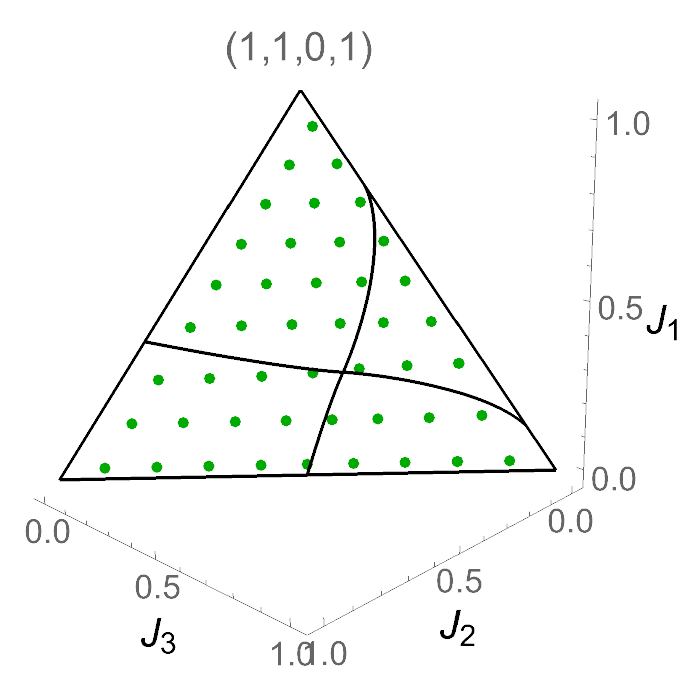}\includegraphics[width=3cm,height=3cm,keepaspectratio]{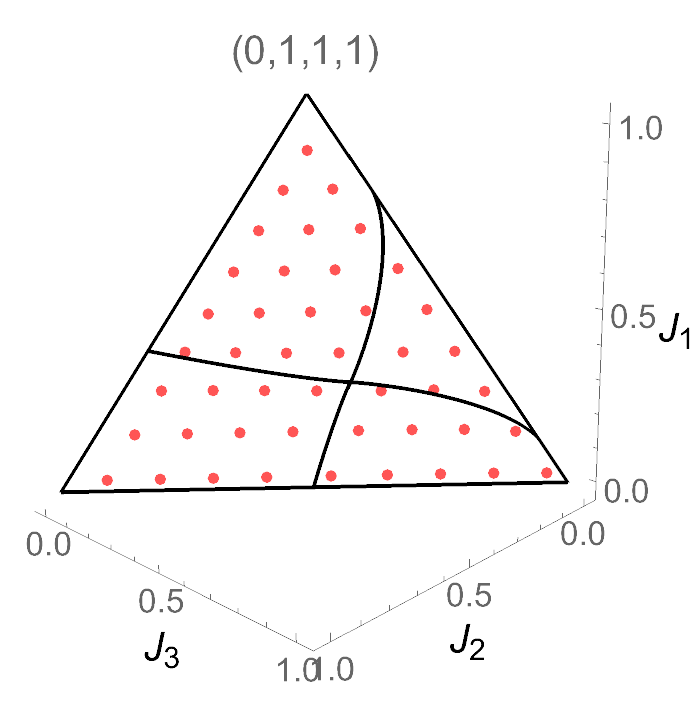}\includegraphics[width=3cm,height=3cm,keepaspectratio]{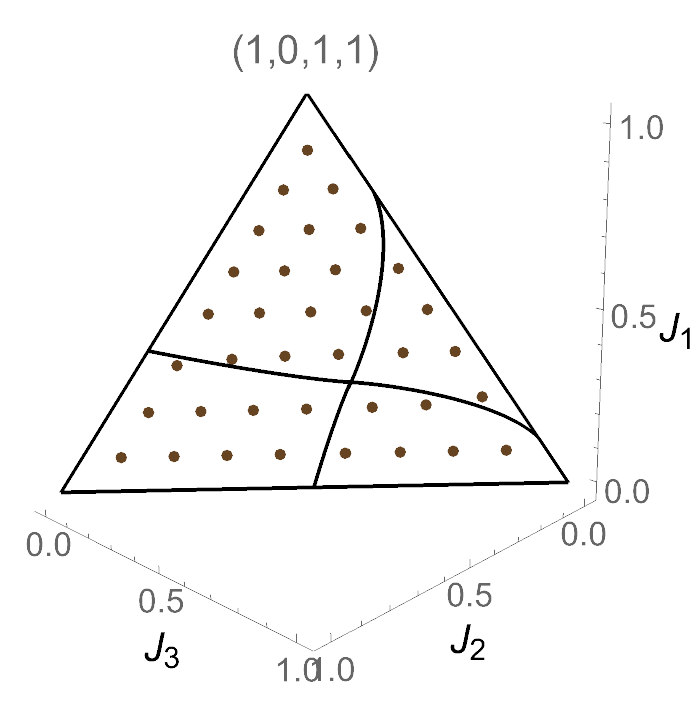}
\par\end{centering}
\caption{Actions for each symmetry class shown in Fig. \ref{fig:Fig1} b).}
\end{figure}

\bibliographystyle{abbrv}
\bibliography{ClassicalS3Ref.bib,all1.bib,hd1.bib,all.bib,hd.bib, all3.bib}
\end{document}